\documentclass[aps,10pt,twocolumn,showpacs,pra,citeautoscript,amsmath,amssymb,floatfix,nofootinbib,superscriptaddress,longbibliography]{revtex4-1}

\usepackage{amsmath}
\usepackage{dsfont}
\usepackage{amssymb}
\usepackage{bbm}
\usepackage{amsthm, thm-restate}
\usepackage{physics}
\usepackage{ulem}

\usepackage{graphicx}

\usepackage{xcolor, soul}

\usepackage{natbib}
\usepackage[linktocpage]{hyperref} 
\usepackage[capitalise]{cleveref}
\crefname{figure}{Figure}{figures}

\usepackage{algpseudocode}
\usepackage{enumerate}

\newcounter{lemmaN}

\newtheorem{lemma}[lemmaN]{Lemma}

\newcounter{lemmaA}

\newcounter{colN}

\newtheorem{corollary}[colN]{Corollary}
\newtheorem{definition}{Definition}
\newtheorem{theorem}{Theorem}

\newtheorem*{lemma*}{Lemma}
\newtheorem*{theorem*}{Theorem}

\newcommand{\id}{\mathbbm{1}}

\newcommand{\nats}{\mathbb{N}}
\newcommand{\ints}{\mathbb{Z}}
\newcommand{\reals}{\mathbb{R}}
\newcommand{\comp}{\mathbb{C}}

\newcommand{\mA}{\mathcal{A}}
\newcommand{\mE}{\mathcal{E}}
\newcommand{\mL}{\mathcal{L}}
\newcommand{\LH}{\mathcal{L}_{\textrm H}}
\newcommand{\LS}{\mathcal{L}_{\textrm S}}
\newcommand{\LSH}{\mathcal{L}_{\textrm {SH}}}

\newcommand{\SU}{{\rm SU}}
\newcommand{\SO}{{\rm SO}}
\newcommand{\GL}{{\rm GL}}
\newcommand{\U}{{\rm U}}
\newcommand{\PU}{{\rm PU}}

\newcommand{\mR}[0]{{\mathcal R}}
\newcommand{\mH}[0]{{\mathcal H}}
\newcommand{\mQ}[0]{{\mathcal Q}}

\newcommand{\mD}[0]{{\mathcal D}}
\newcommand{\tR}[0]{{\mathtt R}}
\newcommand{\tQ}[0]{{\mathtt Q}}

\newcommand{\conv}[0]{\mathrm{conv}}
\newcommand{\diag}[0]{\mathrm{diag}}

\newcommand{\Sym}[0]{\mathrm{Sym}}

\newcommand\scalemath[2]{\scalebox{#1}{\mbox{\ensuremath{\displaystyle #2}}}}

\makeatletter
\newcommand*{\balancecolsandclearpage}{%
  \close@column@grid
  \clearpage
  \twocolumngrid
}
\makeatother

\makeatletter
\def\tocdepth@fullmunge{%
\let\l@section@saved\l@section
\let\l@section\@gobble@tw@
\let\l@subsection@saved\l@subsection
\let\l@subsection\@gobble@tw@
}%
\def\tocdepth@fullrestore{%
\let\l@section\l@section@saved
\let\l@subsection\l@subsection@saved
}%
\newcommand{\hidetoc}[0]{\addtocontents{toc}{\string\tocdepth@fullmunge}}
\newcommand{\restoretoc}[0]{\addtocontents{toc}{\string\tocdepth@fullrestore}}
\makeatother

\begin{document}

\title{Spin-bounded correlations: rotation boxes within and beyond quantum theory}

\author{Albert\ Aloy}
\email{Albert.Aloy@oeaw.ac.at}
\affiliation{Institute for Quantum Optics and Quantum Information,
Austrian Academy of Sciences, Boltzmanngasse 3, A-1090 Vienna, Austria}
\affiliation{Vienna Center for Quantum Science and Technology (VCQ), Faculty of Physics, University of Vienna, Vienna, Austria}
\author{Thomas D.\ Galley}
\email{Thomas.Galley@oeaw.ac.at}
\affiliation{Institute for Quantum Optics and Quantum Information,
Austrian Academy of Sciences, Boltzmanngasse 3, A-1090 Vienna, Austria}
\affiliation{Vienna Center for Quantum Science and Technology (VCQ), Faculty of Physics, University of Vienna, Vienna, Austria}
\author{Caroline L.\ Jones}
\email{CarolineLouise.Jones@oeaw.ac.at}
\affiliation{Institute for Quantum Optics and Quantum Information,
Austrian Academy of Sciences, Boltzmanngasse 3, A-1090 Vienna, Austria}
\affiliation{Vienna Center for Quantum Science and Technology (VCQ), Faculty of Physics, University of Vienna, Vienna, Austria}
\author{Stefan L.\ Ludescher}
\email{Stefan.Ludescher@oeaw.ac.at}
\affiliation{Institute for Quantum Optics and Quantum Information,
Austrian Academy of Sciences, Boltzmanngasse 3, A-1090 Vienna, Austria}
\affiliation{Vienna Center for Quantum Science and Technology (VCQ), Faculty of Physics, University of Vienna, Vienna, Austria}
\author{Markus P.\ M\"uller}
\email{Markus.Mueller@oeaw.ac.at}
\affiliation{Institute for Quantum Optics and Quantum Information,
Austrian Academy of Sciences, Boltzmanngasse 3, A-1090 Vienna, Austria}
\affiliation{Vienna Center for Quantum Science and Technology (VCQ), Faculty of Physics, University of Vienna, Vienna, Austria}
\affiliation{Perimeter Institute for Theoretical Physics, 31 Caroline Street North, Waterloo, Ontario N2L 2Y5, Canada}

\date{November 15, 2024}

\begin{abstract}
How can detector click probabilities respond to spatial rotations around a fixed axis, in any possible physical theory? Here, we give a thorough mathematical analysis of this question in terms of ``rotation boxes'', which are analogous to the well-known notion of non-local boxes. We prove that quantum theory admits the most general rotational correlations for spins $0$, $1/2$, and $1$, but we describe a metrological game where beyond-quantum resources of spin $3/2$ outperform all quantum resources of the same spin. We prove a multitude of fundamental results about these correlations, including an exact convex characterization of the spin-$1$ correlations, a Tsirelson-type inequality for spins 3/2 and higher, and a proof that the general spin-$J$ correlations provide an efficient outer SDP approximation to the quantum set. Furthermore, we review and consolidate earlier results that hint at a wealth of applications of this formalism: a theory-agnostic semi-device-independent randomness generator, an exact characterization of the quantum $(2,2,2)$-Bell correlations in terms of local symmetries, and the derivation of multipartite Bell witnesses. Our results illuminate the foundational question of how space constrains the structure of quantum theory, they build a bridge between semi-device-independent quantum information and spacetime physics, and they demonstrate interesting relations to topics such as entanglement witnesses, spectrahedra, and orbitopes.
\end{abstract}

\maketitle

\tableofcontents

\section{Introduction}

Historically, quantum field theory has been developed by combining the principles of quantum theory with those of special relativity. This development has been a huge success: intersecting both theories turned out to be so constraining that it directly led to a host of novel physical predictions, such as the spin of particles and its relation to statistics, the creation and annihilation of particles, and phenomena such as Unruh radiation.

If, motivated by quantum information theory, we take an operational perspective on this development, then we can describe quantum field theory as the combination of two theories describing different phenomenological aspects of physics: our most successful theory for predicting the probabilities of events (quantum theory), and our most successful theory for describing space and time (special or general relativity). Probabilities have to interplay consistently with spacetime to yield a successful predictive theory.

While it has long been understood that special relativity describes just one possible spacetime geometry among many others, the intuition until recently has been that quantum theory is essentially our only possible choice for describing probabilities of events, except for classical probability theory. Thus, quantum field theory is defined entirely in terms of operator algebras, encompassing both classical and quantum probability theory and their hybrids, and only those.

However, motivated again by quantum information theory and by quantum foundations research, recent years have seen a surge of interest in probabilistic theories that are neither classical nor quantum. One particularly successful direction has been the \textit{device-independent (DI) framework}~\cite{mayers1998quantum,barrett2005no,acin2007device,gallego2010device,brunner2014bell,scarani2019bell} for describing quantum information protocols. The main idea is to certify the security of one's protocols (such as quantum key distribution or randomness generation) by a few simple physical principles only. No assumptions or (in the \textit{semi}-DI framework~\cite{pawlowski2011semi,liang2011semi,branciard2012one,VanHimbeeck2017}) only very mild ones are made on the inner workings of the devices, and the security of the protocol follows from the observed statistics and plausible assumptions such as the no-signalling principle alone.

In this paper, we explore the foundations for studying the interplay of spacetime symmetries with the probabilities of events without assuming the validity of quantum theory. Assuming special relativity, physical systems must react to symmetry transformations (in general, Poincar\'e transformations) in a consistent way: the symmetry group must act continuously on its state space while preserving its structure. In quantum theory, this means that systems must carry projective representations of this group. Here, we consider more general black boxes (which need not be quantum) yielding statistics which responds to such transformations. Instead of the full Poincar\'e group, we study the action of one of its simplest nontrivial subgroups: the group of spatial rotations around a fixed axis, ${\rm SO}(2)$. In an abstract DI language, we study black boxes whose input is  given by a spatial rotation around a fixed axis, and which produce one of a finite number of outputs. This specializes, but also greatly extends  the framework introduced in~\cite{Garner}.

In particular, we consider such ``rotation boxes'' under the semi-DI assumption that their ``spin'', i.e.\ representation label of ${\rm SO}(2)$ on the ensemble of boxes, is upper-bounded by some value $J$. We obtain surprising insights into the structure and possible behavior of such boxes, showing, for example, that for $J=0$, $J=1/2$, and $J=1$, quantum theory describes the most general ways in which \textit{any} theory could respond to spatial rotations, but that for $J\geq 3/2$, correlations exist which cannot be generated by quantum theory with the same $J$. We give a Tsirelson-type inequality~\cite{Tsirelson1980} delineating the quantum correlations from more general ones, and describe a metrological task~\cite{Giovannetti2011,Toth} where post-quantum spin-$3/2$ systems can outperform all quantum ones. Moreover, rotation boxes can be wired together in Bell experiments, and we review and reinterpret existing work showing that our semi-DI assumption on the maximal spin can be used to certify Bell nonlocality with fewer measurements than otherwise possible, as well as to characterize the quantum-$(2,2,2)$ Bell correlations exactly within the set of non-signalling correlations.

Our motivation  for studying such boxes and their generalizations is threefold:
\begin{itemize}
	\item[1.] \textbf{Studying how spacetime structure constrains the structure of quantum theory (QT).} If we assume that a probabilistic theory ``fits into space and time'', does this already imply important structural features of QT? Can we perhaps \textit{derive} QT from this desideratum? Or how much wiggle room is there in spacetime for probabilistic theories that go beyond quantum theory? A version of this question has been posed and studied for correlations generated by space-like separated parties, where the set of quantum correlations is known to be a strict subset of the general set of no-signalling correlations~\cite{Tsirelson1980,tsirelsonNS,PRbox,Popescu}. We formulate and solve an analogous question: how can we characterize the set of quantum spin-$J$ correlations in the space of general spin-$J$ correlations?
	\item[2.] \textbf{Novel theory-independent and physically better motivated semi-DI protocols.} Assumptions on the response of physical systems to spacetime symmetries can be used directly in semi-DI protocols for certification. In particular, such assumptions are sometimes physically simpler or more meaningful (corresponding to e.g.\ energy or particle number bounds~\cite{VanHimbeeck2017,VanHimbeeck2019}) than abstract assumptions often made in the field, such as upper bounds on the Hilbert space dimension of the physical system.	 For example, in~\cite{Jones}, some of us have constructed a semi-DI protocol for the generation of random numbers whose security relies on an upper bound of the system's spin, without assuming the validity of quantum theory.
	\item[3.] \textbf{The study of resource-bounded correlations.} What we study in the ${\rm SO}(2)$-case in this paper is a special case of analyzing resource-bounded correlations: given some spacetime symmetry, and an upper bound on the symmetry-breaking resources, determine the resulting correlations that quantum theory (or a more general theory) admits. The paradigmatic example is the study of quantum speed limits~\cite{MandelstamTamm,Anandan,Hoernedal,MarvianSpekkensZanardi}: upper-bounding the (expectation value or variance of the) energy constrains how quickly quantum states can become orthogonal. Replacing time-translation symmetry by rotational symmetry leads to the formalism of this paper.
\end{itemize}
\textbf{Our article is organized as follows.}
In \Cref{sec:task}, we consider a metrological game to illustrate a gap between the predictions of quantum theory and those of hypothetical, more general theories consistent with rotational symmetry. In \Cref{sec:rot_box_frame}, we introduce the conceptual framework and discuss the background assumptions of rotation boxes. More specifically, in Subsection~\ref{subsec:qspinJ}, we define and analyze the structure of the sets of quantum correlations, when the spin is constrained. In Subsection~\ref{subsec:genspinJ}, we do so for the corresponding sets of general ``rotational correlations'', when boxes are characterized \textit{only} by their response to rotations (but need not necessarily be quantum). 
In Subsection~\ref{SubsecRelaxation}, we discuss how, although defined independently, the rotation set can be interpreted as a relaxation of the quantum set of correlations, and show how this leads to an efficient semidefinite programming (SDP) characterization.

Next, in \Cref{secMainResults}, we outline our main results, which concern rotation boxes in prepare-and-measure scenarios, and the relation between the quantum and general sets. In Subsection~\ref{Subsec012}, we start by analyzing the scenario for the cases $J\in\{0,1/2\}$, for which we show that every rotation box correlation can be generated by a quantum system of the same $J$. In Subsection~\ref{sec:R1}, we consider the $J=1$ case, and show the equivalence of the rotation and quantum sets of correlations specifically for 2 outputs, based on an exact convex characterization of this set. 
In Subsection~\ref{SubsecGap}, we demonstrate that a gap between the sets appears for $J\geq 3/2$. We construct a Tsirelson-like inequality for $J=3/2$ and provide an explicit correlation of rotation box form that violates the quantum bound. Using the same methodology, we further show that the gap exists for all finite $J\geq3/2$. In Subsection~\ref{sec:Jinftyquantum}, we examine the case where $J$ is unconstrained (i.e.\ $J\rightarrow\infty$), in which every rotation correlation can be approximated arbitrarily well by finite-$J$ quantum systems. In Subsection~\ref{subsec:randomness}, we then review our previous results~\cite{Jones}, concerning two input rotation boxes, in which we have applied the framework to describe a theory-independent protocol for randomness generation. Finally, in Subsection~\ref{subsec:classical}, we address how one should understand a ``classical'' rotation box. 

In \Cref{sec:rotationbell}, we consolidate earlier results concerning Bell setups using our framework. First, in Subsection~\ref{subsec:two parties}, we review and shed some new light on the results of~\cite{Garner}, which yield an exact characterization of the $(2,2,2)$-quantum Bell correlations; 
second, in Subsection~\ref{subsec:many parties}, we clarify the additional assumption of~\cite{Nagata} allowing for indirect witnesses of multipartite Bell nonlocality. 
Next, in \Cref{sec:connectionothertopics}, we outline connections to other known results. In particular, in Subsection~\ref{subsec:almostq}, we discuss the conceptual similarity to ``almost quantum'' Bell correlations~\cite{AlmostQuantum} in more depth; in Subsection~\ref{subsec:orbitope}, we show that the state spaces of rotation boxes are isomorphic to Carath{\'e}odory orbitopes~\cite{Sanyal_2011}; and in Subsection~\ref{subsec:symmetric}, we make a connection between the effect space of the rotation GPT system and a family of rebit entanglement witnesses. Finally, we conclude in \Cref{sec:conc}.

\Cref{tab:notation} gives a brief overview on our notation.

\begin{table}
    \centering
    \begin{tabular}{ll} 
        $\mL(V)$ & Space of linear operators on the vector space $V$ \\
        $\LH(\comp^n)$ & Space of Hermitian operators on $\comp^n$  \\
        $\LS(\reals^n)$ & Space of symmetric operators on $\reals^n$\\
        $\mD(\mathcal{H})$ & Set of density operators on Hilbert space $\mathcal{H}$ \\
        $\mathcal{E}(\mathcal{H})$ & Set of POVM elements on $\mathcal{H}$ \\
        $\LSH(\comp^n)$ & Space of symmetric Hermitian operators on $\comp^n$  \\
        $\Sym^d(V)$ & Symmetric subspace of $V^{\otimes d}$ \\
        $\mathbb{N}$ & Natural numbers $\{1,2,3,4,\ldots\}$\\
        $\mathbb{N}_0$ & Non-negative integers $\{0,1,2,3,4\ldots\}$
    \end{tabular}
    \caption{Notation used in the paper.}
    \label{tab:notation}
\end{table}

\section{Invitation: a spin-bounded metrological task}\label{sec:task}

Consider the following situation, which resembles a typical scenario in quantum metrology. A referee promises to perform a spatial rotation by some angle $\theta$. Before this, we may prepare a physical system in some state, submit it to the rotation, and subsequently measure it to estimate $\theta$. How well can we do this?

\begin{figure}[hbt]
\centering 
\includegraphics[width=\columnwidth]{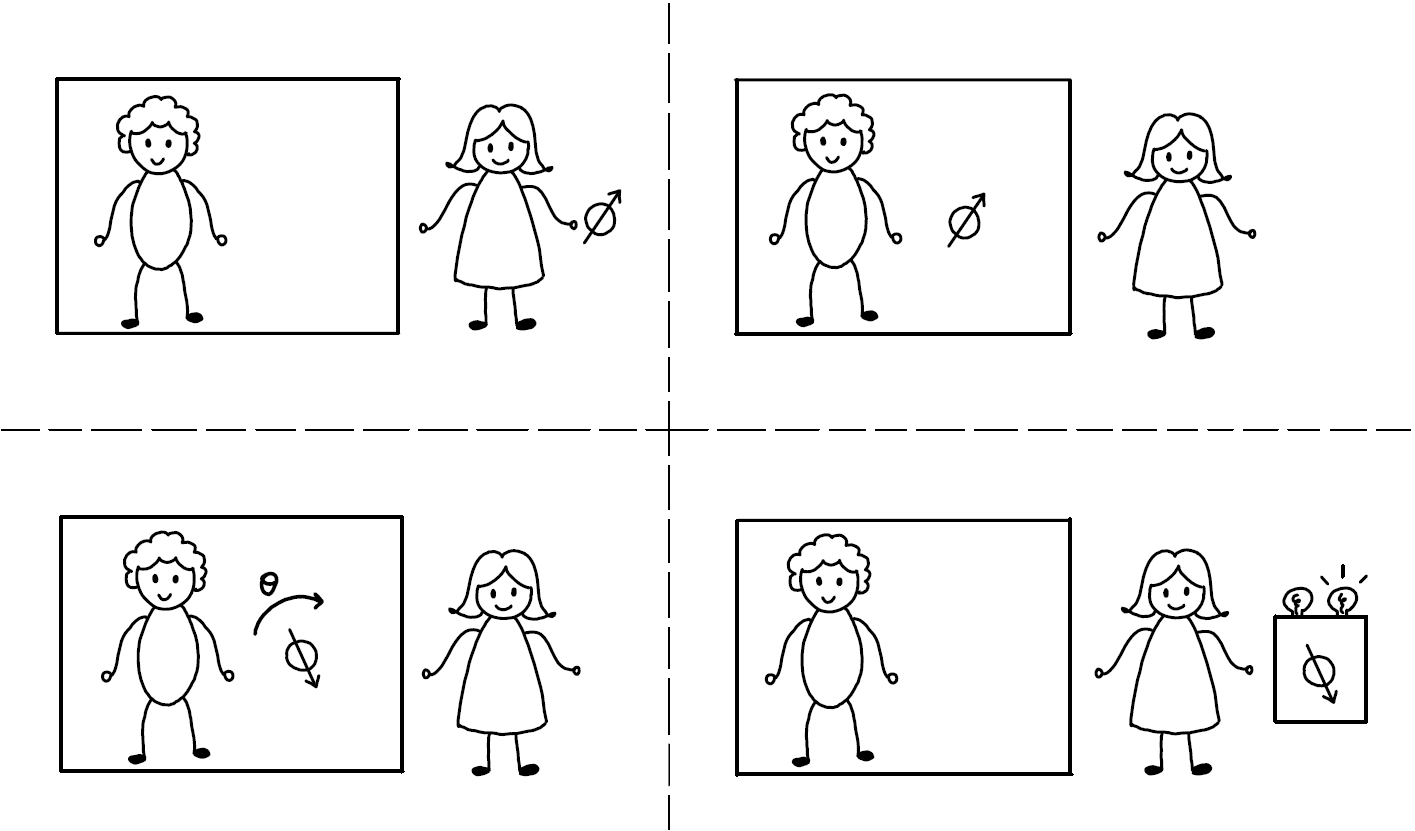}
\caption{Schematic sketch of the metrological task. An agent holds a physical system of spin $J=3/2$, in an initial state $\rho$. She gives it to a referee, who, in a black box with respect to the agent, performs some spatial rotation of angle $\theta$ on the system, where $\theta$ is chosen according to the distribution function $\mu(\theta)$ (defined in the main text and shown in \Cref{fig:gameFunction}). The referee then passes the system back to the agent, who measures it using a two-outcome box in order to determine whether the angle $\theta$ is in the range $R_+$ or $R_-$ (see also \Cref{fig:gameFunction}).}
\label{fig:metrologyfig}
\end{figure}

If our physical system is a classical gyroscope, we can certainly determine $\theta$ perfectly --- the challenge lies in the use of  \textit{microscopic} systems. Think of the system as carrying some intrinsic \textit{spin} $J$, an integer or half-integer, that responds to rotations. Classical systems correspond to the case of $J\to\infty$, supported on an infinite-dimensional Hilbert space with narrowly peaked coherent states, allowing us to resolve the rotation arbitrarily well. Hence, consider a more interesting case: we demand that the  system is a quantum  spin-$J$ system, where $J$ is small. Concretely, let us choose $J=3/2$ (the smallest interesting $J$ for this task, as we will see in subsequent sections). That is, we regard the total spin, as represented by the spin quantum number, as a resource, and are constrained in our access to such resources.

Moreover, suppose that our task is \textit{not} to estimate $\theta$ directly. Instead, our task is to guess whether $\theta$ is in region $R_+$ or in region $R_-$, as depicted in Figure~\ref{fig:gameFunction}, corresponding to the sets of angles where the function $\cos(2\theta)+\sin(3\theta)$ is either positive or negative. That is, our guess will be a single bit, $+$ or $-$, and we would like to maximize our probability that this bit equals the sign of $\cos(2\theta)+\sin(3\theta)$.

\begin{figure}[h]
\centering 
\includegraphics[width=\columnwidth]{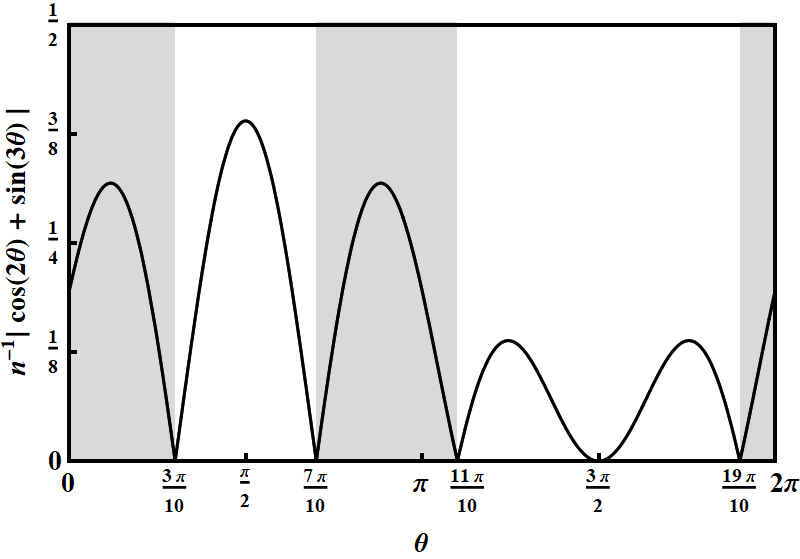}
\caption{The task is to estimate whether $\theta$ is in the range $R_+$ (gray) or in the range $R_-$ (white). These ranges are defined according to where the function $\cos(2\theta)+\sin(3\theta)$ is either positive or negative. Here we plot its normalized absolute value, which is the probability density that our referee uses to draw the angle $\theta$ in our metrological game. The ranges correspond to $R_+ =(0,3\pi/10)\cup (7\pi/10,11\pi/10) \cup (19/10\pi,2\pi)$, $R_-$ is the complement $R_- = (3\pi/10,7\pi/10)\cup(11/10\pi,19/10\pi)$.}
\label{fig:gameFunction}
\end{figure}

Let us summarize the task (also sketched in Figure \ref{fig:metrologyfig}) and specify it some more. First, the referee picks an angle $\theta$, but not uniformly in the interval $[0,2\pi)$, but according to the distribution function $\mu(\theta):=n^{-1}|\cos(2\theta)+\sin(3\theta)|$, where $n$ is a constant such that $\int_0^{2\pi}\mu(\theta)d\theta=1$ (it turns out that $n=\frac 5 3 \sqrt{5+2\sqrt{5}}$). Then, we prepare a spin-$3/2$ system in some state $\rho$ and send it to the referee, who subsequently applies a rotation by angle $\theta$ to it. Finally, we retrieve the system and measure it with a two-outcome POVM $(E_+,E_-)$. Our task is to produce outcome $+$ if the angle was chosen from $R_+$, and outcome $-$ if the angle was chosen from $R_-$.

This may not be the most obviously relevant task to consider, but it will serve its purpose to demonstrate an in-principle gap between quantum and beyond-quantum resources for metrology.

It turns out that the two events $+$ and $-$ both have probability $1/2$, since
\[
\int_{R_+} \mu(\theta)d\theta=\int_{R_-}\mu(\theta)d\theta=\frac 1 2.
\]
But our goal is to improve upon random guessing by preparing and measuring a quantum system used for sensing in the optimal way. By the Born rule, the conditional probability of our measurement outcome is
\begin{eqnarray}
P(\pm|\theta)&=&   \Tr(e^{i\theta Z} \rho e^{-i\theta Z} E_\pm)\nonumber\\
&=& c_0^\pm + c_1^\pm\cos\theta + s_1^\pm\sin\theta + c_2^\pm\cos(2\theta)\nonumber\\
&& +s_2^\pm \sin(2\theta)+ c_3^\pm\cos(3\theta)+s_3^\pm\sin(3\theta),
\label{EqTrigPoly3}
\end{eqnarray}
where $\rho$ is some quantum state, $Z={\rm diag}(3/2,1/2,-1/2,-3/2)$ is the spin-$3/2$ representation of the generator of a rotation around a fixed axis, and $E_\pm \geq 0$, $E_+ + E_-=\mathbf{1}$ is a measurement operator. The coefficients $c_i^\pm,s_i^\pm$ can be determined from the state and measurement operator. The set of all such probability functions will be called the \textit{quantum spin-$3/2$ correlations}, $\mathcal{Q}_{3/2}$. In fact, our construction will be more general than this: we will \textit{not} define spin-$J$ correlations as those that can be realized on the $(2J+1)$-dimensional irreducible representation, but on \textit{any} quantum system where all outcome probabilities are trigonometric polynomials of degree at most $2J$. That these correlations can always be realized on $\mathbb{C}^{2J+1}$ is a non-trivial fact which we are going to prove.

The success probability becomes
\begin{eqnarray*}
P_{\rm succ}&=&\int_{R_+} P(+|\theta)\mu(\theta)d\theta + \int_{R_-} P(-|\theta)\mu(\theta)d\theta\\
&=&\int_{R_+} P(+|\theta)\mu(\theta)d\theta + \frac 1 2 -\int_{R_-} P(+|\theta)\mu(\theta)d\theta\\
&=& \int_0^{2\pi} P(+|\theta)n^{-1}\left(\strut \cos(2\theta)+\sin(3\theta)\right)d\theta+\frac 1 2\\
&=&\frac\pi n(c_2^++s_3^+)+\frac 1 2,
\end{eqnarray*}
where we have used that, by definition, $|f(\theta)|=\pm f(\theta)$ for $\theta\in R_\pm$, where $f(\theta)=\cos(2\theta)+\sin(3\theta)$.
To compute the maximum success probability $P_{\rm succ}^{\rm Q}$ over all spin-$3/2$ quantum systems, we have to determine the maximum value of $c_2+s_3$ on all quantum spin-$3/2$ correlations. We will do this in Subsection~\ref{SubsecGap}, showing in Theorem~\ref{thm:spinThreeHalvesInequality} that this maximum equals $1/\sqrt{3}$. Thus
\begin{eqnarray*}
P_{\rm succ}^{\rm Q}&=&\max_{P\in\mathcal{Q}_{3/2}}\frac \pi n (c_2^+ + s_3^+)+\frac 1 2=\frac 1 2+\frac{3\pi}{5\sqrt{3(5+2\sqrt{5})}}\\
&\approx& 0.8536.
\end{eqnarray*}
Note that we do not allow the system to start out entangled with another system that is involved in the task. In particular, we are not considering the situation that we keep half of an entangled state and send the other half to the referee that performs the rotation. We leave an analysis of this more general situation for future work.

Now suppose that we drop the assumption that quantum theory applies to the scenario. What if we use a spin-$3/2$ system for sensing that is not described by quantum physics? In the following sections, we will discuss in detail how such generalized ``rotation boxes'' can be understood, by considering arbitrary state spaces on which ${\rm SO}(2)$ acts. In summary, a generalized spin-$3/2$ correlation (an element of what we denote by $\mathcal{R}_{3/2}$) will be any probability function $P(\pm|\theta)$ that is a trigonometric polynomial of degree three (as the second line of Eq.~(\ref{EqTrigPoly3})), but without assuming that it comes from a quantum state and measurement (as in the first line of Eq.~(\ref{EqTrigPoly3})).

It turns out that $c_2+s_3$ can take larger values for such more general spin$-3/2$ correlations, and we give an example in Theorem~\ref{thm:spinThreeHalvesInequality}. The maximum value turns out to be $5/8$. Thus, when allowing more general spin-$3/2$ rotation boxes, the maximal success probability is
\begin{eqnarray*}
P_{\rm succ}^{\rm R}&=&\max_{P\in\mathcal{R}_{3/2}}\frac \pi n (c_2^+ + s_3^+)+\frac 1 2=\frac 1 2+\frac{3\pi}{8\sqrt{5+2\sqrt{5}}}\\
&\approx& 0.8828.
\end{eqnarray*}
Hence, general rotation boxes allow us to succeed in this metrological task with about $3\%$ higher probability.

From a foundational point of view, tasks like the above can be used to analyze the interplay of quantum theory with spacetime structure. For example, we will see that for spins $J=0, 1/2, 1$, a gap like the above does not appear, and quantum theory is thus optimal for metrological tasks like the above. From a more practical perspective, the correlation sets $\mathcal{R}_J$ are outer approximation to the quantum sets $\mathcal{Q}_J$ which have characterizations in terms of semidefinite program constraints (in mathematics terminology, the $\mathcal{R}_J$ are projected spectrahedra). This allows us to optimize linear functionals (such as the quantity $c_2+s_3$ above) over $\mathcal{R}_J$ in a computationally efficient way, yielding useful bounds on the possible quantum correlations that are achievable in these scenarios. We will see that general spin-$J$ correlations stand to quantum spin-$J$ correlations in a similar relation as ``almost quantum'' Bell correlations stand to quantum Bell correlations~\cite{AlmostQuantum}.

In the following section, we will introduce the notions of rotation boxes and spin-$J$ correlation functions in a conceptually and mathematically rigorous way, corroborating the above analysis.

\section{Rotation boxes framework}\label{sec:rot_box_frame}

In DI approaches, one often considers quantum networks (such as Bell scenarios) where several \textit{black boxes} are wired together. As sketched in Figure~\ref{fig:rotationboxes}a, a black box of this kind is typically thought of accepting an abstract input $x$ (for example, a bit, $x\in\{0,1\}$) and yielding an abstract output (for example, $a\in\{-1,+1\}$). In QT, this could describe a measurement, where $x$ denotes the choice of measurement and $a$ its outcome.

\begin{figure}[t]
\centering 
\includegraphics[width=0.9\columnwidth]{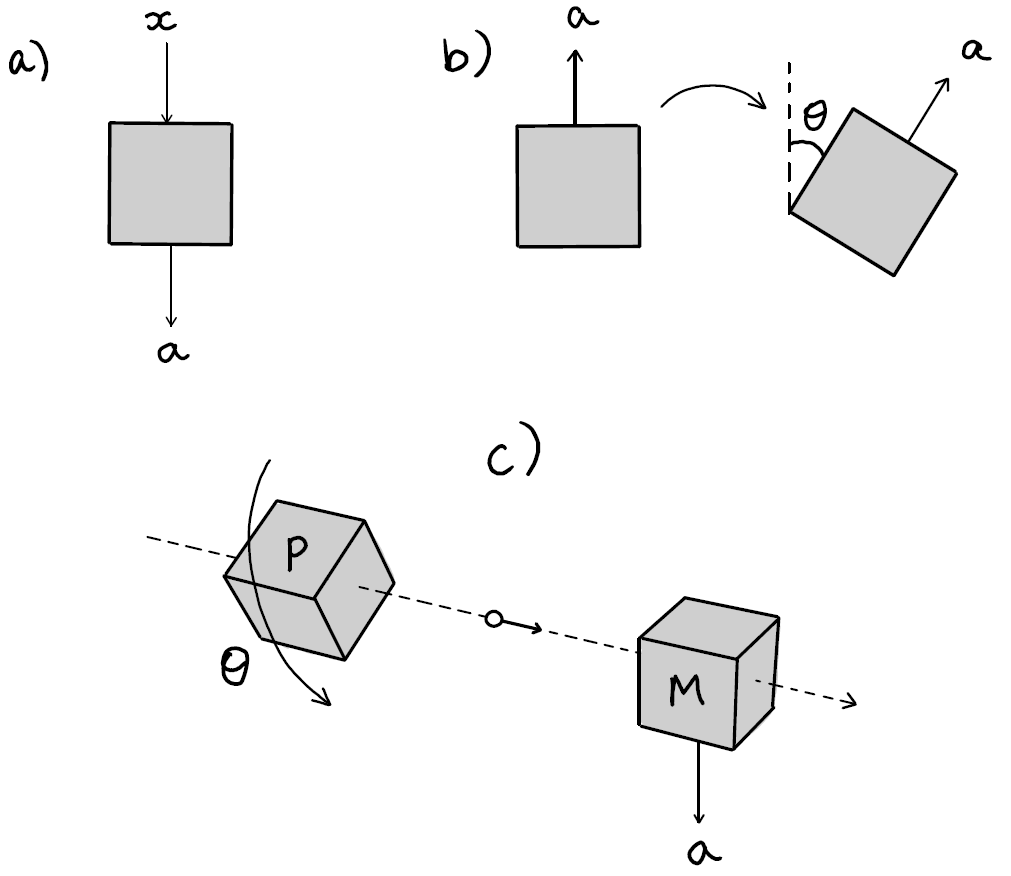}
\caption{Boxes, rotation boxes, and the different ways to think about their physical realization. See the main text for details.}
\label{fig:rotationboxes}
\end{figure}

In this paper, we consider boxes whose input is given by a spatial rotation around a fixed axis. The input is therefore an angle $0\leq\theta<2\pi$. However, we do not just aim at describing generic boxes that accept continuous inputs. The intuition is not that we input a \textit{classical description} of $\theta$ into the box (say, written on a piece of paper or typed on a keyboard), but rather that we \textit{physically rotate the box in space} (Figure~\ref{fig:rotationboxes}b). That is, we assume that we have a notion of a \textit{physical rotation} that we can apply to the box, and that this notion is a clear primitive of spatiotemporal physics. This is comparable to a Bell experiment, where we believe that we understand, in a theory-independent way, what it means to ``spatially separate two boxes'' (say, to transport one of them far away), such that the assumption that no information can travel faster than light enforces the no-signalling condition.

To unpack this idea further, we have to be more specific. A more detailed way to describe black boxes is in terms of a \textit{prepare-and-measure scenario}: we have a \textit{preparation device} which generates a physical system in some state, and a \textit{measurement device} that subsequently receives the physical system and generates a classical outcome. The input $x$ is thought of being supplied to the preparation device such that the resulting state can depend on $x$. Here, instead, we think of a physical operation being applied to the preparation device:

\textit{The input to the rotation box consists of rotating the preparation device by angle $\theta$ around a fixed axis, relative to the measurement device,} see \Cref{fig:rotationboxes}c.

Assuming that physics is covariant under rotations about this fixed axis leads to a representation of the SO(2) group on the state space. To see this, we follow similar argumentation to that of~\cite[Chapter 13]{Wald}. First, consider an observer $O$ equipped with a coordinate system and holding a $k$-outcome measurement device, which measures the state $\omega\in\Omega$ transmitted by the preparation device (which need not necessarily be described by quantum theory). This produces probability tables, which can be characterized by a function $P_O:[0,2\pi)\times\Omega\rightarrow[0,1]^k$, such that every pair of angles and states are mapped to valid probability vectors. We assume that the outcome statistics uniquely characterize the state $\omega$, and that $\Omega$ is finite-dimensional. Next, consider a different observer $O'$, with their own coordinate system and $k$-outcome measurement device, related to $O$ by a rotation $\phi$ of angle $\alpha$ around the fixed axis on which the input angle is defined. This reorientates the coordinate system, which induces a map $\hat{\phi}:[0,2\pi)\rightarrow[0,2\pi)$ on the set of inputs, defined by $\hat{\phi}(\theta):=\theta-\alpha$, i.e. relating the input angles of $O$ to the input angles of $O'$. According to rotational covariance, this is equivalent to a situation in which the observer $O$ is unchanged but a state $\omega'\in\Omega$ exists such that 
\begin{equation}\label{rotinvprobs}
    P_{O}(\theta,\omega')=P_{O'}(\hat{\phi}(\theta),\omega).
\end{equation} 
That is to say, there are no probabilities that could be observed in one frame that could not be observed in another (i.e there are no distinguished frames).
Finally, from Equation~(\ref{rotinvprobs}), a map $\bar{\phi}:\Omega\rightarrow\Omega$ can be defined, as $\bar{\phi}(\omega):=\omega'$. Now we consider all possible rotations around the fixed axis. This collection of rotations $\phi$ relating observers is isomorphic to the group SO(2), hence we label them ${\phi}_\alpha$, where $\alpha$ is the angle of the corresponding SO(2) rotation. From Equation~(\ref{rotinvprobs}), it follows that
\begin{equation}
\bar{\phi}_{\alpha_1}\circ\bar{\phi}_{\alpha_2}=\bar{\phi}_{\alpha_1+\alpha_2}.
\end{equation}
Statistical mixing of preparation procedures should be conserved under rotations, therefore every $\bar{\phi}_\alpha$ must be linear (for further details, see \Cref{subsec:genspinJ}). Therefore, these maps $\{\bar{\phi}_\alpha\}_\alpha$ define a group representation.

Our mathematical formalism below will not depend on this specific interpretation of the ${\rm SO}(2)$-element as a \textit{spatial rotation}: it will also apply to situations where this group action has a different physical interpretation, for example as some periodic time evolution, or as some abstract transformation without any spacetime interpretation whatsoever. However, the specific scenario of preparation procedures that can be physically rotated in space gives us the clearest and perhaps most theory-independent motivation for believing that our formalism applies to the given situation. This is comparable to the study of non-local boxes~\cite{brunner2014bell,scarani2019bell}, where the no-signalling condition is usually motivated by demanding that Alice's and Bob's procedures are spacelike separated, but where the probabilistic formalism does not strictly depend on this interpretation. For such boxes, one might also imagine that the procedures are close-by but separated by a screening wall~\cite{BarrettColbeckKent}, or that the statistics just happens to not be signalling for other reasons. However, the most compelling physical situation in which non-local boxes are realized are those including spacelike separation. Similarly, the most compelling physical realizations of our rotation boxes will be via physical rotations in space.

Note that we do not need to assume a picture that is as specific as depicted in Figure~\ref{fig:rotationboxes}c: there need not literally be a ``transmission of some system'' from the preparation to the measurement device. We can also think of the preparation as just happening somewhere in space, and the measurement happening at the same place later in time. In this case, any time evolution happening in between the two events will be considered part of the preparation procedure. More generally, the physical transmission of the system to the measurement device can \textit{also} be considered part of the measurement procedure. Furthermore, what a physical system really ``is'', and whether we might want to think of it as some actual object with standalone properties, is irrelevant for our analysis.

We will make one further assumption that is often made in the semi-DI framework: essentially, that there is no preshared entanglement between the preparation and measurement devices. More generally:

\textit{The preparation and measurement devices are initially uncorrelated. That is, all correlations between them are established by the preparation procedure.}

This has several important consequences, for example the following. Imagine an entangled state of two spin-$1/2$ particles shared between preparation and measurement devices. Suppose that the preparation device is rotated by $360^\circ$, i.e.\ $2\pi$. Then this may introduce a phase factor of $(-1)$ on the preparation subsystem. After transmission to the measurement device, this relative phase can be detected. Thus, a $2\pi$-rotation of the preparation device would induce a transformation on the physical system that does not correspond to the identity. Our assumption above excludes such behavior.

We will be interested in how the probability of the outcome can depend on this spatial rotation, i.e.\ in the conditional probability $P(a|\theta)$. Without any further assumptions, this probability is not constrained at all: we will see that continuity in $\theta$ is the unique assumption arising from the standard formalism of quantum theory. We will thus add a simple assumption that has often a natural realization in QT: that the physical systems which are generated by the preparation device admit an upper bound $J$ on their ${\rm SO}(2)$-charge, $J\in\{0,\frac 1 2,1,\frac 3 2,\ldots\}$. This is an abstract representation-theoretic assumption about how the physical system is allowed to react to spatial rotations. Within QT, it bounds the system's total angular momentum quantum number relative to the measurement device. If there is no angular momentum, e.g.\ if we imagine sending a point particle on the axis of rotation to the measurement device as depicted in Figure~\ref{fig:rotationboxes}c, then this becomes a bound on the spin of the system. To save some ink, we will always have this idealized example in mind, and talk about ``spin-bounded rotation boxes'' in this paper. A more detailed definition and discussion is given in the following subsections.

Since we will only study sets of correlations that arise from upper bounds on the spin, we can always extend our preparation procedure and allow it to prepare an additional spin-$0$ system (i.e.\ a system that does not respond to spatial rotations at all) in some random choice of classical basis state. Keeping one copy and transferring the other one to the measurement device will establish shared classical randomness between the two devices, and we can imagine that this happens before the rest of the procedure is accomplished. This shows the following:

\textit{All our results remain unchanged if we allow preshared classical randomness between the preparation and measurement devices.}

Mathematically, this will be reflected in the fact that all our sets of spin-bounded correlations will be convex.

Let us now turn to the mathematical description of rotation boxes of bounded spin. We will begin by assuming quantum theory, and drop this assumption in the subsequent subsection.

\subsection{Quantum spin-$J$ correlations $\mathcal{Q}_J$}\label{subsec:qspinJ}

Let us assume that the Hilbert space on which the preparation procedure acts is finite-dimensional. In quantum theory, spacetime symmetries are implemented via projective representations on a corresponding Hilbert space. It is easy to see, and shown by some of us in~\cite{Jones}, that this implies that there is some finite set $\mathcal{J}$ of, either, integers ($\mathcal{J}\subset\mathbb{Z}=\{\ldots,-2,-1,0,1,2,\ldots\}$) or half-integers ($\mathcal{J}\subset\mathbb{Z}+\frac 1 2 =\{\ldots,-5/2,-3/2,-1/2,1/2,3/2,5/2,\ldots\}$) such that the representation is
\[
   U'_\theta=\bigoplus_{j\in\mathcal{J}} \mathbf{1}_{n'_j} e^{ij\theta},
\]
where the $n'_j\in\mathbb{N}$ are integers. That is, the rotation by angle $\theta$ is represented by a diagonal matrix (in some basis) of complex exponentials, repeating an arbitrary number of times. Only integers \textit{or} half-integers may appear, which is an instance of the univalence superselection rule which forbids superpositions of bosons and fermions.

Let us begin by writing the above in a canonical form. Setting $m:=\min\mathcal{J}$ and $M:=\max\mathcal{J}$ as well as $\Delta:=(m+M)/2$, we can obtain the representation $U_\theta:=e^{-i\Delta\theta} U'_\theta$ which acts in the same way on density matrices. It is straightforward to see that it has the form
\begin{equation}
   U_\theta=\bigoplus_{j=-J}^J \mathbf{1}_{n_j} e^{ij\theta},
   \label{eqRepDef}
\end{equation}
where $n_j:=n'_{j+\Delta}$ (or zero if the latter is undefined) and $J:=(M-m)/2$. We stipulate that \textit{quantum spin-$J$ rotation boxes} are those that are described by projective unitary representations of this form. As always in this paper, we have $J\in\{0,\frac 1 2,1,\frac 3 2,2,\ldots\}$. We say that $U_\theta$ is a \textit{proper} quantum spin-$J$ rotation box if it is not also a quantum spin-$(J-\frac 1 2)$ box, i.e.\ if $n_J$ and $n_{-J}$ in~(\ref{eqRepDef}) are both non-zero.

Quantum spin-$J$ rotation boxes can now be described as follows. The preparation device prepares a fixed quantum state $\rho$. The spatial rotation of the device by angle $\theta$ maps this state to $U_\theta\rho U_\theta^\dagger$. Finally, the measurement device performs some measurement described by a POVM $\{E_a\}_{a\in\mathcal{A}}$, where $\mathcal{A}$ is the set of possible outcomes. In this paper, we are only interested in the case that $\mathcal{A}$ is a finite set, but this can straightforwardly be generalized.

\begin{definition}
\label{DefQuantum}
The set of quantum spin-$J$ correlations with outcome set $\mathcal{A}$, where $|\mathcal{A}|\geq 2$, will be denoted $\mathcal{Q}_J^{\mathcal{A}}$, and is defined as follows. It is the collection of all $\mathcal{A}$-tuples of probability functions
\[
   \left(\strut \theta\mapsto P(a|\theta)\right)_{a\in\mathcal{A}},
\]
such that there exists a Hilbert space with a projective representation of ${\rm SO}(2)$ of the form~(\ref{eqRepDef}), some quantum state (i.e.\ density matrix) $\rho$, and a POVM $\{E_a\}_{a\in\mathcal{A}}$ on that Hilbert space such that
\[
   P(a|\theta)=\Tr(U_\theta\rho U_\theta^\dagger E_a).
\]
The special case of two outcomes, $\mathcal{A}=\{-1,+1\}$, will be denoted $\mathcal{Q}_J$ (without the $\mathcal{A}$-superscript). Instead of pairs of probability functions, we can equivalently describe this set by the collection of functions $P(+1|\theta)$ only, because $P(-1|\theta)=1-P(+1|\theta)$ follows from it.
\end{definition}
Note that the integers $n_j$ in Eq.~(\ref{eqRepDef}) can be arbitrary finite numbers, and so there is no a priori upper bound on the Hilbert space dimension on which the rotation box is represented. We can use this to prove convexity of these sets of correlations:
\begin{lemma}
\label{LemConvex}
The sets $\mathcal{Q}_J^{\mathcal{A}}$ are convex.	
\end{lemma}
\begin{proof}
Let $P,\tilde P\in\mathcal{Q}_J^{\mathcal{A}}$, then
\[
   P(a|\theta)=\Tr(E_a U_\theta \rho U_\theta^\dagger),\quad \tilde P(a|\theta)={\rm Tr}(\tilde E_a \tilde U_\theta \tilde \rho \tilde U_\theta^\dagger)
\]
for suitable representations, quantum states, and POVM elements. If $0\leq\lambda\leq 1$, we can define the block matrices
\[
   F_a:=E_a\oplus \tilde E_a,\quad \sigma:=\lambda\rho \oplus (1-\lambda)\tilde \rho,\qquad V_\theta:= U_\theta\oplus \tilde U_\theta,
\]
such that the $F_a$ form a POVM, $\sigma$ is a density matrix, and $V_\theta$ is still a representation of the form~(\ref{eqRepDef}). Then
\[
   \lambda P(a|\theta)+(1-\lambda)\tilde P(a|\theta)={\rm Tr}(F_a V_\theta \sigma V_\theta^\dagger),
\]
hence $\lambda P +(1-\lambda)\tilde P \in \mathcal{Q}_J^{\mathcal{A}}$.
\end{proof}

At first sight, it seems as if our choice of terminology conflicts with its usual use in physics: there, a spin-$J$ system is typically meant to describe a spin-$J$ irrep (irreducible representation) of ${\rm SU}(2)$, living on a $(2J+1)$-dimensional Hilbert space. Remarkably, we will now show that we can realize all quantum spin-$J$ correlations exactly on such systems:
\begin{theorem}\label{thm:QJ_corr_form}
Let $P\in\mathcal{Q}_J^\mathcal{A}$ be any quantum spin-$J$ correlation. Then there exists a pure state $|\psi\rangle\in\mathbb{C}^{2J+1}$ and a POVM $\{E_a\}_{a\in\mathcal{A}}$ on $\mathbb{C}^{2J+1}$ such that
\[
   P(a|\theta)=\langle\psi|U_\theta^\dagger E_a U_\theta |\psi\rangle,
\]
where $U_\theta:=e^{i\theta Z}$, with $Z={\rm diag}(J,J-1,\ldots,-J)$. Moreover, we can choose $|\psi\rangle$ to have real nonnegative entries in any chosen eigenbasis of $Z$.

In particular, without loss of generality, we can always assume that $n_j=1$ in Eq.~(\ref{eqRepDef}).
\end{theorem}
In other words, we can always assume that the ${\rm SO}(2)$-rotation is given by rotations around a fixed axis of a spin-$J$ particle in the usual sense, i.e.\ one that is described by a spin-$J$ irrep of ${\rm SU}(2)$. We note that two different spin-$J$ correlations $P(a|\theta)$ and $P'(a|\theta)$ may require different orbits $U_\theta \ket \psi$ and $U_\theta \ket \psi'$ as well as different POVMs to be generated.

The proof is cumbersome and thus deferred to Appendix~\ref{app:proof_QJ_corr_form}. A simple consequence of Theorem~\ref{thm:QJ_corr_form} is that the sets $\mathcal{Q}_J^\mathcal{A}$ are compact: they arise from the compact sets of $|\mathcal{A}|$-outcome POVMs and quantum states on $\mathbb{C}^{2J+1}$ under a continuous map, mapping the pair $(\{E_a\},\rho)$ to the function $\theta\mapsto {\rm Tr}(U_\theta \rho U_\theta^\dagger E_a)$. Furthermore, multiplying out the complex exponentials in $U_\theta=e^{i\theta Z}$ shows that these functions are all trigonometric polynomials of degree at most $2J$ (as in Lemma~\ref{LemmReal2Complex}). As we show in the appendix, we can say more:
\begin{lemma}
\label{LemCompact}
The correlation sets $\mathcal{Q}_J^{\mathcal{A}}$ are compact convex subsets of full dimension $(|\mathcal{A}|-1)(4J+1)$ of the $|\mathcal{A}|$-tuples of trigonometric polynomials of degree $2J$ or less that sum to one.
\end{lemma}
This lemma is proven in Appendix~\ref{app:proof_LemCompact}.

In particular, for $\mathcal{A}=\{+1,-1\}$, the set $\mathcal{Q}_J$ is a compact subset of the trigonometric polynomials of degree at most $2J$, of full dimension $4J+1$.

As a simple example, consider the case of two outcomes, $\mathcal{A}=\{-1,+1\}$, and $J=1/2$. Then $\mathcal{Q}_{1/2}$ is a compact convex set of dimension $3$. Its elements are pairs $(P(+|\theta),P(-|\theta))$. Since $P(-|\theta)=1-P(+|\theta)$, we need to specify the functions $P(+|\theta)$ only, and can identify $\mathcal{Q}_{1/2}$ with this set of functions. Every such function is a trigonometric polynomial of degree one,
\[
P(+|\theta)=c_0+c_1\cos\theta+s_1\sin\theta,
\]
and we can depict $\mathcal{Q}_{1/2}$ by plotting the possible values of $c_0$, $c_1$ and $s_1$. The result is shown in Figure~\ref{fig:spinOneHalfSDP}. Indeed, as we will show in Subsection~\ref{Subsec012}, in this simple case, the only condition for a trigonometric polynomial of degree one to be contained in $\mathcal{Q}_{1/2}$ is that $P(+|\theta)$ gives valid probabilities, i.e.\ that $0\leq P(+|\theta)\leq 1$ for all $\theta$. This simple characterization will, however, break down for larger values of $J$, as we will see.
\begin{figure}[hbt]
\centering 
\includegraphics[width=1\columnwidth]{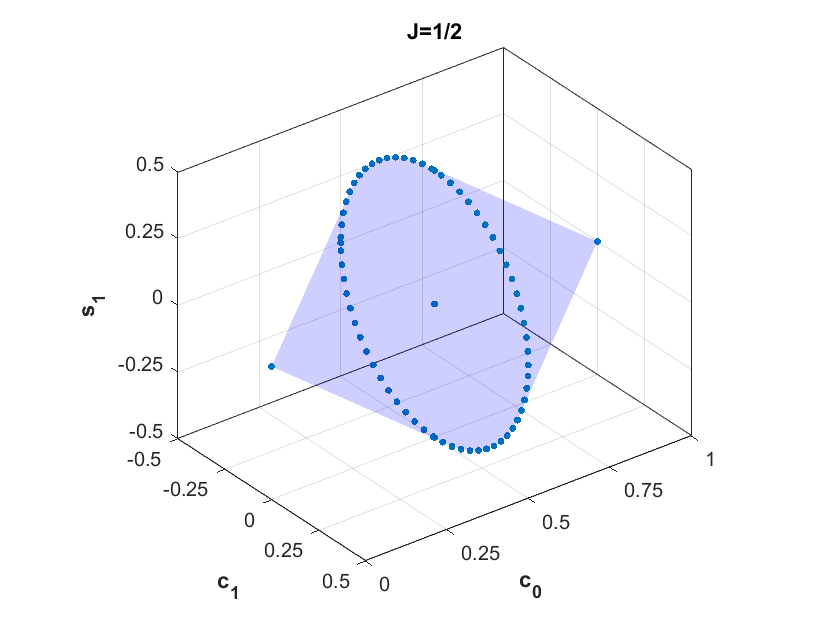}
\caption{The binary quantum spin-$1/2$ correlations $\mathcal{Q}_{1/2}$, which happens to be the set of trigonometric polynomials $P(+|\theta)=c_0+c_1\cos\theta+s_1\sin\theta$ with $0\leq P(+|\theta)\leq 1$ for all $\theta$. The two endpoints are the constant zero and one functions, and the other extremal points on the circle correspond to functions $\theta\mapsto \frac 1 2 +\frac 1 2 \cos(\theta-\varphi)$, with $\varphi$ some fixed angle.}
\label{fig:spinOneHalfSDP}
\end{figure}

Further, as we prove in the Appendix~\ref{app:LemContained}, the set of spin-$J$ quantum correlations for any fixed outcome set $\mathcal{A}$ grows with increasing $J$:
\begin{lemma}
\label{LemContained}
For all $J$, we have $\mathcal{Q}_J^\mathcal{A}\subset \mathcal{Q}_{J+1/2}^\mathcal{A}$.	
\end{lemma}
Since $\dim \mathcal{Q}_J^\mathcal{A}<\dim\mathcal{Q}_{J+1/2}^\mathcal{A}$, this set inclusion is strict.

In the next section, we will drop the requirement that the rotation box -- or, rather, the corresponding prepare-and-measure scenario -- is described by quantum theory. In order to do so, we will leave the framework of Hilbert spaces, and make use of general state spaces that could describe the scenario. To consider quantum boxes as a special case of a general scenario of this kind, we have to slightly reformulate their description: while it is convenient to consider unitary transformations acting on state vectors, quantum states are actually \textit{density matrices}, and the rotations act on them by unitary conjugation, $\rho\mapsto U_\theta \rho U_\theta^\dagger$. The following lemma gives a representation-theoretic characterization of quantum spin-$J$ boxes in terms of the way that spatial rotations act on the density matrices. This reformulation will later on allow us to motivate and derive the generalized definition of rotation boxes beyond quantum theory.

\begin{lemma}\label{lem:rep_decomp_Q}
Let $\theta\mapsto U_\theta$ be any finite-dimensional projective representation of ${\rm SO}(2)$. Then the following statements are equivalent:
\begin{itemize}
    \item[(i)] Up to global phases, the representation can be written in the form~(\ref{eqRepDef}) with $n_J n_{-J}\neq 0$, i.e.\ it is a representation corresponding to a proper quantum spin-$J$ rotation box.
    \item[(ii)] The maximum degree of any trigonometric polynomial $\theta\mapsto {\rm Tr}(U_\theta \rho U_\theta^\dagger E)$, where $\rho$ is any quantum state and $E$ any POVM element, equals $2J$.
    \item[(iii)] The associated real representation on the density matrices, $\theta\mapsto U_\theta \bullet U_\theta^\dagger$, decomposes on the real vector space of Hermitian matrices into
     \begin{align}
       \mathbf{1}_{m_0}\oplus
        \bigoplus_{k= 1}^{2J} \mathbf{1}_{m_k} \otimes \begin{pmatrix}
            \cos(k \theta) & - \sin( k \theta) \\
            \sin(k \theta) & \cos( k \theta)
        \end{pmatrix},
    \end{align}
    where the $m_k$ are non-negative integers with $m_{2J}\neq 0$. In the case where $n_j = 1$ for all $j \in \{-J,..., J\}$, i.e.\ when we have the representation on $\mathbb {C}^{2J+1}$ derived in Theorem~\ref{thm:QJ_corr_form}, we obtain $m_k = 2J+1-k$.
\end{itemize}
\end{lemma}
This lemma is proven in Appendix~\ref{app:proof_lem_rep_decomp_Q}.
Let us now drop the assumption that quantum theory holds, and consider more general rotation boxes.

\subsection{General spin-$J$ correlations $\mathcal{R}_J$}\label{subsec:genspinJ}

We now introduce the framework of \textit{spin-J rotation boxes}~\citep{Garner, Jones}. Similarly to quantum rotation boxes, a general spin-$J$ rotation box has a preparation procedure that can be rotated by some angle $\theta \in \SO(2)$ relative to the measurement procedure, which in turn yields some output $a \in \mathcal A$. The behavior of the box is given by the set of probability functions $\{P(a|\theta)\}_{a \in \mathcal A}$, where $P(a|\theta): \mathbb{R} \to \mathbb{R}$ satisfies $0\leq P(a|\theta)\leq 1$ for all $\theta$ and $P(a|\theta) = P(a| \theta + 2 n \pi)$ for all $n \in \ints$.

But how can we characterize such boxes without appeal to quantum theory, and how can we say what it even means that such a box has spin at most $J$? Let us begin with an obvious guess for what the answer to the second question should be, before we justify this by answering the first question.

Our main observation will be that every $\theta\mapsto P(a|\theta)$ of a quantum spin-$J$ correlation $P\in\mathcal{Q}_J^\mathcal{A}$ is a trigonometric polynomial of degree at most $2J$. In the characterization of the set $\mathcal{Q}_J^\mathcal{A}$, we demand in addition that the resulting probability functions come from a quantum state and POVM together with a unitary representation of ${\rm SO}(2)$ on a Hilbert space, producing these probabilities via the Born rule. It seems therefore natural to drop the latter condition, and to \textit{only} demand that the $P(a|\theta)$ are trigonometric polynomials of degree at most $2J$, giving valid probabilities for all $\theta$. This will be our definition of a general spin-$J$ correlation, to be contrasted with the quantum version in Definition~\ref{DefQuantum}:
\begin{definition}\label{def:spin_J_gen_cor}
The set of (general) spin-$J$ correlations with outcome set $\mathcal{A}$, where $|\mathcal{A}|\geq 2$, will be denoted $\mathcal{R}_J^\mathcal{A}$, and is defined as follows. It is the collection of all $\mathcal{A}$-tuples of functions
\[
\left(\theta\mapsto P(a|\theta)\right)_{a\in\mathcal{A}}
\]
such that every one of the functions is a trigonometric polynomial of degree at most $2J$ in $\theta$, and $0\leq P(a|\theta)\leq 1$ as well as $\sum_{a\in\mathcal{A}} P(a|\theta)=1$ for all $\theta$.

The special case of two outcomes, $\mathcal{A}=\{-1,+1\}$, will be denoted $\mathcal{R}_J$ (without the $\mathcal{A}$-superscript). Instead of pairs of probability functions, we can equivalently describe this set by the collection of functions $P(+1|\theta)$ only, because $P(-1|\theta)=1-P(+1|\theta)$ follows from it.
\end{definition}
For concreteness, and for later use, let us denote here again what we mean by a trigonometric polynomial of degree at most $2J$, and how we typically represent it:
\begin{lemma}
\label{LemmReal2Complex}
Suppose that $P$ is a real trigonometric polynomial of degree $2J$, and write it as
\begin{eqnarray*}
P(\theta)&=& c_0 + \sum_{j=1}^{2J} \left(\strut c_j\cos(j\theta)+s_j\sin(j\theta)\right)=\sum_{k=-2J}^{2J} a_k e^{ik\theta}.
\end{eqnarray*}
Then $a_{-j}=\overline{a_j}$, $a_0=c_0$,  and for all $j\geq 1$, we have $c_j=2\, {\rm Re}(a_j)$ and $s_j=-2\, {\rm Im}(a_j)$.
\label{lemma:inequality}
\end{lemma}
This follows from a straightforward calculation.

Clearly, by construction, this notion of spin-$J$ correlations generalizes that of the quantum spin-$J$ correlations:
\begin{lemma}\label{lem:qspin_sub_gspin}
Every quantum spin-$J$ correlation is a spin-$J$ correlation. That is, $\mathcal{Q}_J^\mathcal{A}\subseteq \mathcal{R}_J^\mathcal{A}$.
\end{lemma}
The comparison of these two sets will be our main question of interest in the following sections. But first, let us return to the question of how to understand rotation boxes without assuming quantum theory, and how to obtain the notion of spin-$J$ correlations in a representation-theoretic manner.

As will be shown, all general rotation box correlations can be generated by an underlying physical system, which may not be quantum. Non-quantum systems can be defined using the framework of Generalized Probabilistic Theories (GPTs). For an introduction to GPTs, see e.g.~\cite{Hardy,Barrett,Mueller,Plavala}. A GPT system $A$ consists of a set of states $\Omega_A$ which is a convex subset of a real finite-dimensional vector space $V_A$ and a convex set of effects $\mE_A \subset V_A^*$. We assume that $\Omega_A$ and $\mE_A$ span $V_A$ and $V_A^*$ respectively. This assumption is automatically satisfied if the GPT is constructed from an operational theory, defining states as equivalence classes of preparation procedures, and effects as equivalence classes of outcomes of measurement procedures~\cite{Holevo,Schmid}. The natural pairing $(e,\omega) \in [0,1]$ gives the probability of the measurement outcome corresponding to the effect $e$ when the system is in state $\omega$. A measurement is a set of effects $\{e_i\}_i$ such that $\sum_i e_i = u$ with $u$ the unit effect, which is the unique effect such that $(u,\omega) = 1$ for all $\omega \in \Omega_A$. A transformation of a GPT system $A$ is given by a linear map $T: V_A \to V_A$ which preserves the set of states, $T(\Omega_A) \subset \Omega_A$, and the set of effects, $T^*(\mE_A)\subset \mE_A$. The linearity of these maps follows from the assumption that statistical mixtures of preparation procedures must lead to the corresponding statistical mixtures of outcome probabilities, for all possible measurements after the transformation. The set of all transformations of the system $A$ is given by a closed convex subset of the linear space $\mathcal{L}(V_A)$ of linear maps from $V_A$ to itself.

The set of reversible transformations ${\rm Rev}(A)$ corresponds to those transformations $T$ for which $T^{-1}$ exists and is also a transformation. It forms a group under composition of linear maps. If there exists a group homomorphism $G \to {\rm Rev}(A)$ (i.e.\ a representation of $G$) for some group $G$ then $G$ is said to be a symmetry of $A$. In this spirit, the set $\{\bar{\phi}_\alpha\}_\alpha$ of \Cref{sec:rot_box_frame} (or, more precisely, the linear extensions of those maps) are an $\SO(2)$ symmetry of the GPT system that describes the scenario. If, given a GPT system $(A, \Omega_A, E_A)$ with an $\SO(2)$ symmetry $\theta \mapsto T_\theta$, with $T_\theta \in {\rm Rev}(A)$, then the probability distribution $P(a|\theta) =  (e_a, T_\theta \omega)$ is a rotation box correlation. In this case, we say that the correlation $P(a|\theta)$ can be generated by the GPT system $A$.
\begin{lemma}
\label{LemRepGPT}
Consider any finite-dimensional GPT system $A=(V_A,\Omega_A,\mE_A)$, together with a representation of ${\rm SO}(2)$, $\theta\mapsto T_\theta$, such that every $T_\theta$ is a reversible transformation. Then the following are equivalent:
\begin{itemize}
    \item[(i)] The maximum degree of any trigonometric polynomial $\theta\mapsto (e, T_\theta \omega)$, where $\omega\in\Omega_A$ is any state and $e\in \mE_A$ any effect, equals $2J$.
    \item[(ii)] The real representation $\theta\mapsto T_\theta$ of ${\rm SO}(2)$ decomposes on the real vector space $A$ into
     \begin{align}\label{eq:GPT_SO2_rep}
       \mathbf{1}_{m_0}\oplus
        \bigoplus_{k= 1}^{2J} \mathbf{1}_{m_k} \otimes \begin{pmatrix}
            \cos(k \theta) & - \sin( k \theta) \\
            \sin(k \theta) & \cos( k \theta)
        \end{pmatrix},
    \end{align}
    where the $m_k$ are integers with $m_{2J}\neq 0$.
\end{itemize}
If one of these two equivalent conditions is satisfied, we call the GPT system a \textit{spin-$J$ GPT system}.
\end{lemma}
\begin{proof}
Since $\theta\mapsto T_\theta$ is a representation of ${\rm SO}(2)$ on the real vector space $V_A$, it can be decomposed into irreps. In some basis, this gives us the representation $T_\theta=\mathbf{1}_{m_0}\oplus
        \bigoplus_{k= 1}^n \mathbf{1}_{m_k} \otimes \begin{pmatrix}
            \cos(k \theta) & - \sin( k \theta) \\
            \sin(k \theta) & \cos( k \theta),
        \end{pmatrix}$ for some finite integer $n$, where $m_n\neq 0$. Now since $\Omega_A$ spans $V_A$ and $\mE_A$ spans $V_A^*$, the linear functionals $T\mapsto (e,T\omega)$ span $\mathcal{L}(V_A)^*$, where $\mathcal{L}(V_A)$ is the set of linear operators on $V_A$. In other words, there will be some real numbers $\alpha_i$, effects $e_i$ and states $\omega_i$ such that $\sum_i \alpha_i (e_i,T_\theta\omega_i)$ yields the component $\cos(n\theta)$, and this is only possible if $\theta\mapsto (e,T_\theta\omega)$ is a trigonometric polynomial of degree at least $n$ for some effect $e$ and state $\omega$. But the degree of this trigonometric polynomial can of course not be higher than $n$.
\end{proof}
This characterization resembles Lemma~\ref{lem:rep_decomp_Q} for the quantum case: it tells us that quantum spin-$J$ rotation boxes are spin-$J$ GPT systems. And it allows us to obtain a justification for our definition of spin-$J$ correlations:
\begin{theorem}
\label{TheSpinJCorrelation}
Let $P\equiv \left(P(a|\theta)\right)_{a\in\mathcal{A}}$ be an $\mathcal{A}$-tuple of functions in $\theta$. Then the following are equivalent:
\begin{itemize}
\item[(i)] $P$ is a spin-$J$ correlation, i.e.\ $P\in\mathcal{R}_J^\mathcal{A}$.
\item[(ii)] There is a spin-$J$ GPT system $(V_A,\Omega_A,E_A)$ with a state $\omega\in\Omega_A$ and measurement $\{e_a\}_{a\in\mathcal{A}}\subset \mE_A$ such that $P(a|\theta)=(e_a,T_\theta \omega)$.
\end{itemize}
\end{theorem}
Proving the implication $(ii)\Rightarrow (i)$ is immediate, given Lemma~\ref{LemRepGPT}. For the converse implication, we will now show how all correlations in $\mathcal{R}_J^\mathcal{A}$ can be reproduced in terms of a single GPT system that we will call $\tR_J$:

\begin{definition}[Spin-$J$ rotation box system $\tR_J$]\label{def:rot_box_system}
Let  $\tR_J$  be a GPT system with state space $\Omega_J \subset \reals^{4J + 1}$ and effect space $\mE_J \subset \mathbb{R}^{4J+1}$ defined as follows: 
\begin{align}
    \Omega_J = \conv \left(\{\omega_J(\theta)\,\,|\,\, \theta \in [0 , 2 \pi)\}\right) ,
    \end{align}
    with
\begin{align}\label{eq:orbitJGPT}
      \omega_J(\theta) = 
    \begin{pmatrix}
        1 \\
        \cos (\theta) \\
        \sin (\theta) \\
        \vdots \\
        \cos (k \theta) \\
        \sin(k \theta) \\
        \vdots \\
        \cos (2 J \theta) \\
        \sin (2J \theta) 
    \end{pmatrix} ,
\end{align}
and
\begin{align}
    \mE_J : = \{ e \in \mathbb{R}^{4J+1}\,\,|\,\, e \cdot \omega \in [0,1] \mbox{ for all } \ \omega \in \Omega_J \} .
\end{align}
The unit effect is 
\begin{align}
    u = (1, 0, ... , 0) .
\end{align}

The system $\tR_J$ carries a representation $\SO(2) \to \mathcal{L}(\mathbb{R}^{4J+1})$,  $\theta\mapsto T_\theta$ of $\SO(2)$, given by
\begin{align}
    T_\theta &= \bigoplus_{k = 0}^{2J}  
    \gamma_k(\theta) ,  \\
    \gamma_0(\theta) &= 1 ,  \label{eq:SO2realirrep0} \\
    \gamma_k(\theta) &= 
    \begin{pmatrix}
        \cos (k \theta) & - \sin(k \theta) \\
        \sin (k \theta) & \cos(k \theta)
    \end{pmatrix}, \ k \in \{1,...,2J\} . \label{eq:SO2realirrep}
\end{align}
\end{definition}
The system $\tR_J$ is an unrestricted system by definition. These systems belong to the family of GPT systems with pure states given by the circle $S^1$ and reversible dynamics $\SO(2)$; i.e.\ for $J \geq 1$, they can be interpreted as rebits with modified measurement postulates~\cite{galley2021dynamics}. The state space $\Omega_J$ is the convex hull of an $\SO(2)$ orbit of the vector $\omega(0) \in \reals^{4J+1}$ and is hence an $\SO(2)$ orbitope~\cite{Sanyal}. 

The system $\tR_J$ is  \textit{canonical} in the sense that the $\SO(2)$ correlation set it generates is exactly  $\mR_J^\mA$, as shown in the following lemma:
\begin{lemma}
\label{LemRJA}
    The set of spin-$J$ correlations $\mR_J^\mA$ can be generated by the system $\tR_J$: for every $P\in\mathcal{R}_J^{\mathcal{A}}$, there is a measurement $\{e_a\}_{a\in\mathcal{A}}$ on $\tR_J$ with
\[
   P(a|\theta)=e_a\cdot \omega_J(\theta).
\]
    Conversely, every tuple of probability functions $(P(a|\theta))_{a \in \mA}$ generated in this way with measurements in $\tR_J$ is in $\mR_J^{\mathcal{A}}$.
\end{lemma}
\begin{proof}
    The set $\mR_J$ is given by all functions  $P: \theta \mapsto [0,1]$ of the form $P(\theta) =  c_0 + \sum_{j=1}^{2J} ( c_j \cos(j\theta)+s_j \sin(j\theta))$.  This can be expressed as
    \begin{align}
        P(\theta) = e \cdot \omega_J(\theta),
    \end{align}
    where $\omega_J(\theta)$ is defined as in \Cref{eq:orbitJGPT} and $e = (c_0, c_1,s_1,..., c_{2J}, s_{2J})$. $e$ is an effect on the system $\tR_J$ since by construction $e \cdot \omega_J(\theta) \in [0,1] $, which in turn implies $e \cdot \omega \in [0,1]$ for all $\omega \in \conv\{\omega_J(\theta)| \theta \in [0,2\pi)\} = \Omega_J$. This show that any $P(\theta)$ can be generated using the orbit of states $\{\omega_J(\theta)\,\,|\,\,\theta \in [0,2 \pi)\}$. 

    Given a tuple $(P(a|\theta))_{a \in \mA} \in \mR_J^{\mathcal A}$, we show that it can be generated by a measurement $\{e_a\}_{a \in \mA}$ applied to the orbit $\Omega_J(\theta)$.

    $P(a|\theta)$ is a function $\theta \mapsto [0,1]$ of the form $P(a|\theta) =  c_0^a + \sum_{j=1}^{2J} ( c_j^a \cos(j\theta)+s_j^a \sin(j\theta))$. The requirement $\sum_{a \in \mA} P(a|\theta) = 1$ for all $\theta$ implies that
    \begin{align}
        \sum_a \left(c_0^a + \sum_{j=1}^{2J} \left(\strut c_j^a \cos(j\theta)+s_j^a \sin(j\theta)\right)\right) = 1,
    \end{align}
    which in turn entails
    \begin{equation}
    \sum_a c_0^a = 1,\quad \sum_a c_j^a=\sum_a s_j^a = 0\quad(1\leq j \leq 2J).
    \label{eq:condsum}
    \end{equation}
    Every $P(a|\theta) = e_a \cdot \omega_J(\theta)$ for $e_a = (c_0^a, c_1^a, s_1^a,..., c_{2J}^a,s_{2J}^a)$ which is a valid effect. Moreover, the conditions of \cref{eq:condsum} entail that $\sum_{a \in \mA} e_a = u$ with $u$ the unit effect. Hence $\{e_a\}_{a \in \mA}$ form a measurement.
    
    Conversely, consider an arbitrary tuple $(P(a|\theta))_{a \in \mA}$ of $\SO(2)$ probability functions generated by $\tR_J$:
    \begin{align}
       P(a|\theta) = e_a \cdot T_\theta \omega ,
    \end{align}
    where $\sum_{a \in \mA} e_a = u$ and $\omega \in \Omega_J$. Since $T_\theta \in \mL(\reals^{4J+1})$ $P(a|\theta)$ is a linear functional, $\mL(\reals^{4J+1}) \to \mathbb{R}$ and hence in  $\mL(\reals^{4J+1})^*$. This implies that $P(a|\theta)$ is a linear combination of entries in $T_\theta$ and therefore a trigonometric polynomial of order at most $2J$. Hence $P(a|\theta) \in \mR_J$.

    The condition $\sum_a e_a = u$ implies
\begin{align}
      \sum_{a \in \mA} P(a|\theta) = \sum_a e_a \cdot T_\theta \omega = u \cdot T_\theta \omega = 1
\end{align}
Thus $(P(a|\theta))_{a \in \mA} \in \mR_J^\mA$.
\end{proof}
 It follows from the proof of the above lemma that the effect space $\mE_J$ is isomorphic to $\mathcal R_J$ as a convex set.

\subsection{General spin-$J$ correlations as a relaxation of the quantum set}
\label{SubsecRelaxation}

The space of spin-$J$ correlations $\mR_J$ is defined independently of the quantum formalism, however it can also be interpreted as arising from a relaxation of the quantum formalism. 

To see that, we start by noting the Fej\'er-Riesz theorem~\cite{RieszFejer}, which has several important applications for quantum and general rotation boxes:
\begin{theorem}[Fej\'er-Riesz theorem]\label{thm:fejerriesz}
Suppose that $P(\theta):=\sum_{j=-2J}^{2J} a_j e^{ij\theta}$ satisfies $P(\theta)\geq 0$ for all $\theta$. Then there is a trigonometric polynomial $Q(\theta):=\sum_{j=-J}^J b_j e^{ij\theta}$ such that $P(\theta)=|Q(\theta)|^2$.
\end{theorem}
From this, we can easily derive the following Lemma:
\begin{lemma}
\label{LemFRtrig}
Let $P(\theta)=\sum_{j=-2J}^{2J} a_j e^{ij\alpha}$ be a trigonometric polynomial. Then we have $P(\theta)\geq 0$ for all $\theta\in\mathbb{R}$ if and only if there exists a vector $b=(b_0,b_1,\ldots,b_{2J})\in\mathbb{C}^{2J+1}$ such that
\[
   a_k=\sum_{0\leq j,j+k\leq 2J} \overline{b_j} b_{j+k}.
\]	
\end{lemma}
Note that necessarily
\[
   \|b\|^2=a_0 =\frac 1 {2\pi}\int_0^{2\pi} P(\theta) {\rm d}\theta,
\]
and the matrix $Q_{jk}:=\overline{b_j} b_k$ is positive semidefinite.
Consequently, the following theorem follows from Fej\'er-Riesz's theorem: 
\begin{theorem}
\label{ThmRJ}
If $P\in\mathcal{R}_J$, then there is a pure quantum state $|\psi\rangle$ on $\mathbb{C}^{2J+1}$ and a positive semidefinite matrix $E_+\geq 0$ such that
\[
   P(+|\theta)=\langle\psi|U_\theta^\dagger E_+ U_\theta |\psi\rangle.
\]
We can always choose $|\psi\rangle$ as the uniform superposition $|\psi\rangle:=(2J+1)^{-1/2}\sum_{j=-J}^J |j\rangle$, $U_\theta$ as defined in Theorem~\ref{thm:QJ_corr_form}, and $E_+=(2J+1)|b\rangle\langle b|$, where $|b\rangle$ is the vector from Lemma~\ref{LemFRtrig}. Note, however, that $E_+$ is not in general a POVM element, i.e.\ it will in general have eigenvalues larger than $1$.
\end{theorem}
\begin{proof}
    Let $P(+|\theta)=\sum_{j=-2J}^{2J} a_j e^{ij\theta} \in \mR_J$, then by ~\Cref{thm:fejerriesz}:
    \begin{align}
        P(+|\theta) = \left(\sum_{j = -J}^J \bar b_j e^{-i j \theta}\right)  \left(\sum_{k = -J}^J b_k e^{ik \theta}\right) .
    \end{align}
    Now use $U_\theta$  as defined in Theorem~\ref{thm:QJ_corr_form}, with orthonormal basis $\{\ket j\}_{j=-J}^J$ such that $U_\theta|j\rangle=e^{ij\theta}|j\rangle$, and define $\ket b = \sum_j b_j \ket j$. Then
    \begin{eqnarray*}
        P(+|\theta) &=&\left(\bra b U_\theta^\dagger \sum_j \ket j\right)  \left(\sum_k \bra{k} U_\theta \ket b\right) \\ &=& \langle\psi|U_\theta^\dagger E_+ U_\theta |\psi\rangle,
    \end{eqnarray*}
    where $|\psi\rangle:=(2J+1)^{-1/2}\sum_{j=-J}^J |j\rangle$ and $E_+=(2J+1)|b\rangle\langle b|$.  
\end{proof}

Therefore, rotation boxes can be regarded as a relaxation of the quantum formalism: instead of demanding that $E_+$ gives valid probabilities on \textit{all} states (which would imply $0\leq E_+ \leq \mathbf{1}$), the above only demands that it gives valid probabilities \textit{on the states of interest}, i.e.\ on the states $U_\theta|\psi\rangle$ for all $\theta$ and some fixed state $\ket \psi$. This is strikingly similar to the definition of the so-called \textit{almost quantum correlations}~\cite{AlmostQuantum}: for these, one demands that the operators in a Bell experiment commute \textit{on the state of interest} and not on all quantum states, which gives a relaxation of the set of quantum correlations.

Moreover, Theorem~\ref{ThmRJ} entails that $\mR_J$ is isomorphic to the linear functionals on $\conv\{U_\theta \ketbra{\psi}{\psi} U_\theta^\dagger|\theta \in [0,2\pi)\}$ giving values in $[0,1]$. As discussed in Section~\ref{subsec:orbitope}, this entails that $\conv\{U_\theta \ketbra{\psi}{\psi} U_\theta^\dagger|\theta \in [0,2\pi)\}$ is isomorphic to the orbitope $\Omega_{J}$. This isomorphism gives a characterization of $\Omega_J$ as a spectrahedron.

That rotation boxes represent a relaxation of the quantum formalism can also be seen by noting the following Lemma which later will be contrasted with its quantum counterpart (Lemma~\ref{LemQtrig}):
\begin{lemma}
\label{lemRotationBox}
Let $P(+|\theta):=\sum_{j=-2J}^{2J} a_j e^{ij\theta}$ be a trigonometric polynomial of degree $2J$. Then $P\in\mathcal{R}_J$ if and only if there exist positive semidefinite $(2J+1)\times (2J+1)$-matrices $Q,S\geq 0$ such that
\begin{itemize}
	\item $a_k=\sum_{0\leq j,j+k\leq 2J} Q_{j,j+k}$,
	\item $1-a_0=\Tr(S)$,
	\item $a_k=-\sum_{0\leq j,j+k\leq 2J} S_{j,j+k}$ for all $k\neq 0$.
\end{itemize}
\end{lemma}
The first condition implies that $0\leq P(+|\theta)$ for all $\theta\in\mathbb{R}$, and the last two constraints guarantee that $P(+|\theta) \leq 1$ for all $\theta \in \mathbb{R}$. The proof of this lemma is a straightforward application of Lemma~\ref{LemFRtrig} and can be found in Appendix~\ref{app:lemRotationBox}.

Remarkably, the constraints in \Cref{lemRotationBox} can be adapted into a \textit{semidefinite program} (SDP)~\cite{Blekherman}. For instance, imagine we want to find the boundary of the coefficient space of spin-$J$ rotation boxes in some direction $\mathbf{n}\in\mathbb{R}^{4J+1}$ of the trigonometric coefficients space. That is, we want to find the maximal value of $f(\mathbf{c},\mathbf{s})=\mathbf{n}\cdot (\mathbf{c},\mathbf{s})^\top$, where $\mathbf{c}, \mathbf{s}\in\mathbb{R}^{2J+1}$ are vectors $\mathbf{c}=(c_0,\ldots,c_{2J})$, $\mathbf{s}=(s_1,\ldots,s_{2J})$ collecting the trigonometric coefficients leading to valid rotation boxes. Then, one can pose the following SDP:
\begin{equation}\label{rotationBoxSDP}
\begin{aligned}
\max_{Q,S} \quad & f(\mathbf{c},\mathbf{s})\\
\textrm{s.t.} \quad & \bullet\; a_k=\sum_{0\leq j,j+k\leq 2J}Q_{j,j+k} \, \text{ for all } k,\\
& \bullet\; a_k=-\sum_{0\leq j,j+k\leq 2J}S_{j,j+k} \, \text{ for all } k\neq 0,\\
& \bullet\; 1-a_0=\mathrm{Tr}(S),\\
  & \bullet\; Q,S\geq 0 ,
\end{aligned}
\end{equation}
where the entries of $Q, S$ are labelled from $0$ to $2J$. For example, for $J=1$ the first condition above becomes
\begin{eqnarray*}
a_{-2}&=& Q_{2,0}\\
a_{-1}&=& Q_{1,0}+Q_{2,1}\\
a_0&=& Q_{0,0}+Q_{1,1}+Q_{2,2}\\
a_1 &=& Q_{0,1}+Q_{1,2}\\
a_2&=& Q_{0,2}.
\end{eqnarray*}
As we show in Appendix~\ref{app:SDPgralOutcomes}, the SDP formulation in~(\ref{rotationBoxSDP}) can be easily generalized to account for an arbitrary finite number of outcomes, i.e.\ for the analysis of $\mathcal{R}_J^{\mathcal{A}}$ with $|\mathcal{A}|\geq 3$. In \Cref{secMainResults} we use the SDP methodology in~(\ref{rotationBoxSDP}) to efficiently derive hyperplanes that bound the set of spin-$J$ rotation boxes (and thus also the set of spin-$J$ quantum boxes). These hyperplanes can be treated as inequalities which, if violated, ensure that the system being probed has spin larger than the $J$ considered. 

Suppose now that we are not interested in optimizing some quantity restricted to $\mathcal{R}_J$, but rather we are given a list of coefficients $\mathbf{\tilde{a}}$ (perhaps by an experimentalist) and we want to know whether these lead to a valid spin-$J$ correlation. Then, one can recast the SDP formulation as a \textit{feasibility problem} (see, e.g.,~\cite{skrzypczyk2023semidefinite}) by setting the given coefficients as constraints. That is, we are now interested in the following problem:
\begin{equation}\label{feasibilitySDP}
\begin{aligned}
\text{find} \quad & Q \text{ and } S\\
\textrm{s.t.} \quad & \bullet\; \tilde{a}_k=\sum_{0\leq j,j+k\leq 2J}Q_{j,j+k} \, \text{ for all } k,\\
& \bullet\; \tilde{a}_k=-\sum_{0\leq j,j+k\leq 2J}S_{j,j+k} \, \text{ for all } k\neq 0,\\
& \bullet\; 1-\tilde{a}_0=\mathrm{Tr}(S),\\
  & \bullet\; Q,S\geq 0 ,
\end{aligned}
\end{equation}
where, contrary to~(\ref{rotationBoxSDP}), the coefficients $\tilde{a}_k$ are now fixed. If the SDP is feasible, then it will give $(2J+1)\times(2J+1)$ matrices $Q,S\geq 0$ certifying that $\mathbf{\tilde{a}}$ leads to a valid spin-$J$ correlation (c.f.\ \Cref{lemRotationBox}). Conversely, if the SDP is infeasible, then one can obtain a certificate that the given coefficients $\mathbf{\tilde{a}}$ cannot lead to a valid spin-$J$ correlation (again see, e.g.,~\cite{skrzypczyk2023semidefinite}). 

We have already noted above that there is a conceptual similarity between general spin-$J$ correlations (as a relaxation of quantum spin-$J$ correlations) and ``almost quantum'' Bell correlations~\cite{AlmostQuantum} (as a relaxation of the quantum Bell correlations). Here we see another aspect of this analogy: the set of almost-quantum Bell correlations has an efficient SDP characterization (derived from the NPA hierarchy~\cite{NPA}), but the set of quantum correlations does not. Similarly, as shown above, general spin-$J$ correlations have an efficient SDP characterization, but we do not know whether quantum spin-$J$ correlations $\mathcal{Q}_J^{\mathcal{A}}$ have an SDP characterization, for arbitrary $J$ and $\mathcal{A}$.

In particular, the quantum counterpart of \Cref{lemRotationBox} is the following:
\begin{lemma}
\label{LemQtrig}
Let $P(+|\theta):=\sum_{j=-2J}^{2J} a_j e^{ij\theta}$ be a trigonometric polynomial of degree $2J$. Then $P\in\mathcal{Q}_J$ if and only if there exists a positive semidefinite $(2J+1)\times (2J+1)$-matrix $Q\geq 0$ such that
\begin{itemize}
	\item $a_k=\sum_{0\leq j,j+k\leq 2J} Q_{j,j+k}$,
	\item $Q$ is the Schur product of a density matrix and a POVM element, i.e.\ there exist $0\leq E \leq \mathbf{1}$ and $0\leq\rho$ with ${\rm Tr}(\rho)=1$ such that $Q_{i,j}=E_{i,j}\rho_{i,j}$ (denoted $Q=E\circ\rho$).
\end{itemize}
\end{lemma}
The proof follows directly from Theorem~\ref{thm:QJ_corr_form} and the Born rule, $P(+|\theta)={\rm Tr}(\rho U_\theta^\dagger E^\top U_\theta)$. Note that the second condition, the Schur product of $\rho,E\geq 0$, breaks the linearity required for an SDP formulation in the general case where both $\rho, E$ act as free optimizing variables. Nonetheless, for numerical purposes, one may be interested in circumventing this limitation by adopting a see-saw scheme~\cite{pal2010maximal,werner2001bell} at the cost of introducing local minima in the optimization problem. The see-saw methodology consists in linearizing the problem by fixing one of the free variables and optimizing only over the other free variable. Then, fix the obtained result and optimize over the variable that had been previously fixed. One would iteratively continue this procedure until the objective function converges to a desired numerical accuracy.

For example, in our case, one could start by picking a random quantum state $\rho$ and use an SDP with the conditions in Lemma~\ref{LemQtrig} to find the optimal POVM $E$ for that given $\rho$. Then, fix the POVM to the new-found $E$ and proceed to optimize using $\rho$ as a free variable in order to update the quantum state to a new more optimal value. One would continue this procedure until eventually the increment gained at each iteration would be negligible. However, as opposed to a general SDP, this approach does not guarantee that a global minimum has been attained due to the possible presence of local minima. To guarantee that a global minimum has been obtained, one has to provide a certificate of optimality (for instance, by means of the \textit{complementary slackness} theorem~\cite{Blekherman}).

\section{Rotation boxes in the prepare-and-measure scenario}
\label{secMainResults}
So far, we have defined quantum and more general spin-$J$ correlations, $\mathcal{Q}_J^\mathcal{A}$ and $\mathcal{R}_J^\mathcal{A}$, describing how outcome probabilities can respond to the spatial rotation of the preparation device in a prepare-and-measure scenario. But how are these two sets related? Do they agree or is there a gap? Can all possible continuous functions $P(+|\theta)$ be realized for large $J$? What can we say in the special case of restricting to two possible input angles only, and what is the correct definition of a ``classical'' rotation box? In this section, we answer all these questions, and we  review earlier work by some of us~\cite{Jones}, which shows how the results can be applied to construct a theory-agnostic semi-device-independent randomness generator.

\subsection{$\mathcal{Q}_0^{\mathcal{A}}=\mathcal{R}_0^{\mathcal{A}}$ and $\mathcal{Q}_{1/2}^{\mathcal{A}}=\mathcal{R}_{1/2}^{\mathcal{A}}$}
\label{Subsec012}
In this subsection, we will see that all the spin-$J$ correlations for $J=0$ and $J=1/2$ have a quantum realization. That is, for every $P\in\mathcal{R}^\mathcal{A}_{0}$ (resp.\ $P\in\mathcal{R}^{\mathcal{A}}_{1/2}$), we can find a spin-$0$ (resp.\ spin-$1/2$) quantum system, a quantum state $\rho$, and a POVM $\{E_a\}_{a\in\mathcal{A}}$ such that $P(a|\theta)=\mbox{tr}(U_\theta \rho U_\theta^\dagger E_a)$.

First, we consider $J=0$. In this case the set of rotation boxes corresponds to all sets with cardinality $|\mathcal{A}|$ of constant functions between zero and one summing to one, i.e.\ $P\in\mathcal{R}_0^\mathcal{A}$ is given by $P(a|\theta)=c_a$ for all $\theta\in[0,2\pi)$, where $0\leq c_a\leq 1$ and $\sum_{a=1}^{|\mathcal{A}|} c_a=1$. In the quantum case, we consider a representation $U_\theta$ of SO(2) consisting of the direct sum of $|\mathcal{A}|$ copies of the trivial representation, i.e. $U_\theta=\mathds{1}_{|\mathcal{A}|}$. Now, to realize $P\in \mathcal R_0^{\mathcal{A}}$, we pick an orthonormal basis $\{\phi_a\}_{a=1}^{|\mathcal{A}|}$ and construct the state $\ket{\psi}=\sum_{a=1}^{|\mathcal{A}|}\sqrt{c_a}\ket{\phi_a}$, such that $P(a|\theta)=|\braket{\phi_a}{U_\theta\psi}|^2=|\braket{\phi_a}{\mathds{1}_{|\mathcal{A}|} \psi}|^2=c_a$ for every $a\in \mathcal{A}$ and therefore $\mathcal{Q}_0^{\mathcal{A}}=\mathcal{R}_0^{\mathcal{A}}$.

Next, we will turn our attention to the first non-trivial case, i.e.\ to $J=1/2$. 
\begin{theorem}
    The  correlation set $\mR^{\mathcal{A}}_{1/2}$ is equal to $\mQ^{\mathcal{A}}_{1/2}$, i.e.\ $\mQ^{\mathcal{A}}_{1/2} = \mR^{\mathcal{A}}_{1/2}$.
\end{theorem}
\begin{proof}
We recall (see \Cref{def:rot_box_system}) that the state space of the GPT system $\tR_{1/2}$ generating $\mR^{\mathcal{A}}_{1/2}$ is given by \begin{equation}
    \Omega_{1/2}:= \mbox{conv}\left\{\begin{pmatrix}
   1& \cos(\theta) & \sin(\theta) \end{pmatrix}^\top\,\,|\,\, \theta\in[0,2\pi)\right\},
\end{equation}
and that $\tR_{1/2}$ is unrestricted. Next, we will show that  the state space $\Omega_{1/2}$ can be identified with the state space of a rebit, which follows from the fact that every pure rebit state $\rho\in \mD(\reals^2)\subset \LS(\reals^2)$, where $\LS(\reals^2)$ is the space of real symmetric $2\times 2$- matrices, can be written as \begin{equation}
        \rho=\frac{1}{2}\left(\mathds{1}+\cos(\theta)\sigma_x+ \sin(\theta)\sigma_z \right),
    \end{equation}
with the Pauli matrices $\sigma_x$ and $\sigma_z$. Hence, we define the bijective linear map $L:\reals^3\rightarrow \mathcal{L}_S(\reals^2)$ by\begin{eqnarray}
        \begin{pmatrix}
            r_0\\r_1\\r_3
        \end{pmatrix}\mapsto\frac{1}{2}(r_0\mathds{1}+r_1\sigma_x+r_3\sigma_z).
    \end{eqnarray}
Since $\tR_{1/2}$ and the rebit are both unrestricted~\cite{Wright2021}, we can map the effects of $\tR_{1/2}$ one to one to the effects of the rebit via the map $(L^{-1})^*:(\reals^3)^*\rightarrow(\mathcal{L}_S(\reals^2))^*$. Furthermore, the system $\tR_{1/2}$ carries the representation $T_\theta$:
\begin{equation}
         T_\theta= \begin{pmatrix}
        1 & 0 & 0\\
        0& \cos(\theta)& -\sin(\theta)\\
        0&\-\sin(\theta)&\cos(\theta)
    \end{pmatrix}.
    \end{equation} 
    Using the map $L$ again, we can define the SO(2)-representation $U$ on the rebit by $U[\theta]=LT_\theta L^{-1}$. Applied to $\rho\in \mathcal{L}_s(\reals^2)$, this family of transformations acts as\begin{eqnarray}
        U[\theta](\rho)=U_\theta\rho U^\dagger_\theta,
    \end{eqnarray} 
    where \begin{eqnarray}
        U_\theta=\exp(i\frac{\theta}{2}\sigma_y)=\begin{pmatrix}
        \cos(\frac{\theta}{2}) & \sin(\frac{\theta}{2})\\
         -\sin(\frac{\theta}{2}) & \cos(\frac{\theta}{2})
        \end{pmatrix}. 
    \end{eqnarray} 
    Now, let $P\in \mathcal{R}^\mathcal{A}_{1/2}$ and let $\omega\in\Omega_{1/2}$ and $\{e_a\}^{|\mathcal{A}|}_{a=1}\subset \mathcal{E}_{1/2}$ be the state and measurement generating $P$. We show \begin{eqnarray}
        P(a|\theta)&=&(e_a,T_\theta \omega)_{\reals^3}=(e_a,L^{-1}LT_\theta L^{-1}L\omega)_{\reals^3}\nonumber\\
        &=&\langle(L^{-1})^*e_a,LT_\theta L^{-1}L\omega\rangle_{\mathcal{HS}}=\langle E_a, U[\theta] (\omega')\rangle_{\mathcal{HS}}\nonumber\\
        &=&\mbox{Tr}(E_a U_\theta\omega' U^\dagger_\theta),
    \end{eqnarray} 
    where $(\cdot,\cdot)_{\reals^3}$ and $\langle\cdot,\cdot\rangle_\mathcal{HS}$ denote the standard inner product in $\reals^3$ and Hilbert-Schmidt product, respectively, and the $E_a=(L^{-1})^*e_a$ and $\omega'=L\omega$ are a rebit effect and a rebit state, respectively.
\end{proof}

For a characterization of the extreme points of $\mathcal{R}_{1/2}$, see~\cite{Jones} and \Cref{fig:spinOneHalfSDP} above.

\subsection{The convex structure of $\mathcal{R}_1$ and $\mathcal{Q}_1=\mathcal{R}_1$}\label{sec:R1}

For clarity, we write the general form of spin-$J$ correlations of Definition~\ref{def:spin_J_gen_cor} in the case $J = 1$. The set $\mR_1$ of correlations generated by spin-$1$ rotation boxes consists of all probability distributions $P(+|\bullet): \mathbb{R}\to\mathbb{R}$ of the following form:
    \begin{align}\label{eq:R1def}
    P(+|\theta) = c_0 + c_1 \cos\theta + s_1 \sin\theta + c_2 \cos(2 \theta) + s_2 \sin(2 \theta) , 
\end{align}
where $c_0,c_1,s_1,c_2,s_2\in\mathbb{R}$ and $0\leq P(+|\theta)\leq 1$ for all $\theta$.

\subsubsection{Characterizing the facial structure of $\mR_1$}

We now characterize some of the properties of the convex set $\mR_1$. Our main goal is to characterize the extreme points of $\mathcal{R}_1$, which will then allow us to obtain explicit quantum realizations of these extreme points and hence of all of $\mathcal{R}_1$. For $\theta_0 \in [0, 2 \pi)$ we define the following face of $\mR_1$:
\begin{align}
    F_{\theta_0} := \{P \in \mR_1\,\,|\,\, P(\theta_0) = 0\} .
\end{align}
The condition $P(\theta_0)= 0$ defines a hyperplane in the space of coefficients $(c_0,c_1,s_1,c_2,s_2)\in\mathbb{R}^5$. Since it is a supporting hyperplane of $\mathcal{R}_1$, its intersection with this compact convex set is a face. For some background on convex sets, their faces, and other convex geometry notions used in this section, see e.g.\ the book by Webster~\cite{Webster}.

\begin{figure*}[t]
\centering 
\includegraphics[width=0.329\linewidth]{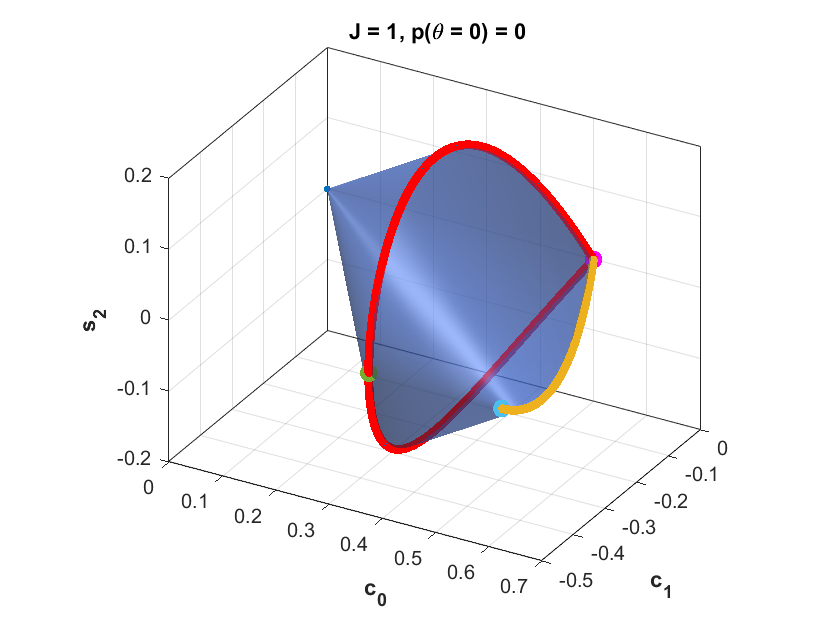}
\includegraphics[width=0.329\linewidth]{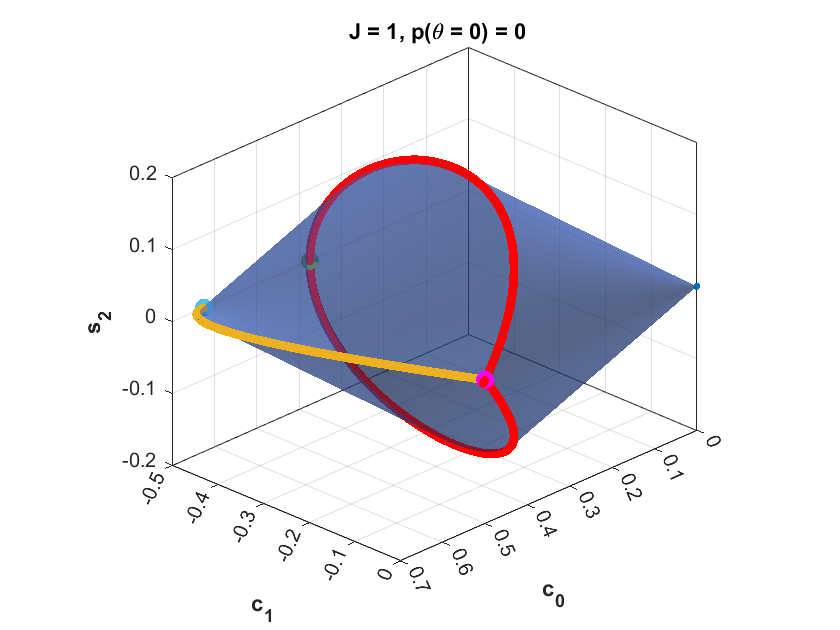}
\includegraphics[width=0.329\linewidth]{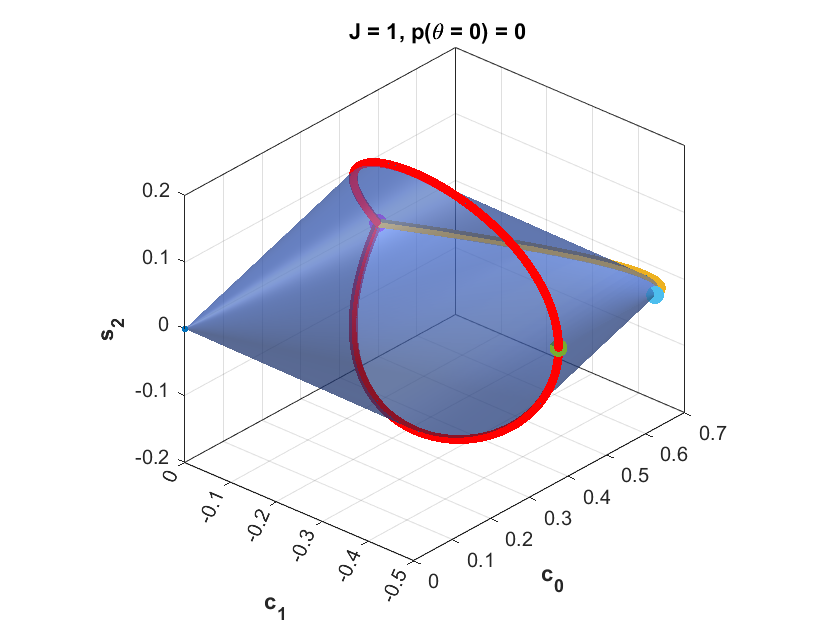}
\caption{Different perspectives of the set containing the associated trigonometric coefficients of the face $F_0$ of the binary spin-$1$ correlations $\mathcal{R}_1$, and its extremal points from \Cref{lem:charac_R1_faces}. The red and yellow lines correspond to the two consecutive extremal points for $F_{0,\theta_1}$ with $\theta_1\in\left(\pi/2,3\pi/2\right)$, the pink dot corresponds to the case $F_{0,\pi/2}=F_{0,3\pi/2}$, and the green and cyan dots correspond to the two consecutive cases for $F_{0,\pi}$.}
\label{fig:extremalsJ1}
\end{figure*}

\begin{lemma}
    The face $F_{\theta_0}$ has dimension $\dim(F_{\theta_0})\leq 3$  for every $\theta_0 \in [0 , 2\pi)$.
\end{lemma}
\begin{proof}
    For every $P \in F_{\theta_0}$ it must be the case that $P(\theta_0)$ is a minimum, since $P(\theta) \geq 0$. This implies that $P'(\theta_0) = \frac{d}{d\theta}P(\theta)|_{\theta=\theta_0} = 0$. Thus, we obtain two linearly independent constraints
\begin{align}\label{eq:Faces3Dconditions}
    P(\theta_0) = 0, \quad    P'(\theta_0) = 0,
\end{align}
and the face $F_{\theta_0}$ is at most three-dimensional.
\end{proof}
For $\theta_0, \theta_1 \in [0,2\pi)$, we define the following subsets of $F_{\theta_0}$:
    \begin{align}
          F_{\theta_0,\theta_1}:=\{P\in\mathcal{R}_1\,\,|\,\, P(\theta_0)=0,P(\theta_1)=1\} . 
    \end{align}
    Every non-empty $F_{\theta_0,\theta_1}$  is a face of $F_{\theta_0}$ and therefore of $\mR_1$ (and thus itself compact and convex). Denote the extremal points of a compact convex set $C$ by $\partial_{\rm ext}C$.

\begin{lemma}\label{lem:non-cst-ext-R1}
    Every non-constant function $P\in\partial_{\rm ext}\mathcal{R}_1$ is contained in at least one face $F_{\theta_0,\theta_1}$.
\end{lemma}
This lemma is proven in Appendix~\ref{app:non-cst-ext-R1}.

If $P$ is extremal in $\mathcal{R}_1$, then it is also extremal in every face in which it is contained. Thus, we can determine the extremal points of $\mathcal{R}_1$ by determining $\partial_{\rm ext} F_{\theta_0,\theta_1}$ (and keeping in mind that the functions which are constant, $P(\theta)=0$ for all $\theta$ and $P(\theta)=1$ for all $\theta$, are also extremal in $\mathcal{R}_1$).

Next, note that it is sufficient to determine the extremal points in the case that $\theta_0=0$. This is because
\begin{align*}
    P(\theta) \in 
 F_{\theta_0,\theta_1} \Leftrightarrow   P(\theta + \theta_0) \in 
 F_{0,\theta_1-\theta_0}  . 
\end{align*}
Hence $F_{\theta_0,\theta_1}$ and $F_{0,\theta_1-\theta_0}$ are related by a linear symmetry $T_{\theta_0}$ of $\mathcal{R}_1$, which is defined by
\[
T_{\theta_0}(P)(\theta):=P(\theta+\theta_0).
\]
That is, $T_{\theta_0}:\mathcal{R}_1\to\mathcal{R}_1$ is a convex-linear map that rotates every rotation box by angle $\theta_0$. Since it is a symmetry of $\mathcal{R}_1$, it maps extremal points of faces to extremal points of faces. To determine $\partial_{\rm ext}F_{\theta_0,\theta_1}$, we only need to ``rotate'' $\partial_{\rm ext}F_{0,\theta_1-\theta_0}$ by $\theta_0$.

We now explicitly characterize the faces $F_{0 , \theta_1}$ by the functions corresponding to their extremal points.
\begin{lemma}\label{lem:charac_R1_faces}
The faces $F_{0,\theta_1}$ for $\theta_1 \in [0,2\pi)$ are characterized as follows:
\begin{enumerate}
    \item  If $\theta_1 \in [0 , \frac{\pi}{2}) \cup (\frac{3 \pi}{2} , 2 \pi)$ , then \[F_{0,\theta_1}=\emptyset  . \]
    \item If  $\theta_1 \in \{\frac \pi 2 , \frac{3\pi}2 \}$, then $F_{0,\theta_1}$ contains a single element:
    \[F_{0,\frac\pi 2}=F_{0,\frac{3\pi}2}=\left\{P(\theta)=\sin^2\theta\right\} .
    \]
    \item If $\theta_1\in\left(\frac \pi 2,\frac{3\pi}2 \right)\setminus\{\pi\}$, then $F_{0,\theta_1}$ contains exactly two distinct extremal points,
\[
\partial_{\rm ext}F_{0,\theta_1}=\{P(\theta),\tilde P(\theta)\},
\]
where
\begin{align*}
     P(\theta) &=c(1-\cos\theta)(1-\cos(\theta-\theta_0')), \\
   \tilde P(\theta) &= 1-P(\theta_1-\theta)  ,
\end{align*}
and  $\theta'_0 = 2 \theta_1 $ for $\theta_1 \in (\frac{\pi}{2}, \pi)$ and $\theta'_0 = 2 ( \theta_1 - \pi) $ for $\theta_1 \in (\pi, \frac{3\pi}{2})$. The parameter $c>0$ is uniquely determined by the condition $\max_\theta P(\theta)=  1$. 
\item If $\theta_1 = \pi$ then the face $F_{0,\pi}$ contains exactly two extremal points, namely
\[
   F_{0,\pi}=\{P(\theta),\tilde P(\theta)\},
\]
where
\begin{align*}
    P(\theta) &=\sin^4\frac\theta 2 ,\\
    \tilde P(\theta)&=  1-P(\theta_1-\theta) = \frac 1 4(1-\cos\theta)(3+\cos\theta)  .
\end{align*}
\end{enumerate}
\end{lemma}
This lemma is proven in Appendix~\ref{app:charac_R1_faces}.

In Figure~\ref{fig:extremalsJ1}, we plot the face $F_0$ in the coefficients space, illustrating the resulting extremal points from \Cref{lem:charac_R1_faces}. Note that from the conditions~(\ref{eq:Faces3Dconditions}) for $\theta_0=0$, one has $c_0=-c_1-c_2$ and $s_1=-s_2$, thus $\dim F_0=3$.

\subsubsection{Quantum realizability of $\mR_1$}

Having characterized the facial structure of $\mR_1$ and its extremal functions, we now ask if this set of correlations can be realized by a quantum spin-$1$ system.

By Theorem~\ref{thm:QJ_corr_form}, the space $\mQ_1$ of $\SO(2)$-correlations generated by a quantum spin-$1$  system is given by the functions   $P(+|\theta) = \bra{\psi} U_\theta^\dagger E_+ U_\theta  \ket \psi$, where $\ket \psi \in \comp^3$, $E_+$ a POVM  element on $\comp^3$, and $U_\theta = e^{i Z \theta}$ with $Z = \diag(1,0,-1)$.

It follows immediately from the convexity of $\mR_1$ and of $\mQ_1$ that it is sufficient to show that the extremal points of $\mR_1$ are quantumly realizable to show that all the correlations in $\mR_1$ are quantumly realizable.
\begin{lemma}
\label{LemExtremalImpliesR1}
    $\delta_{\rm ext} \mR_1 \subseteq \mQ_1$ implies $\mR_1 = \mQ_1$.
\end{lemma}
This will be used to prove the main result of this subsection:
\begin{theorem}[$\mQ_1 = \mR_1$]\label{thm:R1eqQ1}
    The  correlation  set $\mR_1$ is equal to $\mQ_1$.
\end{theorem}
For the proof, see Appendix~\ref{app:proofR1eqQ1}. It follows from constructing explicit quantum spin-$1$ realizations of all the extremal points of $\mathcal{R}_1$ which have been enumerated in Lemma~\ref{lem:charac_R1_faces}.

Although the correlation spaces $\mR_1$  and $\mQ_1$ are equal, the $J=1$ general rotation box system $\tR_1$ (which generates $\mR_1$) is not equivalent to a quantum spin-$1$ system. This can be seen immediately from the fact that $\tR_1$ is a $5$-dimensional GPT system, while a quantum spin-$1$ system is a $9$-dimensional system (since $\dim(\LH(\comp^3)) = 9$). 

In the next section, we will see that these two GPT systems, although they generate equivalent $\SO(2)$-correlations, have distinct informational properties.

\subsubsection{Inequivalence of spin-$1$ rotation box system and quantum system}

Every $P \in \mR_1^\mathcal{A}$ can be decomposed in the following way:
\begin{align}\label{eq:R1decomp}
    P(a|\theta)  &= c_0^{(a)} + c_1^{(a)} \cos\theta + s_1^{(a)} \sin\theta + c_2^{(a)} \cos(2\theta) + s_2^{(a)} \sin(2\theta)  \nonumber \\
    &=\begin{pmatrix}
        c_0^{(a)}, & c_1^{(a)}, & s_1^{(a)}, & c_2^{(a)}, & s_2^{(a)}
    \end{pmatrix}
    \cdot 
    \begin{pmatrix}
        1 \\
        \cos(\theta)\\
        \sin(\theta) \\
        \cos(2 \theta) \\
        \sin(2 \theta) 
    \end{pmatrix} \nonumber \\
    &=  e_a \cdot  \omega(\theta),
\end{align}
where $e_a$ and $ \omega(\theta)$ are an effect and state of the spin-$1$ rotation box system $\tR_1$, as defined in Definition~\ref{def:rot_box_system} for general spin-$J$.

We give an explicit definition of the $\tR_1 = \left( \mathbb{R}^5,\Omega_1, \mE_1  \right)$ GPT system here. The state space $\Omega_1$ is given by:
    \begin{align*}
        \Omega_1 := \conv \left\{\omega(\theta) \,\,|\,\, \theta \in [0,2\pi)] \right\},
    \end{align*}
    where
    \begin{align*}
        \omega(\theta) = \begin{pmatrix}
        1 \\
        \cos(\theta)\\
        \sin(\theta) \\
        \cos(2 \theta) \\
        \sin(2 \theta) 
    \end{pmatrix}.
    \end{align*}
Let $V \simeq \reals^5$ be the real linear span of $\Omega$ and $V^*$ its dual space. The effect space of $\tR_1$ is
\begin{align*}
    \mE_1 := \{e \in V^*\,\,|\,\, 0\leq (e,\omega)\leq 1 \mbox{ for all }\omega\in\Omega_1 \}.
\end{align*}
By definition, $\tR_1$ is an unrestricted GPT. The state space $\Omega_1$ belongs to a family of $\SO(2)$-orbitopes of the form $C_{a,b} := \conv\{( 1, \cos(a \theta), \sin(a \theta), \cos(b \theta) , \sin(b \theta)\,\,|\,\,\theta \in [0 ,2\pi) \}$ for integers $a <b$. The facial structure of these orbitopes was studied in~\cite{smilansky_convex_1985}. They are a subset of the Carath\'eodory orbitopes defined in Section~\ref{subsec:orbitope}. The $\SO(2)$ reversible transformations are given by
\begin{equation}
   T(\theta) =  \begin{pmatrix}
        1 & 0& 0& 0&0   \\
        0 & \cos(\theta) & -\sin(\theta) &0 & 0\\
        0  & \sin(\theta) & \cos(\theta)  &0 & 0\\
        0  &     0         &     0          &  \cos(2\theta) & -\sin(2\theta)  \\
       0   &      0        &      0         &  \sin(2\theta) & \cos(2\theta)
    \end{pmatrix} . \label{eq:R1_SO2_rep}
\end{equation}
\begin{lemma}
    The effect space $\mE_1$ is isomorphic (as a convex set) to $\mR_1$, i.e.\ there is an invertible linear map that maps one of these sets onto the other.
\end{lemma}
\begin{proof}
    The effect space $\mE_1$ consists of all $(c_0,c_1,s_1,c_2,s_2) \in \reals^5$ such that Eq.~\eqref{eq:R1decomp} is in $[0,1]$ for all $\theta \in [0, 2 \pi)$. This is equivalent to the condition that $P(+|\theta)\in[0,1]$ for all $\theta$ which defines $\mR_1$ in Eq.~\eqref{eq:R1def}.
\end{proof}

We now describe some informational properties of $\tR_1$:
\begin{lemma}[Properties of $\tR_1$]\label{lem:R1prop}
    The GPT system $\tR_1$
    \begin{enumerate}
        \item   has three jointly perfectly distinguishable states and no more;
        \item  has four pairwise perfectly distinguishable states;
        \item violates bit symmetry.
    \end{enumerate}
\end{lemma}
This lemma is proven in Appendix~\ref{app:R1prop}.

Bit symmetry is the property that any pair of perfectly distinguishable pure states $(\omega_0,\omega_1)$ of a GPT system can be reversibly mapped to any other pair of perfectly distinguishable pure states $(\omega_0',\omega_1')$ of that system~\cite{PhysRevLett.108.130401}. Namely, there exists a reversible transformation $T$ such that $(\omega_0',\omega_1') = (T \omega_0, T \omega_1)$. 

We note that $\tR_1$ violates bit symmetry not just for the set of $\SO(2)$ reversible transformations but for the set of all symmetries. This set is larger than the $\SO(2)$ transformations of \cref{eq:R1_SO2_rep} and includes the transformation $\diag(1,1,-1,1,-1)$ which is not of the form $T(\theta)$. 

Considering the full set of symmetries is important when contrasting to a qutrit, since the qutrit when restricted to the spin-1 $\SO(2)$-transformations violates bit symmetry, but it obeys bit symmetry when considering the full symmetry group $\SU(3)$.

Although the space of correlations $\mR_1 \cong \mE_1$, the GPT system $\tR_1$ contains additional structure, namely in its state space $\Omega_1$.  Hence, although every $P(+|\theta) \in \mR_1$ can be generated using a quantum system $\mQ_1$, this does not imply that every information-theoretic game carried out using the system $\tR_1$ can be equally successfully carried out with a spin-$1$ quantum system. For instance, a game which required one to encode a pair of bits $(i,j) \in \{0,1\}^2$ in four states of a GPT system such that one could perfectly decode either the first bit or the second bit can be implemented with $\tR_1$ with $100\%$ success probability, but will necessarily have some error when implemented on a quantum spin-$1$ system.

A key difference between the the GPT system $\tR_1$ and the $\SO(2)$ quantum spin-$1$ system $\tQ_1$ (i.e.\ a qutrit with dynamics restricted to $U_\theta = e^{i Z \theta}$) is that inequivalent $\SO(2)$-orbits of pure states of the qutrit are needed to generate $\mR_1$, whilst a single $\SO(2)$-orbit of states $\{\omega(\theta)\,\,|\,\, \theta \in [0 ,2\pi) \}$ of $\tR_1$ is needed to generate $\mR_1$.

A formal way to understand the equivalences and inequivalences of $\tR_J$ and $\tQ_J$ for different values of $J$ is in terms of linear \textit{embeddings}~\cite{MuellerGarner}. We say that a GPT $A=(V_A,\Omega_A,E_A)$ can be embedded into a GPT $B=(V_B,\Omega_B,E_B)$ if there is a pair of linear maps $\Phi,\Psi$ such that $\Psi(\Omega_A)\subset \Omega_B$ and $\Phi(E_A)\subset E_B$ which reproduces all probabilities, $(\Phi(e_A),\Psi(\omega_A))=(e_A,\omega_A)$ for all  $e_A\in E_A,\omega_A\in\Omega_A$. As argued in~\cite{MuellerGarner}, this means that $B$ can simulate the GPT $A$ ``univalently'', i.e.\ in a way that generalizes the concept of noncontextuality for simulations by classical physics.

In the proof that $\mathcal{Q}_{1/2}^\mathcal{A}=\mathcal{R}_{1/2}^\mathcal{A}$ in \Cref{Subsec012}, we have used the fact that the spin-$1/2$ GPT system $\tR_{1/2}$ (the rebit) can be embedded into the qubit $\tQ_{1/2}$, seen as a quantum spin-$1/2$ system. Moreover, it can be done in a way such that the orbit $\theta\mapsto \omega(\theta)$ is mapped to an orbit $\rho(\theta)=\Psi(\omega(\theta))$. That is, the quantum system can reproduce the full probabilistic behavior of the general spin-$1/2$ system.

However, it is easy to see that no such embedding can exist for the case of $J=1$. If we had such a pair of linear maps, and if it mapped the orbit $\omega(\theta)$ to some orbit $\rho(\theta)$, then it could not reproduce all probabilities: it would give us four states $\rho(0),\rho(\frac{\pi}{2}),\rho(\pi),\rho(\frac{3 \pi}{2})$ of the qutrit which are pairwise perfectly distinguishable. But no four pairwise orthogonal states can exist on a qutrit. Clearly, the converse is also true: The spin-1 quantum system $\tQ_1$ spans the vector space $\LH(\comp^3) \simeq \reals^9$ and hence cannot be embedded in the GPT system $\tR_1$ which spans $\reals^5$. More generally, we can say the following:
\begin{lemma}
The spin-$1$ GPT system $\tR_1$ cannot be embedded into any finite-dimensional quantum system.
\end{lemma}
\begin{proof}
According to Theorem 2 of~\cite{MuellerGarner}, all unrestricted GPTs that can be so embedded are special Euclidean Jordan algebras. For all such systems, the numbers of jointly and pairwise perfectly distinguishable states coincide. This can be seen e.g.\ by noting that perfectly distinguishable pure states in Euclidean Jordan algebras are orthogonal (with respect to the self-dualizing inner product) idempotents (see e.g.~\cite[Lemma 3.3]{BarnumHilgert}), and pairwise orthogonality implies that they are elements of a Jordan frame and hence jointly perfectly distinguishable. But as we have shown in \Cref{lem:R1prop} above, this correspondence does not hold for $\tR_1$.
\end{proof}
Hence, even though the set of spin correlations $\mathcal{R}_1$ and $\mathcal{Q}_1$ agree, the corresponding GPT systems have genuinely different information-theoretic and physical behaviors. This is also the reason why we do not currently know whether $\mathcal{Q}_1^\mathcal{A}=\mathcal{R}_1^\mathcal{A}$ for $|\mathcal{A}|\geq 3$.

\subsection{$\mathcal{Q}_J\subsetneq \mathcal{R}_J$ for $J\geq 3/2$}
\label{SubsecGap}

Up until now we have seen that for $J\leq 1$ an equivalence holds between the correlation sets $\mathcal{Q}_J$ and $\mathcal{R}_J$. However, in this section we show that this equivalence breaks for $J\geq 3/2$. We split the analysis in two parts: First, we provide an explicit counterexample of a spin-$J$ correlation outside of the quantum set for $J=3/2$; Second, we use the same methodology to show that a non-empty gap exists between both sets for any $J\geq 3/2$.

\subsubsection{$\mathcal{Q}_{3/2}\subsetneq \mathcal{R}_{3/2}$}

We start by showing that $\mathcal{Q}_{3/2}\subsetneq \mathcal{R}_{3/2}$. Every spin-$3/2$ correlation can be expressed as a degree-$3$ trigonometric polynomial:
\begin{eqnarray}
P(\theta)&=&c_0+c_1 \cos\theta + s_1 \sin\theta + c_2 \cos(2\theta)+s_2 \sin(2\theta)\nonumber\\
   && + c_3\cos(3\theta)+s_3 \sin(3\theta),\label{eqRepThreeHalves}
\end{eqnarray}
where the coefficients $c_i$ and $s_i$ are suitable real numbers such that $0\leq P(\theta)\leq 1$ for all $\theta$. To show that there exist correlations $P\in\mathcal{R}_{3/2}$ which are not contained in $\mathcal{Q}_{3/2}$, we construct an inequality that is satisfied by all quantum boxes, but violated by some $P^\star\in\mathcal{R}_{3/2}$. In particular, we show the following:
\begin{theorem}
\label{thm:spinThreeHalvesInequality}
If $P\in\mathcal{Q}_{3/2}$, then its trigonometric coefficients, as taken from representation~(\ref{eqRepThreeHalves}), satisfy
\[
   c_2+s_3\leq \frac 1 {\sqrt{3}}\lesssim 0.5774.
\]
On the other hand, the trigonometric polynomial
\[
   P^\star(\theta):=\frac 2 5 + \frac 1 4 \sin\theta + \frac 7 {20} \cos(2\theta)+\frac 1 4 \sin(3\theta)
\]
satisfies $0\leq P^\star(\theta)\leq 1$ for all $\theta$, hence $P^\star\in\mathcal{R}_{3/2}$, but $c_2+s_3=0.6$, i.e.\ $P^\star\not\in\mathcal{Q}_{3/2}$. Therefore, $\mathcal{Q}_{3/2}\subsetneq\mathcal{R}_{3/2}$.
\end{theorem}
\begin{figure}
\centering 
\includegraphics[width=1\columnwidth]{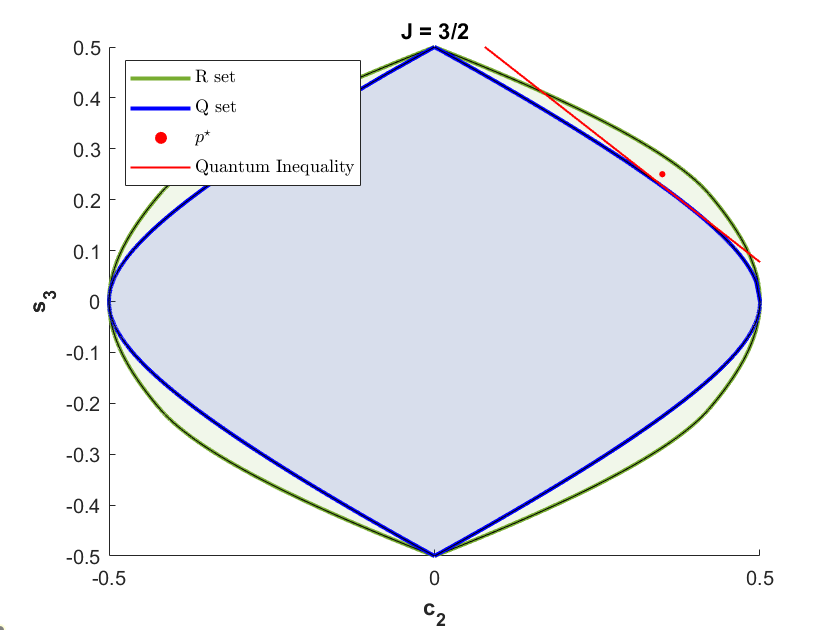}
\caption{Spin-$3/2$ rotation and quantum correlations sets in the $c_2$-$s_3$ plane projection illustrating $\mathcal{Q}_{3/2}\subsetneq \mathcal{R}_{3/2}$. The inequality corresponds to the case that saturates \Cref{thm:spinThreeHalvesInequality}, \textit{i.e.}, $c_2+s_3=1/\sqrt{3}$. The boundary of the 2D projections for the sets $\mathcal{Q}_{3/2}$ (blue) and $\mathcal{R}_{3/2}$ (green) have been numerically obtained using the SDP methodology presented in \Cref{app:SDPboundary}. The quantum inequality (red line) and validity of the rotation box (red dot) $P^\star \in \mathcal{R}_{3/2}$ but $P^\star \notin \mathcal{Q}_{3/2}$ are analytically proven in the main text.}
\label{fig:2DshadowJthreeHalves}
\end{figure}
Clearly, this also implies that $\mathcal{Q}_{3/2}^\mathcal{A}\subsetneq \mathcal{R}_{3/2}^\mathcal{A}$ for three or more outcomes, $k:=|\mathcal{A}|\geq 3$, since $P^\star$ can always appear as the probability of the first of the $k$ outcomes.

In the remainder of this section, we prove this theorem by solving the optimization problem
\begin{equation}
   \beta:=\max_{P\in\mathcal{Q}_{3/2}} (c_2+s_3)[P],
   \label{eqQuantumBound}
\end{equation}
and show that the quantum bound is $\beta=\frac 1 {\sqrt{3}}$. Since $(c_2+s_3)[P^\star]=\frac 3 5$, $P^\star$ violates the inequality, thus proving $\mathcal{Q}_{3/2}\subsetneq \mathcal{R}_{3/2}$. For the sake of completion, by adapting the SDP in \cref{rotationBoxSDP} one can show that the maximal value attainable with rotation boxes is $\beta_{\mathcal{R}}=\max_{P\in\mathcal{R}_{3/2}} (c_2+s_3)[P]= \frac{5}{8} = 0.625$, hence $\beta < (c_2+s_3)[P^\star] < \beta_{\mathcal{R}} $. In \Cref{fig:2DshadowJthreeHalves} we illustrate \Cref{thm:spinThreeHalvesInequality} by showing the 2D projection of the correlation sets onto the $c_2$-$s_3$ plane and plotting the inequality given by $c_2+s_3\leq 1/\sqrt{3}$ as well as the point $P^\star$ violating it.   

Suppose that there exists a quantum realization $P\in\mathcal{Q}_{3/2}$, i.e.\ that there exist a POVM element $0\leq E\leq\mathbf{1}$ and a quantum state $\rho$ such that $P(\theta)={\rm Tr}(E^\top U_\theta \rho U_\theta^\dagger)$ (the transpose on $E$ is not necessary, but is used by convention to relate to the Schur product in \Cref{LemQtrig}). Following \Cref{LemQtrig}, then one has
\begin{eqnarray}
(c_2+s_3)[P]&=&2\,{\rm Re}(a_2[P])-2\,{\rm Im}(a_3[P])\nonumber\\
&=& 2\,{\rm Re}(Q_{02}+Q_{13})-2\,{\rm Im}(Q_{03})\\
&=& 2\,{\rm Re}(E_{02}\rho_{02}+E_{13}\rho_{13})-2\,{\rm Im}(E_{03}\rho_{03}) \nonumber \\
&=& {\rm Tr}(M[E]\rho), \nonumber
\end{eqnarray}
where
\[
M[E]:=\left(\begin{array}{cccc} 0 & 0 & E_{20} & -i E_{30} \\ 0 & 0 & 0 & E_{31} \\ E_{02} & 0 & 0 & 0 \\ i E_{03} & E_{13} & 0 & 0 \end{array}\right).
\]
Maximizing this over $\rho$ yields the largest eigenvalue of $M[E]$, see e.g.~\cite{Bhatia}. We determine this eigenvalue in \Cref{app:proofBhatia}, and the result is as follows:
\begin{lemma}
\label{LemEigenvalue}
The quantum bound of Eq.~(\ref{eqQuantumBound}) satisfies
\begin{eqnarray*}
    2\beta^2&=&\max_E \left(\strut |E_{20}|^2+|E_{30}|^2+|E_{31}|^2\right.+\\
    &&\left.+\sqrt{(|E_{20}|^2+|E_{30}|^2+|E_{31}|^2)^2-4|E_{20}|^2 |E_{31}|^2}\right),
\end{eqnarray*}
where the maximization is over all POVM elements $0\leq E\leq \mathbf{1}$ or, equivalently, over all orthogonal projectors $E=E^\dagger=E^2$ on $\mathbb{C}^4$.
\end{lemma}
Matrix entries of orthogonal projectors satisfy certain inequalities as described, for example, in~\cite{Baksalary}. There, it is shown that $|E_{20}|^2+|E_{30}|^2\leq \frac 1 4$, $|E_{30}|^2+|E_{31}|^2\leq \frac 1 4$, and thus
\begin{equation}
   2\beta^2\leq\max_{{{x,y,z\geq 0,}\atop{x+y\leq 1/4,}}\atop{y+z\leq 1/4}} \left(
      x+y+z+\sqrt{(x+y+z)^2-4xz}
   \right).
   \label{eqMaxOverPolytope}
\end{equation}
The maximum is here over a polytope in three dimensions, and we perform the corresponding optimization in \Cref{app:maxPolytope}. We find that the maximum equals $2/3$, and thus $\beta\leq 1/\sqrt{3}$. In \Cref{app:maxPolytope}, we also provide an explicit POVM element $E$ and quantum state $\rho$ saturating this bound, hence $\beta=1/\sqrt{3}$. Furthermore, since $\beta=1/\sqrt{3}<(c_2+s_3)[P^*]=3/5$, we have shown that $P^*\in\mathcal{R}_{3/2}$ lies outside of $\mathcal{Q}_{3/2}$ and, therefore, $\mathcal{Q}_{3/2}\subsetneq \mathcal{R}_{3/2}$. See \Cref{fig:2DshadowJthreeHalves}, where we plot $P^*$ for a visual illustration of this result. This proves Theorem~\ref{thm:spinThreeHalvesInequality}.

\subsubsection{$\mathcal{Q}_J\subsetneq \mathcal{R}_J$ for $J\geq 2$}

In order to show that $\mathcal{Q}_J \subsetneq \mathcal{R}_J$ for any $J \geq 2$, one can easily generalize the inequality from the previous section to the following one:
\[
  P\in\mathcal{Q}_{J}\quad\Longrightarrow\quad  (c_{2J-1}+s_{2J})[P]\leq \beta = \frac{1}{\sqrt{3}}.
\]
See the proof in \Cref{app:quantumBoundGralJ}.

Therefore, we now want to find a spin-$J$ correlation $P^\star_J \in \mathcal{R}_{J}$ such that this inequality is violated for $J\geq 2$, i.e., $(c_{2J-1}+s_{2J})[P^\star_J]>\frac{1}{\sqrt{3}}$. For instance, an educated guess motivated by numerical results is the following trigonometric polynomial:
\[
  P^\star_J(\theta):=\sum_{k=-2J}^{2J}  a_k e^{ik\theta},
\]
with $a_{-k}=\overline{a_k}$, $a_0=\frac 1 2$, $a_{2J}=-\frac i 8$, and
\begin{eqnarray*}
    a_{2J-1-2m}&=&\frac 3 {16}\left(-\frac 1 4\right)^{m}\qquad m=0,\ldots,\lfloor J-1\rfloor,\\
    a_{2J-2-2l}&=&-\frac{3i}{32} \left(-\frac 1 4\right)^l \qquad l=0,\ldots,\lceil J-2\rceil.
\end{eqnarray*}
Indeed, this trigonometric polynomial has $s_{2J}+c_{2J-1}=5/8 > \beta$, thus violating the quantum bound of the inequality above. Furthermore, in \Cref{SubsecExamplesSpin2andHigher}, we show that this trigonometric polynomial satisfies $0\leq P_{J}^\star(\theta) \leq 1$ for $J\geq 7/2$ and, thus, it is a valid rotation box probability distribution for $J\geq 7/2$ which lies outside of the quantum set. However, for values of $J\leq 3$, the trigonometric polynomial $P^*_J(\theta)$ is not a probability distribution. 
The way in which we deal with the remaining cases $J\in \{2,5/2,3\}$ is to treat them on a case-by-case basis. In particular, in \Cref{SubsecExamplesSpin2andHigher} we provide an explicit example for each case of a $P_J^\star\in\mathcal{R}_J$ which is not in $\mathcal{Q}_J$. In order to find these examples, we have adapted the SDP in~(\ref{rotationBoxSDP}) to the following one:
\begin{equation}
\begin{aligned}
\max_{Q,S} \quad & c_{2J-1}+s_{2J}=2\,{\rm Re}(a_{2J-1})-2\,{\rm Im}(a_{2J})\\
\textrm{s.t.} \quad & \bullet \; a_k=\sum_{0\leq j,j+k\leq 2J}Q_{j,j+k} \, \text{ for all } k,\\
& \bullet \;  a_k=-\sum_{0\leq j,j+k\leq 2J}S_{j,j+k} \, \text{ for all } k\neq 0,\\
& \bullet \;  1-a_0=\mathrm{Tr}(S),\\
  &\bullet \;  Q,S\geq 0 .
\end{aligned}
\end{equation}
When the SDP is feasible, it returns some $(2J+1)\times (2J+1)$ matrices $Q,S$ and some complex variables $a_k$ with $k\in\{0,\ldots,2J\}$ that lead to a valid spin-$J$ correlation (c.f.\ \Cref{lemRotationBox}). Indeed, as shown in \Cref{SubsecExamplesSpin2andHigher}, the SDP for each of these cases is feasible and, moreover, its solutions are such that $c_{2J-1}^\star+s_{2J}^\star>1/\sqrt{3}$, thus showing that there exist spin-$J$ correlations that go beyond the quantum set for any $J\geq 2$.

\subsection{$\mathcal{Q}_J$ approximates all correlations for $J\to\infty$}\label{sec:Jinftyquantum}

In this section, we will concern ourselves with the case of rotation boxes of unbounded spin (producing correlations which we will denote by $\mathcal{R}_\infty$) and their quantum realization. We will see that in this case, we can approximate those boxes arbitrarily well with quantum boxes of finite spin $J$.

Elements of $\mathcal{R}_\infty$ are conditional probability distributions $\theta\mapsto P(+|\theta)$, but we do not make any assumptions on the spin as in the case of $\mathcal{R}_J$. However, one remaining physically motivated assumption is to demand that these outcome probabilities depend continuously on the angle $\theta$. In fact, this is always the case in quantum theory: there, it is typically assumed that representations $\theta\mapsto U_{\theta}$ are strongly continuous. It is easy to convince oneself that this implies that also the probabilities $P(+|\theta)={\rm Tr}(U_\theta^\dagger E_+ U_\theta \rho)$ are continuous in $\theta$. Thus, we will define
\[
   \mathcal{R}_\infty :=\{f\in \mathcal{C}({\rm SO}(2))\,\,|\,\, 0\leq f(\theta)\leq 1 \mbox{ for all } \theta \in [0,2\pi) \}.
\]
Here, $\mathcal{C}({\rm SO}(2))$ denotes the continuous real functions on ${\rm SO}(2)$, which we parametrize by the angle $\theta$. Note that periodicity holds, $f(2\pi)=f(0)$, by definition of ${\rm SO}(2)$.

We will now show that every function in $\mathcal{R}_\infty$ can be approximated to arbitrary precision by quantum spin-$J$ correlations, for large enough $J$. We are interested in \textit{uniform} approximation, i.e.\ if $P\in\mathcal{R}_\infty$, we would like to find some $Q\in\mathcal{Q}_J$, where $J$ is finite (but typically large), such that $\|P-Q\|_\infty:=\max_\theta |P(\theta)-Q(\theta)|$ is small. The following theorem makes this claim precise:

\begin{theorem}\label{inftyrotasquantumboxes}
The set of continuous rotational correlations $\mathcal{R}_\infty$ is the closure of the union of all sets of spin-$J$ quantum boxes $\mathcal{Q}_J$ with finite $J<\infty$, i.e. \begin{equation}
    \mathcal{R}_\infty=\overline{\bigcup_J \mathcal{Q}_J},
\end{equation}
where the closure is taken with respect to the uniform norm $\norm{\cdot}_\infty$ \label{theoreminftyunion}. 
\end{theorem}
As we will explain at the end of this subsection, this statement holds in completely analogous form  for more than two outcomes too, i.e.\ $\mathcal{R}_\infty^\mathcal{A}=\overline{\bigcup_J \mathcal{Q}_J^\mathcal{A}}$, with the obvious definition of $\mathcal{R}_\infty^\mathcal{A}$.

Note that the corresponding statement with $\mathcal{Q}_J$ replaced by $\mathcal{R}_J$ is trivially true: it is well-known that every continuous function on the circle can be uniformly approximated by trigonometric polynomials~\cite{Rudin}. However, at this point, we do not know whether all probability-valued trigonometric polynomials are contained in some $\mathcal{Q}_J$.
\begin{proof}
Here, we will only outline the proof idea. The technical details can be found in \Cref{inftyproofs}. 
  The proof can be divided into three steps. In the first step, we will use the Hilbert space $L^2({\rm SO}(2))$ of equivalence classes of square integrable functions over the circle and construct quantum models for elements of $R_\infty$.
To construct a quantum model for any given rotation box correlation $\theta'\mapsto P(+|\theta_0+\theta')\in \mathcal{R}_\infty $ we find an operator $\hat{P}\in \mathcal{E}(L^2(\mbox{SO(2)}))$ and a sequence of states $\{[f_{\theta_0,n}]\}_n\in\mathbb{N}\subset L^2(\mbox{SO(2)})$ such that $P(+|\theta_0+\theta')=\lim_{n\to \infty} \braket{U^\dagger (\theta')f_{\theta_0,n}}{\hat{P}U^\dagger(\theta')f_{\theta_0,n}}$, where $U$ is the regular representation, acting as $U(\theta)f(\theta')=f(\theta'+\theta)$. In more detail, we define the operator $\hat{P}$ in the following way:
\begin{eqnarray}
      (\hat{P}\psi)(\theta)=P(+|\theta)\psi(\theta).
  \end{eqnarray} 
The sequence $\{f_{\theta_0,n}\}_n$ is given by the normalized functions that are constant in the interval $[\theta_0-\frac{1}{n},\theta_0+\frac{1}{n}]$ and 0 everywhere else. The limit of these sequences can be thought as generalized normalized eigenfunctions $\ket{\theta_0}$ of $\hat{P}$, and we can write $\braket{\theta}{\theta_0}=\lim_{n\to\infty}\braket{f_{\theta,n}}{f_{\theta_0,n}}=\delta(\theta-\theta_0)$. It is easy to convince oneself that $U_\theta f_{\theta_0,n}=f_{(\theta_0-\theta),n}$ and hence, $U_\theta \ket{\theta_0}=\ket{\theta_0-\theta}$, and the claim $P(+|\theta_0+\theta')=\lim_{n\to \infty} \braket{U^\dagger (\theta')f_{\theta_0,n}}{\hat{P}U^\dagger(\theta')f_{\theta_0,n}}=\bra{\theta_0+\theta'}\hat{P}\ket{\theta_0+\theta'}$ follows. In total, we have seen that by making $n$ larger and larger, the quantum box $P_n(+|\theta_0+\theta')=\langle U^\dagger(\theta')f_{\theta_0,n}|\hat{P}U^\dagger(\theta')f_{\theta_0,n}\rangle$ more and more closely models the behavior of the rotation box $P(+|\theta_0+\theta')$.
  
In the second step, we will approximate the described quantum box $P_{n}(+|\theta_0+\theta')$ by a finite-dimensional quantum model. We will start with the same model as before, and then project it on to a finite-dimensional subspace. We recall that for the regular representation, we have a decomposition of the Hilbert space $L^2({\rm SO}(2))=\bigoplus_j \mathcal{H}_j$, where ${\mathcal H}_j$ is a one-dimensional subspace corresponding to the $j$-th irrep of ${\rm SO}(2)$, i.e.\ $U(\theta)\ket{\phi_j}=e^{ij\theta}\ket{\phi_j}$ for every $\ket{\phi_j}\in \mathcal{H}_j$. Using a basis of $L^2({\rm SO}(2))$ that respects this decomposition, we can define the projector $\Pi_J=\sum_{j=-J}^{J}\ketbra{\phi_j}$. Using this projection, we can define $P^J_n(+|\theta_0+\theta')=\Tr(\Pi_J \hat{P}\Pi_J U^\dagger(\theta')\Pi_J \ketbra{f_{\theta_0,n}}\Pi_J U(\theta'))$, which is an element of $\mathcal{Q}_J$. From the Gentle Measurement Lemma~\cite{Wilde2017} and Theorem 9.1.\ of~\cite{Nielsen2010}, it follows that if $\mbox{Tr}(\Pi_J|f_{\theta_0,n}\rangle\langle f_{\theta_0,n}|)\geq 1-\epsilon$ then $\sqrt{\epsilon}\geq \abs{P_n(+|\theta_0+\theta')-P^J_n(+|\theta_0+\theta')}$.

In the third and final step, we show that we can make $\epsilon$ arbitrarily small by making $J$ larger and larger. This is the case since $\Pi_J\rightarrow \mathds{1}$ strongly for $J\rightarrow\infty$.
\end{proof}
The above theorem can be generalized to $N$-outcome boxes. We say that an $N$-outcome rotation box is a family of functions $\{P_k\}_{k=1}^N$ such that every $P_k$ is a non-negative and continuous function on the circle, $P_k(\theta):=P(a_k|\theta)$ for $\mathcal{A}=\{a_1,\ldots,a_N\}$, and $\sum_{k=1}^N P_k(\theta)=1$ for every $\theta$. For the construction of the quantum model, we use the family of operators $\{\hat{P}_k\}_{k=1}^N $ defined by
\begin{eqnarray}
    (\hat{P}_k\psi)(\theta)=P_k(\theta)\psi(\theta),
\end{eqnarray}
and the rest of the extension to $N$ outcomes is straightforward. For the details, see again \Cref{inftyproofs}.

\subsection{Two settings: $\mathcal{Q}_{J,\alpha}=\mathcal{R}_{J,\alpha}$ and a theory-independent randomness generator}\label{subsec:randomness}
In previous work~\cite{Jones}, some of us have shown that the quantum and rotation sets of correlations are precisely the same for all $J$, when one considers just two settings (i.e.\ two possible rotations $\theta\in\{0,\alpha\}$). This equivalence is used to describe semi-device-independent protocols for randomness certification, which do not need to assume quantum theory, but instead implement some physical assumption on the response of any transmitted system to rotations.

The setup is as follows (see \cref{fig:rotationboxes}c for an illustration). The ``preparation'' box with settings $x\in\{1,2\}$ is either left unchanged for $x=1$, or rotated by some fixed angle $\alpha>0$ for $x=2$. The prepared system is then communicated to the ``measurement'' box, which outputs $a\in\{\pm1\}$. Like every semi-device-independent protocol, we have to make some assumption about the transmitted systems. Here, we assume that the spin is upper-bounded by some value $J$.

The statistics of the setup is described by a conditional probability $P(a|\theta_x)$, where $\theta_1=0$ and $\theta_2=\alpha$. There may be other variables $\Lambda$ that would admit an improved prediction of the outcome $a$, such that $P(a|\theta_x)$ is a statistical average over $\lambda$,
\[
P(a|\theta_x)=\sum_{\lambda\in\Lambda} q_\lambda P(a|\theta_x,\lambda),
\]
with some probability distribution $q_\lambda$. Equivalently, we can describe the statistics with the correlations $(E_1,E_2)$, where $E_x=P(+1|\theta_x)-P(-1|\theta_x)$. The protocol works by showing that the observation of certain correlations $(E_1,E_2)$ implies for the conditional entropy
\begin{equation}
H(A|X,\Lambda)\geq H^\star>0,
\label{EqCondEntropy}
\end{equation}
which essentially means that the setup produces $H^\star$ random bits, unpredictable even by eavesdroppers holding additional classical information $\lambda\in\Lambda$.

If we assume that quantum theory holds, the set of possible correlations in this scenario is
\[
\mathcal{Q}_{J,\alpha}:=\{(E_1,E_2)\,\,|\,\, E_x=P(+1|\theta_x)-P(-1|\theta_x),P\in\mathcal{Q}_J\},
\]
where $\theta_1=0$ and $\theta_2=\alpha$. Based on earlier work by other authors~\cite{VanHimbeeck2017,VanHimbeeck2019}, we have shown in~\cite{Jones} that this quantum set of correlations $\mathcal{Q}_{J,\alpha}$ can be exactly characterized by the inequality
\begin{equation}
\frac{1}{2}\left(\sqrt{1+E_1}\sqrt{1+E_2}+\sqrt{1-E_1}\sqrt{1-E_2}\right) \geq \delta,
\label{quantumset}
\end{equation}
where
\begin{equation}
\delta=\left\{\begin{matrix}
\cos(J\alpha)&\mbox{if}& |J\alpha|< \frac{\pi}{2}\\
0 &\mbox{if}&|J\alpha|\geq\frac{\pi}{2}
\end{matrix}\right. .
\label{eqQuantum2}
\end{equation}
If we do \textit{not} assume quantum theory, the corresponding set of correlations is
\[
\mathcal{R}_{J,\alpha}:=\{(E_1,E_2)\,\,|\,\, E_x=P(+1|\theta_x)-P(-1|\theta_x),P\in\mathcal{R}_J\}.
\]
Using a lemma~\cite[Thm.\ 1.1]{Devore} that constrains the derivative of trigonometric polynomials (also used here for the convex characterization of $\mathcal{R}_1$, see \cref{eqDiff}), we show that rotation box correlations must satisfy precisely the same condition as in the quantum case, i.e.
\begin{equation}\label{2-setting equivalence}
\mathcal{Q}_{J,\alpha}=\mathcal{R}_{J,\alpha}.
\end{equation}
Thus, for two settings and two outcomes, the possible quantum and general spin-$J$ correlations are identical. For example, statements like \textit{``the system must be rotated by at least $\pi/(2J)$ to obtain a perfectly distinguishable state''} are not only true in quantum theory, but in every physical theory:
\begin{lemma}
\label{LemDistZeroOne}
Suppose that $P\in\mathcal{R}_J$ with $P(+|\theta_0)=0$ and $P(+|\theta_1)=1$. Then $|\theta_1-\theta_0|\geq \pi/(2J)$.
\end{lemma}

This equivalence, \cref{2-setting equivalence}, allows us certify randomness independently of the validity of quantum theory. In particular, we characterize the set of ``classical'' correlations, i.e. for a given set of correlations, the subset containing all those that admit a description as the convex combination of \textit{deterministic} correlations. This is clearly the same for both quantum and rotation cases, due to the equivalence expressed in \cref{2-setting equivalence}. Moreover, for $0<J\alpha<\pi/2$, the classical set is a strict subset of the quantum and rotation sets: $\mathcal{C}_{J,\alpha}\subsetneq\mathcal{Q}_{J,\alpha}=\mathcal{R}_{J,\alpha}$. Therefore, there exist correlations (predicted by quantum theory) that are incompatible with any deterministic description, even when one allows for post-quantum strategies. Observing such correlations $(E_1,E_2)\in\mathcal{Q}_{J,\alpha}\setminus\mathcal{C}_{J,\alpha}$ certifies a number $H^\star$ of random bits, as in \cref{EqCondEntropy}, which is independent of whether quantum theory holds. That is, even an eavesdropper with arbitrary additional classical information $\lambda\in\Lambda$, as well as access to post-quantum physics, could not anticipate the outputs of the device. 

Accordingly, we can conceive of a random number generator whose outputs are provably random irrespective of the validity of quantum theory, with its security instead anchored in the geometry of space. This analysis is further shown to be robust under some probabilistic assumption that allows for experimental error in the spin bound.

\subsection{What are classical rotation boxes?}\label{subsec:classical}

Classical rotation box correlations are generated by a classical system with an $\SO(2)$ symmetry. For finite-dimensional systems, this entails there is a representation of $\SO(2)$ of the form given in Eq.~\eqref{eq:GPT_SO2_rep} acting on the state space of the classical system. For $n \in {\mathbb N}$, the finite-dimensional $n$-level classical system has a state space given by an $n$-simplex~\cite{Plavala,Mueller}:
\begin{align}
    \Delta_n= \{(p_1,...,p_n) \,\,|\,\, p_i \geq 0, \sum_{i =1}^n p_i =1\} \subset \reals^n, 
\end{align}
and an effect space given by a $n$-dimensional hypercube 
\begin{align}
    \Box_n = \{(e_1,...,e_n)\,\,|\,\, 0 \leq e_i \leq 1\} \subset \reals^n.
\end{align}
The set of symmetries of $\Delta_n$ is $\Sigma(n)$, which is the symmetric group on $n$ objects. Since $\SO(2)$ is not a subgroup of $\Sigma(n)$, it follows that the only representation of $\SO(2)$ which maps $\Delta_n$ to itself is the trivial representation. Thus the set of finite-dimensional classical systems generate the set $\mR_0$ of trivial spin-$0$ correlations.

Infinite-dimensional classical systems can carry non-trivial actions of $\SO(2)$. Consider a system with configuration space given by the circle $S^1$ which carries the standard action of $\SO(2)$. 

The circle $S^1$ has a topology induced by the standard topology on $\reals^2$, and thus a Borel $\sigma$-algebra~\cite{Rudin}. States of the $S^1$ classical system are probability measures on $S^1$, while effects are given by measureable functions $f: S^1 \to \mathbb{R}$ that take values between zero and one everywhere, i.e.\ $0\leq f(\theta)\leq 1$ for all $\theta$. We denote the space of probability measures on $S^1$ by $M_1^+(S^1)$, and the space of measureable functions on $S^1$ by $M^*(S^1)$. 

Note that every continuous function $f(S^1) \to \reals$  is such that the preimage $f^{-1}(A)$ is open if A is open. Since the Borel $\sigma$-algebra is the $\sigma$-algebra generated by open sets, every $f \in \mathcal C(S^1)$ is measurable.  Since trigonometric polymomials are continuous, every trigonometric polynomial $P(a|\theta) \in M^*(S^1)$.

Denoting by $\delta_\theta$ the Dirac measure at the point $\theta$, we have that every element in $\mR_\infty$ can be generated by this infinite-dimensional classical system:
\begin{align}
    P(a|\theta) = \int_{\theta'} P(a|\theta') \delta_\theta.
\end{align}
We note that the standard action of $\SO(2)$ on the circle induces an action on $M_1^+(S^1)$, which acts on the extremal measures as:
\begin{align}
    \delta_\theta \mapsto \delta_{\theta + \theta'} .
\end{align}
The classical system can be thought of as `containing' every spin-$J$ system, since the subspace of $M^*(S^1)$ of trigonometric polynomials of degree $2J$ or less carries a representation $\bigoplus_{k = 0}^{2J} \gamma_k$, where $\gamma_k$ is the real representation of $\SO(2)$ given in \cref{eq:SO2realirrep0,eq:SO2realirrep}. Thus, there is no finite $J$ that characterizes this classical system. Moreover, for any fixed finite value of $J$, this mathematical subspace cannot be interpreted as an actual standalone physical subsystem in any operationally meaningful way.

Conversely, every classical system has the property that all pure states are perfectly distinguishable. Thus, if the ${\rm SO}(2)$-action $\theta\mapsto T_\theta$ acts non-trivially on at least one pure state $\omega$, then $\omega(\theta):=T_\theta\omega$ will be another pure state that is perfectly distinguishable from $\omega$, no matter how small the angle $\theta>0$. But this is incompatible with a finite value of $J$, as observed in \Cref{LemDistZeroOne}.

The inexistence of any classical finite-spin boxes means that while any rotation box correlation $P(a|\theta)$ can be arbitrarily well approximated by a finite-dimensional quantum spin-$J$ system, one always needs an infinite-dimensional $\SO(2)$ classical system to approximate or reproduce it, unless $P(a|\theta)$ is constant in $\theta$ for every $a$.

Our discussion above has focused on the paradigmatic examples of classical systems described by finite- or infinite-dimensional simplices of probability distributions, but one might instead ask more nuanced and detailed questions about the compatibility of finite spin $J$ and different notions of classicality. For example, how about classical systems with an epistemic restriction~\cite{SpekkensEpistemic}? Are systems of finite spin always \textit{contextual} in the sense that they cannot be linearly embedded into any classical system~\cite{Schmid}, and if so, how crucial is the assumption of transformation-noncontextuality~\cite{SpekkensContextuality}? We leave the discussion of these interesting questions to future work.

\section{Rotation boxes in the Bell scenario}\label{sec:rotationbell}

In this section, we consolidate and generalize two earlier results which show how the notion of rotation boxes can be applied in the context of Bell nonlocality: assumptions on the local transformation behavior can be used to characterize the quantum Bell correlations for 2 parties with 2 measurements and 2 outcomes each~\cite{Garner}, and they allow us to construct witnesses of Bell nonlocality for $N$ parties~\cite{Nagata}. Since many experimental scenarios indeed feature continuous periodic inputs, we think that these are only two examples of a potentially large class of applications of the framework.

\subsection{Two parties: exact characterization of the quantum $(2,2,2)$-behaviors}\label{subsec:two parties}

One of us and co-authors have shown in~\cite{Garner} that the quantum (2,2,2)-correlations can be characterized exactly in terms of the local transformation behavior with respect to rotations in $d$-dimensional space, for every $d\geq 2$. Here, we give a stand-alone argument for the special case $d=2$.

This result contributes to the longstanding research program of characterizing the set of quantum correlations inside the larger set of correlations that satisfy the no-signalling (NS) principle, see~\cite{Popescu} for an overview. The no-signalling principle formalizes the idea that information transfer has finite speed in order to constrain the influence between space-like separated events: one party's choice of measurement cannot instantaneously influence the local statistics of the other. The NS principle, initially introduced in~\cite{tsirelsonNS}, was established as a foundational component of a framework in~\cite{PRbox} where the so-called Popescu-Rohrlich correlations (or PR boxes) revealed that non-local correlations beyond those allowed by quantum mechanics are theoretically possible under the constraints of relativistic causality. That is, the set of NS correlations is known to contain the set of quantum correlations as a proper subset. However, while the NS principle has proven useful in several contexts for upper-bounding feasible correlations, characterizing the set of quantum correlations $Q$ via simple physical principles remains an open problem~\cite{Popescu}.

Suppose that Alice holds a spin-$1/2$ rotation box, $P\in\mathcal{Q}_{1/2}$: she can choose her input by performing a spatial rotation by some angle $\alpha$, and obtain one of two outcomes $a\in\{-1,+1\}$. Furthermore, suppose that the outcome is not only an abstract label, but has an additional geometric interpretation: Alice's input is a spatial vector $\vec n=(\cos\alpha,\sin\alpha)$ (say, of a magnetic field), and her output is physically realized by giving her an answer that is either \textit{parallel} ($a=1$) or \textit{antiparallel} ($a=-1$) to $\vec n$. Indeed, this situation is realized by a Stern-Gerlach experiment on a spin-$1/2$ particle in $d=3$ dimensions; here we restrict ourselves to $d=2$.

This physical intuition can be expressed as the following expectation:
\begin{itemize}
\item[(i)] If outcome $a$ is obtained on input $\alpha$, then outcome $-a$ would have been obtained on input $\alpha+\pi$.
\end{itemize}
To make this mathematically rigorous, we have to adapt this (untestable) counterfactual claim to a (testable) statement about probabilities, namely:
\begin{itemize}
\item[(ii)] $P(a|\alpha)=P(-a|\alpha+\pi)$.
\end{itemize}
Since we can always write $P(+|\alpha)=c_0+c_1\cos\alpha+s_1\sin\alpha$, this is equivalent to the condition $c_0=\frac 1 2$, and it is also equivalent to
\begin{itemize}
    \item[(iii)] $\frac 1 {2\pi}\int_0^{2\pi} P(a|\alpha)\,d\alpha=\frac 1 2 $ for $a=+1$ and $a=-1$.
\end{itemize}
That is, on average (over all directions), no outcome is preferred. We say that Alice's box is \textit{unbiased}~\cite{Garner} if one of the two (and thus both) equivalent conditions (ii) or (iii) hold. As explained above, this property follows from a geometric interpretation of Alice's outcome as indicating that she obtains a resulting vector that is either parallel or antiparallel to her input vector.

Now consider a Bell experiment, where both Alice and Bob hold unbiased spin-$1/2$ boxes. Let us \textit{not} assume that quantum theory holds; let us only assume that the no-signalling principle is satisfied. In this case, Alice and Bob would choose inputs $\alpha$ and $\beta$ and obtain outputs $a,b\in\{-1,+1\}$ such that the resulting behavior
\[
   P(a,b|\alpha,\beta)
\]
satisfies the no-signalling conditions
\begin{eqnarray*}
    \sum_a P(a,b|\alpha,\beta)&=& \sum_a P(a,b|\alpha',\beta)=:P^B(b|\beta),\\
    \sum_b P(a,b|\alpha,\beta)&=&\sum_b P(a,b|\alpha,\beta')=:P^A(a|\alpha).
\end{eqnarray*}
Let us assume that Alice's and Bob's local boxes are \textit{always} spin-$1/2$ boxes, and are \textit{always} unbiased, regardless of what the other party measures. That is, consider the situation in which Bob decides to input angle $\beta$ into his box, and obtains outcome $b$, and subsequently communicates this choice and outcome to Alice (say, over the telephone). In this case, Alice would update her probability assignment to
\[
P^A_{b,\beta}(a|\alpha):=\frac {P(a,b|\alpha,\beta)}{P^B(b|\beta)},
\]
where $P^B(b|\beta)$ is the probability for Bob to obtain outcome $b$. We will assume that this ``conditional box'' still produces an unbiased spin-$1/2$ correlation, for all values of $\beta$ and $b$, and we make the analogous assumption if the roles of Alice and Bob are interchanged.

Note that we are \textit{not} making any assumptions about the \textit{global} correlations (or their transformation behavior) directly, except that we demand no-signalling. 

Surprisingly, the conditions above enforce that the global correlations are quantum (see Appendix~\ref{AppProofCorr} for the proof):
\begin{theorem}
\label{TheQuantumCorrelations}
Under the assumptions above, the behavior $P$ is a quantum behavior. That is, there exists a quantum state $\rho_{AB}$ on the two-qubit Hilbert space $AB$ and a positive map $\tau$ on $B$ with $\tau(\mathbf{1}_B)=\mathbf{1}_B$ such that
\[
P(a,b|\alpha,\beta)=\Tr\left(\strut \rho_{AB} e^{-i\alpha Z} |a\rangle\langle a| e^{i\alpha Z} \otimes \tau(e^{-i\beta Z} |b\rangle\langle b| e^{i\beta Z})
\right),
\]
where $Z=\frac 1 2\left(\begin{array}{cc} 1 & 0 \\ 0 & -1 \end{array}\right)$ is half of the Pauli-$Z$ matrix, and $|\pm 1\rangle=\frac 1 {\sqrt{2}}(|0\rangle\pm|1\rangle)$.
\end{theorem}
We do not currently know whether the unitary rotation by angle $\beta$ can be pulled out of the map $\tau$, or whether this positive, but not necessarily completely positive, map can perhaps be dropped completely. This map $\tau$ is, however, necessary in the analogous statement for dimension $d=3$: it is well-known that the quantum singlet state of two spin-$1/2$ particles leads to perfect \textit{anticorrelation} between Alice's and Bob's binary outcomes~\cite{DakicBrukner}, but that there is no quantum state that would lead to perfect \textit{correlation}. Formally, perfect correlation can be obtained by taking the partial transpose of one half of the singlet state, and considering the resulting action on Bob's local measurement (while leaving the singlet state intact) can be interpreted as a reflection of Bob's description of spatial geometry relative to Alice's.

Note that $P$ will be a quantum correlation even if a non-completely positive map $\tau$ is necessary: this map cannot be physically implemented, but Bob can still use it to calculate the set of POVM elements that he should use to measure. This way, Alice and Bob can make sure to generate correlations according to $P(a,b|\alpha,\beta)$.

If Alice and Bob restrict themselves to input one of two angles each, $\alpha_0,\alpha_1$ or $\beta_0,\beta_1$, they generate an instance of what has been called the quantum $(2,2,2)$-behaviors (2 parties, 2 settings and 2 outcomes each):
\[
   P(a,b|x,y):=P(a,b|\alpha_x,\beta_y)\quad (x,y\in\{0,1\}).
\]
The above theorem shows that if Alice's and Bob's local conditional boxes are spin-$1/2$ boxes and unbiased, then $P(a,b|x,y)$ will be a quantum $(2,2,2)$-behavior. In this case, the mere possibility that Alice and Bob \textit{could have} input other angles, and that the outcome probabilities would have had to depend linearly on the resulting two-dimensional vectors, constrains these correlations to be quantum.

The results of~\cite{Garner}, however, show more: \textit{all} quantum $(2,2,2)$-behaviors can be obtained in this way, if supplemented with shared randomness:
\begin{theorem}
The set of quantum $(2,2,2)$-behaviors is \textit{exactly} the set of non-signalling behaviors that can be obtained in Bell experiments from ensembles of nonlocal boxes that are locally unbiased and locally spin-$1/2$.
\end{theorem}
That is, regardless of which theory holds, the resulting behaviors will be quantum. Moreover, all such quantum behaviors can be realized in some theory, namely quantum theory, via random choices among boxes that are locally spin-$1/2$ and unbiased.

The proof is based on the well-known fact that all extremal quantum $(2,2,2)$-behaviors can be generated on two qubits (and, locally, on the equatorial plane of these qubits, i.e.\ on two rebits)~\cite{Tsirelson1980,AlmostQuantum,Masanes,Toner}. To obtain all non-extremal quantum $(2,2,2)$-behaviors, Alice and Bob need additional shared randomness that allows them to select at random between one of several such boxes. See~\cite{Garner} for an explanation of why shared randomness cannot be avoided.

To see that local unbiasedness cannot be removed as a premise of the theorems above, consider the following example. Suppose that Alice and Bob hold local spin-$1/2$ boxes $S_A,S_B\in\mathcal{R}_{1/2}^{\{0,1\}}$, satisfying
\[
Q_A(1|\alpha)=\frac 1 2+\frac 1 2 \cos\alpha, \enspace Q_B(1|\beta)=\frac 1 2+\frac 1 2 \cos\beta.
\]
What they do is the following. Alice and Bob input their angles into their local boxes, and feed their respective outcomes $x,y\in\{0,1\}$ into a PR box
\[
P_{\rm PR}(a,b|x,y)=\frac 1 2\delta_{(1-ab)/2,xy}\quad (a,b\in\{-1,+1\}).
\]
That is, if the inputs to the PR box are $x=y=1$, they obtain perfectly anticorrelated outputs, and otherwise, perfectly correlated ones. The result of this procedure defines their non-signalling behavior $P$. It is not difficult to see that $P^B(b|\beta)=\frac 1 2$ for all $b$ and all $\beta$, and hence
\begin{eqnarray}
P_{b,\beta}^A(a|\alpha)&=&2 P(a,b|\alpha,\beta)\nonumber\\
&=& 2\sum_{c,d=0}^1 P_{\rm PR}(a,b|c,d)Q_A(c|\alpha)Q_B(d|\beta)
\label{eqConditionalState}
\end{eqnarray}
is a trigonometric polynomial of degree $1$ in $\alpha$, for every fixed $b$, $\beta$, and $a$. Similar reasoning applies to $P_{a,\alpha}^B(b|\beta)$. Hence, all local conditional boxes are spin-$1/2$ boxes. Set $\alpha_0=\beta_0:=\pi$ and $\alpha_1=\beta_1:=0$, then $Q_A(c|\alpha_x)=\delta_{cx}$ and $Q_B(d|\beta_y)=\delta_{dy}$, and so
\begin{eqnarray*}
P(a,b|\alpha_x,\beta_y)&=& \frac 1 2 \sum_{c,d=0}^1 \delta_{(1-ab)/2,cd} Q_A(c|\alpha_x)Q_B(d|\beta_y)\\
&=& \frac 1 2 \delta_{(1-ab)/2,xy}=P_{\rm PR}(a,b|x,y).
\end{eqnarray*}
Since $P$ can reproduce the PR box correlations for two fixed angles, it is not a quantum behavior. And this is consistent with the theorems above because $P$ is not locally unbiased. To see this, use Eq.~(\ref{eqConditionalState}) and find, for example,
\[
P_{-1,\beta}^A(+1|\alpha)=\left(\frac 1 2 + \frac 1 2 \cos\alpha\right)\left(\frac 1 2 +\frac 1 2 \cos\beta\right).
\]
Treating this as a trigonometric polynomial in $\alpha$, the coefficient $c_0$ equals $\frac 1 4\left(1+\cos\beta\right)$, which is not for all $\beta$ equal to $\frac 1 2$. That is, $P$ is not locally unbiased.

\subsection{Many parties: witnessing Bell nonlocality}\label{subsec:many parties}

Our framework also helps to clarify and generalize the results of Nagata et al.~\cite{Nagata}. In this paper, the authors offer an additional constraint on local realistic models of physical phenomena, which they refer to as rotational invariance, but we shall call \textit{spin direction linearity} (reasons for which will become clear). This allows for indirect witnesses of Bell nonlocality, for correlations that would otherwise have a local hidden variable description.

They consider an $N$-party Bell-type scenario, in which every party holds a spin-$\frac{1}{2}$ particle. Each party measures the spin component in a chosen direction $\vec{n}_j$, and outputs a local result $r_j(\vec{n}_j)\in\{\pm1\}$. The ``Bell'' correlation function is introduced as the average of the product of all local results: $E(\vec{n}_1,\dots,\vec{n}_N)=\langle r_1(\vec{n}_1)\dots r_N(\vec{n}_N)\rangle_{\rm avg}$. Their additional assumption (spin direction linearity) enforces the following structure for any such correlations:
\[
E(\vec{n}_1,\dots,\vec{n}_N)=\hat{T}\cdot(\vec{n}_1\otimes\dotsc\otimes\vec{n}_N),
\]
where $\hat{T}$ is the correlation tensor $T_{i_1,\dots,i_N}\equiv E(\vec{x}_1^{(i_1)},\dotsc,\vec{x}_N^{(i_N)})$, where $\vec{x}_j^{(i_j)},i_j\in\{1,2,3\}$ are unit directional vectors of the local coordinate system of the $j$th party.
This is to say that the correlation function is \textit{linearly} dependent on the unit directions $\vec{n_j}$ along which the spin component is measured, i.e.
\[
E(\vec n_1,\ldots,\vec n_N)=T_{i_1,\dots,i_N}n_{i_1}\dots n_{i_N},
\]
with summation over repeated indices.

The three assumptions allow the authors to derive a more restrictive Bell-type inequality, namely:
\[
\pi^N \sum_{i_1,\dots,i_N=1,2}T_{i_1,\dots,i_N}^2\leq4^NT_{\max},
\]
where $T_{\max}$ is the maximal possible value of the correlation tensor component, i.e. 
\[
T_{\max}=\max_{\vec{n}_1,\dots,\vec{n}_N}E(\vec{n}_1,\dots,\vec{n}_N).
\]
This would be evaluated by measuring the components $T_{i_1,\dots,i_N}$ that compose $\hat{T}$, and then using the tensor to determine the maximum value of $E(\vec{n}_1,\dots,\vec{n}_N)$.
Their inequality being strictly less general than Bell's theorem then allows for the certification of ``non-classical'' phenomena by observing correlations that would otherwise not violate any Bell inequality. In such an instance, non-classicality is to say that the assumptions of locality, realism and spin direction linearity cannot jointly hold. In particular, the authors of~\cite{Nagata} give an example of correlations $T$ that admit a local hidden variable model, but that do not admit such a model if one assumes in addition spin direction linearity.

Although their result is formulated for ${\rm SO}(3)$, with spin directions defined by vectors $\vec{n}_j$ in three dimensions, the authors use the reparameterization $\vec{n}_j(\alpha_j)=\cos(\alpha_j)\vec{x}^{(1)}_j+\sin(\alpha_j)\vec{x}^{(2)}_j$, for the plane defined by $\vec{x}^{(1)}_j,\vec{x}^{(2)}_j$, such that their main result is stated in terms of just one parameter $\alpha_j$ per party. Accordingly, the results hold equally for rotations constrained to a 2D-plane, i.e. ${\rm SO}(2)$ rather than ${\rm SO}(3)$. It follows that our framework may be relevant to understand or generalize their results. 

In particular, spin direction linearity is not actually about rotational invariance, as is claimed in their paper, but rather captures the assumption that the local systems are spin-$\frac{1}{2}$ particles. (Moreover, we will claim that one need only assume that the local systems can be described by a spin-$\frac{1}{2}$ box.) The states of a single spin-$\frac{1}{2}$ system (a qubit) can be represented by unit vectors on the Bloch ball:
\[
\rho=\frac{1}{2}(\mathds{1}+\vec{n}\cdot\vec{\sigma}),
\]
which, by measuring in the basis as defined by the $j$th observer, are mapped via unitary transformations $U_\theta$ to states 
\[
\rho'=\frac{1}{2}(\mathds{1}+(R_\theta\cdot\vec{n})\cdot\vec{\sigma}).
\]
Local probabilities are linear in states, so are affine-linear in spin direction $\vec{n}_j=R_\theta\cdot\vec{n}$. The local, conditional boxes $\tilde{P}(r_j|\vec{n}_j)$ (an $N$-party extension of the conditional boxes introduced in Section~\ref{subsec:two parties}) can be written as
\[
\frac{P(r_1,\ldots,r_N|\vec{n}_1,\ldots,\vec{n}_N)}{P(r_1,\ldots,r_{j-1},r_{j+1},\ldots,r_N|\vec{n}_1,\ldots,\vec{n}_{j-1},\vec{n}_{j+1},\ldots,\vec{n}_N)},
\]
so probabilities $P(r_1,\ldots,r_N|\vec{n}_1,\ldots,\vec{n}_N)$ will be affine-linear in spin directions $\vec{n}_j$, for all $1\leq j\leq N$. The constant drops out when going from probabilities to correlations, so then we get spin direction linearity when all subsystems are spin $\frac{1}{2}$.

So far, this demonstrates that the systems being spin-$\frac{1}{2}$ is a \textit{sufficient} condition for $E(\vec{n}_1,\dots,\vec{n}_N)$ to be linear in spin directions. This can also be seen in our framework, by noting that the local systems being spin-$\frac{1}{2}$ means that the local conditional boxes $\tilde{P}(r_j|\vec{n}_j)$ should be in $\mathcal{R}_{1/2}$; i.e. they are trigonometric polynomials in $\alpha_j$ of degree $1$ at most. On the other hand, if the local systems are not spin-$\frac{1}{2}$, then the probabilities may contain $\sin(k\alpha_j)$ or $\cos(k\alpha_j)$ terms (for $k\geq 2$), in which case spin direction linearity is violated.
As such, we can note that the systems being spin-$\frac{1}{2}$ is also a \textit{necessary} condition for spin direction linearity. This is to say, the main result of~\cite{Nagata} can be clarified using our framework as an inequality derived from locality, realism and the assumption that the local systems can be characterized as spin-$\frac{1}{2}$ boxes. Notably, this reformulation does not rely on the validity of quantum theory (the systems do not need to be \textit{quantum} spin-$\frac{1}{2}$ particles, as in their paper); all three assumptions are theory-independent.

\section{Connection to other topics}\label{sec:connectionothertopics}

\subsection{Almost quantum correlations}\label{subsec:almostq}

As discussed in Subsection~\ref{SubsecRelaxation}, the set of rotation box correlations bears close resemblance to the set of \textit{almost quantum correlations}~\cite{AlmostQuantum}. Indeed, any $P \in \mR_J$ can be generated as follows:
\begin{align}
     P(+|\theta)=\langle\psi|U_\theta^\dagger E_+ U_\theta |\psi\rangle,
\end{align}
where $\ket \psi \in \comp^{2J + 1}$ and $E_+$ is positive semidefinite but not necessarily a POVM element. The only requirement is that $E_+$ gives valid probabilities on the states of interest, i.e.\ on the states $U_\theta \ket \psi$ for all $\theta$. 

This is analogous to almost quantum correlations which are a relaxation of the Bell correlations generated by quantum systems. In standard quantum theory, local separation of the measurement parties (and therefore the no-signalling condition) is implemented by assigning commuting subalgebras to them. For example, consider the case where we have two observers Alice and Bob. We denote  Alice's subalgebra by $\mathcal{A}\subseteq\mathcal{C}$ and Bob's subalgebra by $\mathcal{B}\subseteq\mathcal{C}$, where $\mathcal{C}$ can be thought of as a larger global algebra. Here, the commutativity of $\mathcal{A}$ and $\mathcal{B}$ means that every $A\in\mathcal{A}$ commutes with every $B\in\mathcal{B}$, i.e. $[A,B]=0$. When we describe the measurements of Alice and Bob, we equip them with collections of PMs (projective measurements) $\{E^A_{a|x}\}_{a,x}\subset\mathcal{A}$ and $\{E^B_{b|y}\}_{b,y}\subset\mathcal{B} $ respectively, where for every input $x$, the set $\{E^{A}_{a|x}\}_{a}$ is a valid PM (and similarly for Bob). Then, the correlations between Alice and Bob are given by $P(a,b|x,y) = \bra{\psi} E^A_{a|x}E^B_{b|y}\ket{\psi}$. For ``almost quantum'' correlations, the assumption that Alice's and Bob's collections of PMs are subsets of two commuting subalgebras is relaxed. That is, not all elements of Alice's collection of PMs have to commute with all elements of Bob's PM collections, but it is only assumed that they commute on the state of interest for a given setup. In other words, if a given preparation is described by the state $\ket{\psi}$, it is assumed that $[E^A_{a|x},E^B_{b|y}]\ket{\psi}=0$ for all inputs $x$ and $y$ and outputs $a$ and $b$. Furthermore, the correlations are still computed by the Born rule $p(a,b|x,y) = \bra{\psi} E^A_{a|x}E^B_{b|y}\ket{\psi}$. We note that the product $E^{A}_{a|x}E^B_{b|y}$ cannot be considered a bipartite local effect by itself, but only obtains its meaning by combining it with the state $\ket{\psi}$ describing the physical situation. This resembles the situation for the rotation boxes, where $E_+$ by itself is not a POVM element, and only the combination of $E_+$ with the states $\{U_\theta\ket{\psi}\}_\theta$ has a physical meaning.

Furthermore, a notable feature both relaxation sets share is that they admit a characterization in terms of semidefinite constraints (as we have seen in~\ref{SubsecRelaxation}), which allows us to efficiently solve optimization problems within their set by means of SDP in order to bound quantum solutions~\cite{TavakoliSDP}. This is in contrast to the quantum sets (of Bell resp.\ spin correlations) which are not known to have characterizations in terms of SDPs.

\subsection{Orbitopes and spectrahedra}~\label{subsec:orbitope}
In this section, we show that the state spaces of the spin-$J$ rotation box systems $\Omega_J$ are isomorphic to universal Carath{\'e}odory orbitopes. Moreover, we show they are isomorphic to spectrahedra. A spectrahedron is the intersection of an affine space
with the cone of positive-semidefinite matrices.

Given a list of integers $A = (a_1,..., a_n) \in \nats^n$, the Carath{\'e}odory orbitope $C_A$~\cite{Sanyal_2011} is defined as the convex hull of the following $\SO(2)$ orbit in $\reals^{2n}$:
\begin{align}
    \{(\cos(a_1 \theta), \sin(a_1 \theta) ,..., \cos(a_n \theta) , \sin(a_n \theta))| \theta \in [0,2\pi) \}.
\end{align}
The orbitope $C_{(1,...,d)}$ is known as the universal Carath{\'e}odory orbitope $C_d$, and is affinely isomorphic to the state space $\Omega_{J = \frac{d}{2}}$ of the spin-$J$ rotation box system.  Similarly $\hat C_d^o$, the co-orbitope cone dual to $C_{(1,...,d)}$ is the set of non-negative trigonometric polynomials and is isomorphic to the cone generated by the effect space $\mE_J$.

Explicitly, $\hat C_d^o$ is given by:
\begin{align}
    \{(c_0,c_1,s_1,...,c_d,s_d) \in \reals^{2d+1} | c_0 \nonumber\\
     + \sum_{k = 1}^d c_k \cos(k \theta) + s_k \sin(k \theta) \geq 0\} .
\end{align}
We can characterize $C_d$ in terms of $\hat C_d^o$ as follows: a point $(a_1,b_1,...,a_d,b_d) \in \reals^{2d}$ is in the universal Carath{\'e}odory orbitope $C_d$ if and only if
\begin{align}
    c_0 + \sum_{k =1}^d c_k a_k + s_k b_k \geq 0, \forall (c_0,c_1,s_1,...,c_d,s_d) \in \hat C_d^o .
\end{align}

By Theorem 5.2 of~\cite{Sanyal_2011}, the universal Carath{\'e}odory orbitope $C_d$ (and therefore $\Omega_{J = \frac{d}{2}}$) is isomorphic to the following spectahedron:
\begin{align}
    \begin{pmatrix}
        1 & x_1 & \hdots & x_{d-1} & x_d \\
        y_1 & 1 &  \ddots & x_{d-2} & x_{d-1} \\
        \vdots & \ddots & \ddots & \ddots  & \vdots \\
        y_{d-1} & y_{d-2} & \ddots & 1 & x_1 \\
        y_d & y_{d-1} & \hdots & y_1 & 1 ,
    \end{pmatrix} , 
\end{align}
where
\begin{align}
    x_j = a_j + i b_j , \\
    y_j = a_j - i b_j ,
\end{align}
and $(a_1,b_1,...,a_d,b_d) \in \reals^{2d}$ is a point in the orbitope $C_d$. The extremal points occur for $a_k = \cos(k \theta)$ and $b_k = \sin(k \theta)$, thus the orbitope $C_d$ is the convex hull of:
\begin{align}\label{eq:univorbit}
    \begin{pmatrix}
        1 & e^{i \theta} & \hdots & e^{i (d-1) \theta} & e^{i  d\theta}\\
        e^{- i \theta} & 1 &  \ddots & e^{i (d-2) \theta}&e^{i (d-1) \theta} \\
        \vdots & \ddots & \ddots & \ddots  & \vdots \\
        e^{-i (d-1) \theta} & e^{-i (d-2) \theta}  & \ddots & 1 & e^{i \theta} \\
         e^{-i  d\theta} &  e^{-i (d-1) \theta} & \hdots & e^{-i \theta} & 1 
    \end{pmatrix} .
\end{align}

Let us note that this statement is equivalent to Theorem~\ref{ThmRJ}. Consider the orbit $U_\theta \ketbra{\psi}{\psi}U_\theta^\dagger$ for $\ket \psi$ and $U_\theta$ as defined in Theorem~\ref{ThmRJ}:
\begin{align}\label{eq:jorb}
    \frac{1}{2J + 1} \begin{pmatrix}
        1 & e^{i  \theta} & \hdots & e^{i (2 J -1) } & e^{i 2J} \\
         e^{-i  \theta} & 1 &  \ddots & e^{i (2 J -2) } & e^{i (2 J -1) } \\
        \vdots & \ddots & \ddots & \ddots  & \vdots \\
        e^{-i (2 J -1 )} & e^{-i (2 J -2 )} & \ddots & 1 & e^{i  \theta}\\
        e^{-i 2J}  & e^{-i (2 J -1 )}& \hdots & e^{-i  \theta} & 1 
    \end{pmatrix} . 
\end{align}
This orbit is isomorphic to the orbit of Eq.~\eqref{eq:univorbit} for $d = 2J$. According to Theorem~\ref{ThmRJ}, every spin-$J$ correlation $P\in\mathcal{R}_J$ can be written
\[
P(+|\theta)={\rm Tr}(E_+ U_\theta |\psi\rangle\langle\psi|U_\theta^\dagger),
\]
i.e.\ is a linear functional that takes values in $[0,1]$ on this orbitope; and, conversely, every such functional is an element of $\mathcal{R}_J$. Therefore, we may say that $\Omega_{J=\frac d 2}$, the state space of the spin-$J$ GPT system $\tR_J$, is an orbitope, and moreover, it can be interpreted, due to Theorem~\ref{ThmRJ}, as a subset of the quantum state space.

\subsection{Symmetric entanglement witnesses for rebits}\label{subsec:symmetric}

Consider the following orbit of qubit states $\ketbra{\psi(\theta)}{\psi(\theta)}$ in $\mathcal{D}(\comp^2)$,  where
\begin{align}\label{eq:rebitorbit}
     \ket{\psi(\theta)} &= U(\theta) \ket + = \frac{1}{\sqrt{2}} (e^{i \frac{\theta}{2}} \ket{0} + e^{-i \frac{\theta}{2}} \ket{1} ),
\end{align}
with
\[
    U(\theta) = \begin{pmatrix}
        e^{i \frac{\theta}{2}} & 0 \\
        0  & e^{-i \frac{\theta}{2}}
    \end{pmatrix} ,\quad \ket  \pm = \frac{\ket 0 \pm \ket 1}{\sqrt{2}} .
\]
By writing the orbit in the $\{\ket +, \ket - \}$-basis
\begin{align}
    \ket{\psi(\theta)} = \cos\frac{\theta}{2} \ket + + \sin\frac{\theta}{2} \ket -  ,
\end{align}
we see that it corresponds to the pure states of a rebit (a qubit in quantum theory over the real numbers $\mathbb{R}$), acted on by a real projective representation of $\SO(2)$. The orbit $\ketbra{\psi(\theta)}{\psi(\theta)}$ can thus be viewed as an orbit of rebit states in $\LS(\reals^2)$, the symmetric linear operators on $\reals^2$, or alternatively as an orbit of symmetric qubit states in $\LSH(\comp^2) \subset \LH(\comp^2)$, where $\LSH(\comp^2)$ are the symmetric Hermitian operators, in this case with respect to the $\ket\pm$ basis.

Given $d$ rebits with pure states corresponding to rays in $(\reals^2)^{\otimes d}$, the pure symmetric states are those lying in $\Sym^d(\reals^2)$, the symmetric subspace of $(\reals^2)^{\otimes d}$. The set of pure symmetric product states is the set of $\ket{\psi}^{ \otimes d}$, where $\ket\psi$ is an arbitrary rebit state, and they span the space $\Sym^d(\reals^2)$. The mixed symmetric states are given by the positive unit-trace operators in $\LS(\Sym^d(\reals^2)) \simeq \LSH(\Sym^d(\comp^2))$. This isomorphism follows from the fact that $\Sym^d(\comp^2)$ is the complexification of $\Sym^d(\reals^2)$ and that $\LS(\reals^d) \simeq \LSH(\reals^d \otimes \comp)$, as shown in Lemmas~\ref{lem:realcompsym} and~\ref{lem:LSLSHcomp}. 

Now consider the orbit of a symmetric two-rebit pure state $\ket{\psi(\theta)}^{\otimes 2}$, where $\ket{\psi(\theta)}$ defined in \cref{eq:rebitorbit}. Explicitly, $\ketbra{\psi(\theta)}{\psi(\theta)}^{\otimes 2} \in \LSH(\Sym^d(\comp^2)) \subset \Sym^2(\comp^2) \otimes \Sym^2(\comp^2)$ is
\begin{align}
    \ketbra{\psi(\theta)}{\psi(\theta)}^{\otimes 2} = 
    \frac{1}{4} \begin{pmatrix}
        1 & e^{i \theta} & e^{i \theta} & e^{i 2 \theta} \\
        e^{- i \theta} & 1 & 1 & e^{i \theta} \\
        e^{- i \theta} & 1 & 1 & e^{i \theta} \\
        e^{- 2 i \theta} & e^{- i \theta} & e^{- i \theta} & 1 
    \end{pmatrix} . 
\end{align}
Compare this to the orbit $U_\theta \ketbra{\psi}{\psi}U_\theta^\dagger \in \LH(\comp^3)$ defined in \cref{eq:jorb} for $J =1$, where
\begin{align}
     U_\theta \ket{\psi}= \frac{1}{\sqrt{3}} (e^{-i \theta} \ket{-1} + \ket 0 + e^{i \theta} \ket 1 ),
\end{align}
and
\begin{align}
    U_\theta \ketbra{\psi}{\psi}U_\theta^\dagger = \frac{1}{3}
 \begin{pmatrix}
     1 & e^{i \theta} & e^{2 i\theta} \\
     e^{-i \theta} & 1 & e^{i \theta} \\
      e^{-2 i\theta} &  e^{-i \theta} & 1
 \end{pmatrix} .
 \end{align}
There exists an invertible linear map that maps $\ketbra{\psi(\theta)}{\psi(\theta)}^{\otimes 2}$ to  $U_\theta \ketbra{\psi}{\psi}U_\theta^\dagger$ which can be constructed as follows:
\begin{align}
    &L \ketbra{\psi(\theta)}{\psi(\theta)}^{\otimes 2} L^\top =  U_\theta \ketbra{\psi}{\psi}U_\theta^\dagger, \\
    &L =\sqrt{\frac 4 3} \begin{pmatrix}
        1 & 0 & 0 & 0 \\
        0 & \frac{1}{2} & \frac{1}{2} & 0 \\
        0 & 0 & 0 & 1
    \end{pmatrix}.
\end{align}
The inverse of this map is given by
\begin{align}
    &  M U_\theta\ketbra{\psi}{\psi}U_\theta^\dagger M^\top= \ketbra{\psi(\theta)}{\psi(\theta)}^{\otimes 2}, \\
    &M =\sqrt{\frac 3 4}\begin{pmatrix}
        1 & 0 & 0  \\
        0 & 1 & 0 \\
        0 & 1 & 0 \\
        0 & 0 & 1
    \end{pmatrix}.
\end{align}
This shows that the convex hulls of the two orbits are isomorphic as convex sets. This entails that the space of linear functionals that map every element $\ketbra{\psi(\theta)}{\psi(\theta)}^{\otimes 2}$ into the interval $[0,1]$ is isomorphic to $\mR_1$. Thus, for every $P \in \mR_1$, there exists a linear operator $W \in \LSH(\Sym^2(\comp^2))$ and therefore also in $\LS(\Sym^2(\mathbb{R}^2))$ such that 
\begin{align}
    P(+|\theta) = \Tr(W \ketbra{\psi(\theta)}{\psi(\theta)}^{\otimes 2}) .
\end{align}
The set of linear operators $W$ such that $ \Tr(W \ketbra{\psi(\theta)}{\psi(\theta)}^{\otimes 2})\geq 0$ for all $\theta$ are two-rebit symmetric entanglement witnesses. Thus, the cone generated by $\mR_1$, is isomorphic to the cone of two-rebit symmetric entanglement witnesses.

The following theorem generalizes the above observation to arbitrary $J$:

\begin{theorem}\label{thm:symmetric_entanglement}
    Every $P \in \mR_J$ can be realized as
    \begin{align}
        P(+|\theta) = \Tr(\ketbra{\psi(\theta)}{\psi(\theta)}^{\otimes 2 J} E_+),
    \end{align}
with $E_+$ an operator in $\LS(\Sym^{2J}(\mathbb{R}^2))$, the symmetric operators on the symmetric subspace of $2J$ rebits, such that $\Tr(\ketbra{\psi(\theta)}{\psi(\theta}^{\otimes 2 J} E_+) \in [0,1]$. 
\end{theorem}
This theorem is proven in Appendix~\ref{app:symmetric_entanglement}.

The possible operators $E_+$ include positive operators in $\LS(\Sym^{2J}(\mathbb{R}^2))$, which correspond to standard POVM elements on $2J$ rebits. However, the possible operators $E_+$ also include non-positive operators such as rebit symmetric entanglement witnesses. A $d$-rebit symmetric entanglement witness $W\in \LS(\Sym^d(\mathbb{R}^2))$ is an operator defined as:
\begin{align}
    \bra{\psi}^{\otimes d} W \ket{\psi}^{\otimes d}  \geq 0\quad \mbox{for all }\psi\in\mathbb{R}^2.
\end{align}
In typical applications of entanglement witnesses, it is assumed that there exists at least one state $\rho$ such that ${\rm Tr}(\rho W)<0$, which must then be entangled. Here, however, we are using the notion of an entanglement witness in the generalized sense, such that it also includes $W$ that are non-negative on \textit{all} symmetric states.
Thus, we obtain the following corollary:
\begin{corollary}
    The cone generated by  $\mR_J$ is isomorphic to the set of $2J$-rebit symmetric entanglement witnesses. 
\end{corollary}

The fact that $\mathcal{Q}_1=\mathcal{R}_1$, but $\mathcal{Q}_{3/2}\subsetneq \mathcal{R}_{3/2}$ can thus be interpreted as follows: all correlations (in $\theta$) generated by two-rebit symmetric entanglement witnesses can also be generated by proper two-rebit measurement operators (and similarly for zero or one rebits, because $\mathcal{Q}_0=\mathcal{R}_0$ and $\mathcal{Q}_{1/2}=\mathcal{R}_{1/2}$). However, the analogous statement for three rebits is false.

There is a compelling analogy of this behavior to the study of Bell correlations: all non-signalling correlations on \textit{pairs} of quantum systems are realizable within quantum theory~\cite{Barnum}, but this is not true for all non-signalling correlations on \textit{triples} of quantum systems~\cite{Acin}. The proof of this uses the fact that non-signalling correlations of quantum systems can always be generated by entanglement witnesses, regarded as a generalization of the notion of quantum states, which is yet another similarity to our result above.

\section{Conclusions and outlook}\label{sec:conc}

In this paper, we have introduced a notion of ``rotation boxes'', describing all possible ways in which measurement outcome probabilities could respond to spatial rotations around a fixed axis, in any covariant physical theory. We have thoroughly analyzed the resulting notion of spin-bounded correlations, and have demonstrated a variety of interesting results and applications. First, for the prepare-and-measure scenario, we have shown  that, for spin $J\in\{0,1/2\}$ systems, quantum theory predicts the same observable correlations as the most general physics consistent with the SO(2)-symmetry of the setup. For scenarios with two outcomes, the same is also true for the spin-$1$ case, although it remains an open questions as to whether this generalizes  to any number of outcomes.

However, for spin $J\geq3/2$, we have demonstrated a gap between quantum and more general predictions; we have derived a Tsirelson-type inequality and constructed an explicit counterexample consistent with general rotation boxes, but inconsistent with quantum rotation boxes. Moreover, we have presented a family of GPT systems that generate these ``post-quantum'' correlations. On the one hand, this result could hint at possible  probabilistic phenomena consistent with spacetime geometry that,  if indeed observed, would not be consistent with quantum theory. On the other hand, it is conceivable that the gap closes when we consider the full Lorentz or Poincar\'e group, which would thus reproduce crucial predictions of quantum theory from spacetime principles alone.  For $J\rightarrow\infty$, we have shown that every continuous rotational correlation can be approximated arbitrarily well by finite-$J$ quantum systems. 

Given the theoretical gap between quantum and more general  rotational correlations, we have presented a metrological game in which general spin-$3/2$ resources outperform all quantum ones, demonstrating a post-quantum advantage. We have further applied our framework to Bell scenarios, building on previous results. First, we have demonstrated why the ``local unbiasedness'' assumption introduced in~\cite{Garner} is crucial to recover the $(2,2,2)$-quantum Bell correlations from the no-signalling set, and that it has a geometric interpretation relating the outputs to the inputs of the box. Second, we have clarified the ``rotational invariance'' assumption used in~\cite{Nagata}, from which the authors derive indirect witnesses of multipartite Bell nonlocality. In particular, we argued that their assumption actually expresses the statement that all local subsystems are spin-$1/2$ (quantum or otherwise), and therefore that is does not rely on the validity of quantum theory.

In addition to addressing foundational questions, our work offers several interesting applications to explore in future work, such as the semi-device-independent analysis of experimental data. For instance, recent experiments have successfully probed Bell nonlocality in many-body systems like Bose-Einstein condensates, using so-called Bell correlation witnesses~\cite{Schmied}. These witnesses have the advantage of being experimentally accessible by treating the Bose-Einstein condensate as a single party in which collective observables can be measured. However, a disadvantage of this approach is that it requires additional assumptions compared to a typical Bell test, namely the validity of spin-algebra in quantum mechanics and trust in the measurements, making it device-dependent. Our framework is a suitable candidate for providing weaker assumptions for carrying out semi-device-independent analysis of the observed experimental data, in particular in situations where the experimental parameters are spatiotemporal in nature.

Another interesting application would be to devise self-testing-inspired protocols via rotations. Typical self-testing~\cite{scarani2019bell,Supic} protocols are tailored to specific pairs of states and measurements, but do not tell us how to operationally implement other valid measurements on the state.
It would be interesting to explore whether semi-device-independent self-testing-inspired protocols can be devised where the inputs correspond to directions in physical space (on which the rotation group acts), and the outputs are angular-momentum-valued physical quantities (instead of abstract labels), in order to not only certify a certain state and the implemented  measurements, but also certify the state with all other valid measurements in different directions. 

A further direction to explore would be whether one can carry a similar study than the one in this manuscript by replacing the local spin bound by a local energy bound (for instance, making use of the Mandelstam-Tamm quantum speed limit~\cite{MandelstamTamm,Hoernedal}). The settings would then not correspond to two different directions in space, but to two different time intervals according to which we let the systems evolve locally. Formally, this would replace the group of rotations ${\rm SO}(2)$ of this paper by the time translation group $(\mathbb{R},+)$. More generally, it will be a natural next step to consider other groups of interest, such as the full rotation group ${\rm SO}(3)$ or the Lorentz group, and to see which novel statistical phenomena arise from the non-commutativity and other strutural properties of these groups. 

Furthermore, the interplay of entanglement and nonlocality with the group theoretic structure deserves more study. The paradigmatic example is that of spin-$1/2$ fermions obtaining a $(-1)$ phase on $(2\pi)$-rotations, visible in the presence of initial entanglement. This already demonstrates one surprising insight, potentially amongst others still waiting to be discovered, at the intersection of probabilistic and spacetime structure.

\section*{Declarations}
\begin{center}
\textbf{Funding}	
\end{center}
We acknowledge support from the Austrian Science Fund (FWF) via project P 33730-N. This research was supported in part by Perimeter Institute for Theoretical Physics. Research at Perimeter Institute is supported by the Government of Canada through the Department of Innovation, Science, and Economic Development, and by the Province of Ontario through the Ministry of Colleges and Universities. A. A. also acknowledges financial support by the ESQ Discovery programme (Erwin Schr{\"o}dinger Center for Quantum Science \& Technology), hosted by the Austrian Academy of Sciences ({\"O}AW).

\begin{center}
\textbf{Data availability statement}
\end{center}
Data sharing is not applicable to this article as no data sets were generated or analyzed during the current study.

\begin{center}
\textbf{Conflict of interest}	
\end{center}
The authors have no relevant financial or non-financial interests to disclose.

\restoretoc

\section*{Appendices}

\appendix
\addtocontents{toc}{\string\tocdepth@munge}

\section{Background material}

\subsection{Finite-dimensional projective representations of $\SO(2)$}

Theorem 16.47 of~\cite{Hall2013} states that given a compact group $G$ with universal cover $\tilde G$, a covering map $\Phi: \tilde G \to G$, and a finite-dimensional projective unitary representation $\Pi: G \to \PU(\mathcal{H})$, there is a unitary representation $\Sigma: \tilde G \to \U(\mathcal{H})$ such that $\Pi \circ \Phi = Q \circ \Sigma$, where $Q$ is the quotient homomorphism $Q: \U(\mathcal{H}) \to \PU(\mathcal{H})$, $Q: U \mapsto U /\{e^{i \theta}\}$ for $\theta \in \reals$. Any such $\Sigma$ is irreducible if and only if $\Pi$ is irreducible.

If  $G = \SO(2)$,  then $\tilde G = (\reals, +)$.  The irreducible unitary representations $\reals \mapsto \U(1)$ are given by $x \mapsto e^{itx}$ with $t \in \reals$. These are projective representations of $\SO(2)$ and are projectively equivalent to the trivial representation $x \to 1$. Thus the only irreducible projective representation of $\SO(2)$ is the trivial representation. Equivalently, unitary projective irreducible representations are maps $\SO(2) \to \PU(1)$, and $\PU(1)$ is just the trivial group.

We now characterize reducible projective representations of $\SO(2)$.

\begin{lemma}\label{lem:projSO2charac}
    Any finite-dimensional projective representation of $\SO(2)$ can be written in the form of Equation~\eqref{eqRepDef}:
    \begin{align}
          U_\theta=\bigoplus_{j=-J}^J \mathbf{1}_{n_j} e^{ij\theta},
    \end{align}
    where $J \in \{0,\frac{1}{2}, 1,...\}$ and $n_j \in \nats_0$.
\end{lemma}
\begin{proof}
    A generic representation $\reals \to \U(\mathcal{H})$ is of the from
\begin{align}
    x \mapsto e^{i \diag(j_1,...,j_n) x} \quad (j_i \in \reals)
\end{align}
in some basis, where there can be repeated entries and, without loss of generality, $i \geq k \implies j_i \geq j_k$. 

The requirement that it is a projective representation of $\SO(2)$ entails that
\begin{align}
    e^{i \diag(j_1,...,j_n) 2 \pi} = e^{i \phi },
\end{align}
for some $\phi\in\mathbb{R}$, which entails
\begin{align}
    2 \pi j_i + 2 \pi q_i &= \phi  \quad\mbox{for some }q_i \in \ints, \\
    j_i + q_i& = \frac{\phi}{2 \pi}.
\end{align}
Thus, $j_i - j_k = q_i - q_k$, and the difference $  j_i - j_k$ is integer-valued for all $i,j$.

Setting $j_1 = \phi_0$ and $j_i = j_1 + k_i$ with $k_i \in \nats_0$, the projective representation is of the form:
\begin{align}
    e^{i \phi_0} e^{i \diag(0,k_2,...,k_n)},
\end{align}
and can be characterized by a list of non-negative integers $\{k_2,...,k_n\}$. We are however interested in special unitary representations and can transform as follows:
\begin{align}
    e^{i \phi_0} e^{i \diag(0,k_2,...,k_n)}  \mapsto  e^{i (\phi_0 + \frac{k_n}{2})}  e^{i \diag(-\frac{k_n}{2},k_2 - \frac{k_n}{2},...\frac{k_n}{2})}.
\end{align}
Thus, every projective unitary representation can be characterized by a list of integers or half-integers  $\{k_1',...,k_n'\} = \{-\frac{k_n}{2},k_2 - \frac{k_n}{2},...\frac{k_n}{2}\}$, where $k_1' = - k_n'$. 
\end{proof}

This lemma entails that any projective representation of $\SO(2)$ is characterized by a list $\{j_1,...,j_n\}$ of integers or half-integers.

\begin{lemma}
    Projective representations of $\SO(2)$  of the form $\{-J, -J +1,..., J-1, J\}$ with $J \in \nats$ are also representations of $\SO(2)$, while those with $J \in \nats/2$ are purely projective representations.
\end{lemma}

\begin{proof}
This is because $e^{i \diag(-J,...,J) 2 \pi}$ equals $\mathbf{1}$ for integer $J$ and $-\mathbf{1}$ for half-integer $J$.
\end{proof}

\subsection{Real projective representations of $\SO(2)$}

Real irreducible representations of $\SO(2)$ are labelled by non-negative integers $k\in\mathbb{N}_0$ and are given by the trivial representation for $k = 0$ and by
\begin{align}\label{eq:SO2realirrep2}
    \begin{pmatrix}
        \cos(k \theta) & -\sin(k \theta) \\
        \sin(k \theta) & \cos(k \theta) 
    \end{pmatrix}
\end{align}
for $k \in \nats$. Thus, a real representation of $\SO(2)$ is labelled by a list of non-negative integers $\{k_1,..., k_n\}$. We note that for $k$ a half-integer, \cref{eq:SO2realirrep2} defines a real irreducible projective representation of $\SO(2)$.

\begin{lemma}
    The complexification of the real irreducible projective representation $\{k\}$ of $\SO(2)$ with $k\neq 0$ integer or half-integer is the complex reducible protective representation $\{k,-k\}$.
\end{lemma}
\begin{proof}
The real matrix
\begin{align}
     \begin{pmatrix}
         \cos(k \theta) & -\sin(k \theta) \\
         \sin(k \theta) & \cos(k \theta)
     \end{pmatrix} ,
     \end{align}
acting on $\comp^2$ can be diagonalized:
    \begin{align*} 
    \begin{pmatrix}
         \cos(k \theta) & -\sin(k \theta) \\
         \sin(k \theta) & \cos(k \theta)
     \end{pmatrix} \mapsto 
     \begin{pmatrix}
         e^{i k \theta} & 0 \\
         0 &  e^{-i k \theta}
     \end{pmatrix} .
\end{align*}
\vskip -1.5em
\end{proof}
Our general framework of rotation boxes implies that we have real representations of ${\rm SO}(2)$, because the space of ensembles of boxes (the vector space carrying the GPT system which represents it) will always be a vector space over $\mathbb{R}$. This is also true for projective representations in quantum theory, where ${\rm SO}(2)$ acts on the vector space of Hermitian matrices that contains the density matrices. However, the following lemma will be useful when discussing \textit{quantum theory over the real numbers $\mathbb{R}$}:
\begin{lemma}
    Representations of $\SO(2)$  $\{-J, -J +1,..., J-1, J\}$ with integer $J$ are also real representations of $\SO(2)$ $\{0,...,J\}$, while projective representations $\SO(2)$  $\{-J, -J +1,..., J-1, J\}$  with half-integer $J$ are real projective representations $\{\frac{1}{2},..., J\}$.
\end{lemma}
\begin{proof}
    Consider the following change of basis:
    \begin{align}
         \begin{pmatrix}
         e^{i k \theta} & 0 \\
         0 &  e^{-i k \theta}
     \end{pmatrix} \mapsto 
    \begin{pmatrix}
         \cos(k \theta) & -\sin(k \theta) \\
         \sin(k \theta) & \cos(k \theta)
     \end{pmatrix} .
\end{align}
Thus, for integer $J$:
\begin{align}
     e^{i \diag(-J,...,J) } \mapsto \bigoplus_{j=0}^J      \begin{pmatrix}
         \cos(k \theta) & -\sin(k \theta) \\
         \sin(k \theta) & \cos(k \theta) 
     \end{pmatrix} ,
\end{align}
which is a real representation of $\SO(2)$. 

For half-integer $J$:
\begin{align}
     e^{i \diag(-J,...,J) } \mapsto \bigoplus_{j=\frac{1}{2}}^J      \begin{pmatrix}
         \cos(k \theta) & -\sin(k \theta) \\
         \sin(k \theta) & \cos(k \theta)
     \end{pmatrix} ,
\end{align}
which is a real projective representation of $\SO(2)$. 
\end{proof}

\subsection{Representation-theoretic background}

We introduce some necessary representation-theoretic concepts before proceeding with the proofs. Here vector spaces $V$ are isomorphic to $\comp^n$, unless otherwise stated. A representation of $G$ is a homomorphism $\rho: G \to \GL(V)$ with the general linear group $\GL(V)$ the group of automorphisms on $V$. We note that we do not require the representation to be faithful (i.e the map is not required to be injective) since we are interested in finite-dimensional unitary representations  of $(\reals, +)$, which is the universal cover of $\SO(2)$. The vector space $V$ is the carrier space or representation space of $\rho$; however, we sometimes call it the representation. 
 
When two representations $\rho: G \to \GL(V)$ and $\sigma: G \to \GL(W)$ are isomorphic, we write $\rho \simeq \sigma$, or, when the context is clear, $V \simeq W$. An isomorphism of representations is given by an invertible linear map $L: V \to W$ which is equivariant: $\sigma(g) L(v) = L(\rho(g) v)$.

Given a representation $\rho: G \to \GL(V)$, we denote by $\bar \rho: G \to \GL(\bar V)$ the complex conjugate representation and by $\rho^*: G \to \GL(V^*)$ the dual representation. For finite-dimensional representations, we have $\bar \rho \simeq \rho^*$.

We denote the space of linear maps from $V$ to $W$ by $\mL(V,W)$. It carries a representation $\tau: G \to \GL(\mL(V,W))$ given by $(\tau(g)(M))(v) = \sigma(g) M(\rho(g^{-1}) v)$.

Given a complex vector space $V$, restricting scalar multiplication from $\comp$ to $\reals$ defines the real vector space $V_\reals$, known as the realification of $V$, where $\dim_\reals(V_\reals) = 2 \dim_\comp(V)$. Given a representation $\rho: G \to \GL(V)$, the space $V_\reals$ carries a real representation $\rho_\reals: G \to \GL(V_\reals, \reals)$~\cite{itzkowitz_note_1991}. 

Given a real vector space $W$ with basis $\{e_i\}_i$, it can be complexified to obtain $W_\comp = \comp \otimes_\reals W$ with basis $\{1 \otimes_\reals e_i\}_i$.  Given a real representation $\rho: G \to \GL(W,\reals)$, the complexification of the representation $\rho$ is $\rho_\comp: G \to \GL(W_\comp, \comp)$ defined as $\rho_\comp(g) (1 \otimes e_i)= 1 \otimes \rho(g)(e_i)$~\cite{itzkowitz_note_1991}.
\begin{definition}[Real structure]
    Given a complex vector space $V$, a real structure $j$ is an antilinear map $j: V \to V$ which is an involution: $j \circ j = \id_V$. If $V$ carries a representation $\rho: G \to \GL(V)$, then the representation $\rho$ carries the real structure $j$ if $j$ is equivariant: $\rho(g) j(v) = j (\rho(g) v)$.
\end{definition}
Given a complex vector space $V$ with a real structure $j$, an arbitrary $v \in V$ can be expressed as $v = v^{j =+1} + v^{j= -1}$ where $v^{j= +1} = \frac{v + j(v)}{2}$ and $v^{j = -1} = \frac{v - j(v)}{2}$. Hence the realification $V_\reals$ decomposes into the direct sum $V \simeq_\reals V^{j =1} \oplus_\reals V^{j = -1}$  where  $V^{j=\pm 1} := \{v \in V| j(v)  = \pm v\}$.

Equivariance of $j$ implies that the real subspaces $V_{j = \pm 1}$ are closed under the action of $\rho(g)$, and hence  $\rho_\reals$ decomposes into the direct sum of real representations  $\rho^{j = +1} \oplus_\reals \rho^{j = -1}$, where  $\rho^{j = +1} \simeq_\reals \rho^{j = -1}$~\cite[p.95]{broecker2003representations}.
\begin{lemma}\label{lem:real_comp_rep}
    Given a representation $\rho: G \to \GL(V)$ with real structure $j$, we have $(\rho^{j = +1})_\comp \simeq \rho$. 
\end{lemma}
\begin{proof}
Given a complex representation $\rho$ with real structure $j$, we define the map $\Pi^+: (\rho,j) \mapsto \rho^{j = +1}$, where $\rho^{j =+1}$ is the real representation defined above. 

Given a real representation $\sigma$, we define the map $\Pi: \sigma \mapsto (\sigma_\comp, k)$, where $\sigma_\comp$ is the complexification of $\sigma$ and the real structure $k$ is defined as $k (z \otimes v) = \bar z \otimes v$. 

From~\cite[p.94]{broecker2003representations}, it follows that $\Pi  \Pi^+$ is the identity morphism, which implies
\begin{align}
    (\rho,j)\simeq \Pi \Pi^+(\rho,j) \simeq ((\rho^{j = 1})_\comp, k).
\end{align}
Defining the map $\Gamma: (\rho, j) \mapsto \rho$, the claim $e_\reals^\comp = r_+ e_+$ of~\cite[Proposition (6.1)]{broecker2003representations} can be expressed in our notation as
\begin{align}
    \Gamma \Pi (\rho) = \rho_\comp.
\end{align}

Combining the above two equivalences gives
\begin{align}
    \rho \simeq \Gamma(\rho,j) \simeq \Gamma \Pi  \Pi^+(\rho,j) \simeq \Gamma ((\rho^{j = 1})_\comp, k) \simeq (\rho^{j = 1})_\comp ,
\end{align}
which proves the lemma.
\end{proof}

\begin{lemma}\label{lem:real_rep_sym}
    Given a representation $\rho: G \to \GL(V)$, the real subspace $\Sym(V \otimes \bar V) \subset V \otimes \bar V$ carries the real representation $\rho': G \to \GL(\Sym(V \otimes \bar V))$, whose complexification is isomorphic to $\rho(g) \otimes \bar \rho(g)$.
\end{lemma}
\begin{proof}
Given a linear space $V$ and its complex conjugate space $\bar V$, where $\bar V$ has the same elements, but scalar multiplication given by $\alpha \star v = \bar \alpha v$, we can define the tensor product space $W \simeq V \otimes \bar V$, where scalar multiplication is defined as
\begin{align}
    \alpha \star_W (v \otimes w) = \alpha v \otimes w = v \otimes \alpha \star w = v \otimes \bar \alpha w.
\end{align}
This space carries a representation $\rho(g) \otimes \bar \rho(g)$. 
Consider the swap map $S:W\to W$, $v\otimes w \mapsto w\otimes v$. This map is anti-linear since
\begin{align}
    S(\alpha \star_W(v \otimes w)) &= S(\alpha v \otimes w) = w \otimes \alpha v \\
    &= w \otimes \bar \alpha \star v =  \bar \alpha \star_W S(v \otimes w).
\end{align} 
Moreover, $S$ is equivariant:
\begin{align}
    S(\rho(g) \otimes \bar \rho(g) (v \otimes w)) 
    = \rho(g) \otimes \bar \rho(g) S(v \otimes w).
\end{align}
The existence of an equivariant anti-linear map $S:  V \otimes \bar V \to V \otimes \bar V $  entails that $V \otimes \bar V$ has a real structure given by $S$. The $+1$ eigenspace of $S$ is ${\rm Sym}(V \otimes \bar V):= \{w \in V \otimes \bar V| S(w) = w\}$. By~\Cref{lem:real_comp_rep}, $\Sym(V \otimes \bar V)$ carries a real representation $\rho': G \to \GL(\Sym(V \otimes \bar V))$, whose complexification is $\rho$.
\end{proof}
\begin{lemma}\label{lem:herm_real_rep}
    The real subspace of Hermitian operators on $\mathcal{H}$, $\LH(\mathcal{H}) \subset \mL(\mathcal{H})$, carries a real representation of $G$.
\end{lemma}
\begin{proof}
The Hermitian adjoint of a map $M \in \mL(\mathcal{H}, \mH')$ is the map $M^* \in \mL(\mH',\mH)$ defined by $\langle Mv ,w \rangle_{\mH'} = \langle v,M^* w \rangle_{\mH}$. The resulting map (``adjoint map'') $^*{ }: \mL(\mH,\mH') \to \mL(\mH',\mH)$, $M\mapsto M^*$ is anti-linear. Moreover, it is equivariant:
\begin{align*}
    \langle v, (\sigma(g)  M \circ \rho(g^{-1}))^* w \rangle_{\mH} = \langle \sigma(g)  M \circ \rho(g^{-1}) v, w \rangle_{\mH'} \\
    = \langle M \circ \rho(g^{-1}) v ,  \sigma(g^{-1}) w \rangle_{\mH'}  = \langle \rho(g^{-1}) v, M^* \circ \sigma(g^{-1}) w \rangle_{\mH} \\
    = \langle  v, \rho(g) M^* \circ \sigma(g^{-1}) w \rangle_{\mH}.
\end{align*}
Thus, the $+1$ eigenspace $\LH(\mH):=\{M = M^*\,\,|\,\, M \in \mL(\mH,\mH)\}$ carries a real representation of $G$.
\end{proof}
\begin{lemma}\label{lem:herm_sym_equiv}
    $\LH(\mH) \simeq \Sym(\mH \otimes \bar \mH)$ as real representations.
\end{lemma}

\begin{proof}
We define an equivariant invertible linear map $\Sym(\mH \otimes \bar \mH) \to \LH(\mH)$. First, we define the map
\begin{align}
    L&: \mH \otimes \bar \mH \to \mL( \mH, \mH) \simeq  \mH \otimes \mH^* ,\\
    &\, e_i \otimes e_j \mapsto e_i \otimes e_j^*
\end{align}
which is a group representation isomorphism $ \mH \otimes \bar  \mH \to  \mH \otimes  \mH^*$. We now show it maps $\Sym( \mH \otimes \bar  \mH)$ to $\LH( \mH)$ and hence is an isomorphism of real representations:
\begin{align*}
    L(S(e_i \otimes e_j)) &= L (e_j \otimes e_i) = e_j \otimes e_i^* = (e_i \otimes e_j^*)^* \\
   &= (L(e_i \otimes e_j))^*.
\end{align*}
Hence $L(S(\cdot)) = L(\cdot)^*$, which implies that for all $w \in \Sym( \mH \otimes \bar  \mH)$,we have $L(w) = L(w)^*$. Conversely, for all $w \in  \mH \otimes  \mH^*$ such that $w = w^*$, we have $L^{-1}(w) \in \Sym( \mH \otimes \bar  \mH)$.
\end{proof}

\subsection{Relevant vector space isomorphisms}

\begin{lemma}\label{lem:LSLSHcomp}
    Given a real Hilbert space $ \mH$ and its complexification $ \mH'$, the space $\LS( \mH)$ of symmetric operators on $ \mH$ is isomorphic to the space $\LSH( \mH')$ of symmetric Hermitian operators on $ \mH'$.
\end{lemma}
\begin{proof}
    $ \mH \simeq \reals^n$ and its complexification $ \mH' \simeq \comp^{n}$. Fixing a basis, an operator $O \in \LH( \mH')$ is symmetric if and only if its entries are real-valued. Thus the symmetric Hermitian operators on $ \mH'$ are given by the $n \times n$ real symmetric matrices and therefore isomorphic to $\LS( \mH)$. 
\end{proof}

\begin{lemma}\label{lem:realcompsym}
    The space $\Sym^d(\comp^2)$ is the complexification of $\Sym^d(\reals^2)$. 
\end{lemma}
\begin{proof}
   Consider a basis $\{\ket 0,\ket 1\}$ for both $ \mH\in\{\reals^2,\comp^2\}$. The symmetric group $\Sigma_d$ acts on $ \mH^{\otimes d}$ by permuting the tensor factors. A basis for $\Sym^d( \mH)$ is $\{\ket i\}_0^{d}$, with
\begin{align}
    \ket i = \sum_{x \in \{0,1\}^d|H(x) = i} \ket x,
\end{align}
where $H(x)$ is the Hamming weight of the bit string $x \in \{0,1\}^d$. Thus, there is a common basis $\{\ket i\}_0^{d}$ for $\Sym^d(\reals^2)$ and $\Sym^d(\comp^2)$, showing that $\Sym^d(\comp^2)$ is the complexification of  $\Sym^d(\reals^2)$.
\end{proof}
\begin{corollary}
    The real vector space of symmetric operators $\LS(\Sym^d(\reals^2))$ is isomorphic to the real vector space of symmetric Hermitian operators $\LSH(\Sym^d(\comp^2))$.
\end{corollary}

\subsection{Relevant $\SO(2)$ group representation isomorphisms}

In the following, $\comp^2$ carries the $\SO(2)$ projective representation $\{-\frac{1}{2},\frac{1}{2}\}$, and $\reals^2$ carries the real projective representation $\{\frac{1}{2}\}$. $W \simeq V$ indicates that the representation of $\SO(2)$ on $V$ is isomorphic to the representation of $\SO(2)$ on $W$.
We note that the representation $\{-\frac{1}{2},\frac{1}{2}\}$ is isomorphic to its conjugate and thus to its dual and  hence is known as self-dual.
\begin{lemma}
    $\Sym^d(\comp^2)$ carries the representation $\{-\frac{d}{2}, ..., \frac{d}{2}\}$ of ${\rm SO}(2)$.
\end{lemma}
\begin{proof}
     A basis for  $\Sym^d(\comp^2)$ is $\{\ket k\}_{k=0}^{d}$, where
\begin{align}
    \ket k = \sum_{x \in \{0,1\}^d|H(x) = k} \ket x .
\end{align}
The action of $\diag(e^{i \frac{\theta}{2}}, e^{-i \frac{\theta}{2}})^{\otimes d}$  on this basis is:
\begin{align}
    \ket k = \sum_{x \in \{0,1\}^d|H(x) = k} \ket x \mapsto \sum_{x \in \{0,1\}^d|H(x) = k} e^{i \frac{k}{2}} \ket x.
\end{align}
Thus, $\Sym^d(\comp^2)$ carries the representation $\{-\frac{d}{2}, ..., \frac{d}{2}\}$.
\end{proof}
\begin{corollary}\label{cor:self_dual}
    The representation $\Sym^d(\comp^2)$ is self-dual.
\end{corollary}
\begin{lemma}
    $\Sym^{2d}(\reals^2)$ carries the real representation $\{0,1,...,d \}$ of ${\rm SO}(2)$.
\end{lemma}
\begin{proof}
    The space $\reals^2$ carries a real irreducible projective representation $\{\frac{1}{2}\}$. The complexification $\comp^2$ carries a complex  irreducible projective representation $\{-\frac{1}{2},\frac{1}{2}\}$. $\Sym^d(\comp^2)$ carries the complex projective  representation $\{-\frac{d}{2},-\frac{d}{2}+1 ,..., \frac{d}{2}\}$ and $\Sym^{2d}(\comp^2)$ carries the complex representation $\{-d,-d+1,...,d-1,d\}$. Thus, $\Sym^{2d}(\mathbb{R}^2)$ carries a real representation $\{0,...,d\}$.
\end{proof}

\section{Proofs for Section~\ref{sec:rot_box_frame}}

\subsection{Proof of Theorem~\ref{thm:QJ_corr_form}}\label{app:proof_QJ_corr_form}
The following lemma has been established in a different context by Miguel Navascu\'es (unpublished). The conditions (ii) and (iii) in this lemma are a priori inequivalent if $G$ has degenerate spectrum, and the distinction of these two cases will become useful in the proof of Lemma~\ref{Lem1N}.
\begin{lemma}[Miguel Navascu\'es~\cite{Miguel}]
\label{LemMiguel}
Let $G=G^\dagger$ be an observable on a finite-dimensional Hilbert space, and let $P(a|\cdot):\mathbb{R}\to\mathbb{R}$, $\theta\mapsto P(a|\theta)$, $a\in \mathcal{A}$ with $\mathcal{A}$ some finite set, be real functions. Then the following statements are equivalent:
\begin{itemize}
	\item[(i)] There exists a quantum state $\rho$ and a POVM $\{E_a\}_{a\in \mathcal{A}}$ such that
		\[
		   P(a|\theta)={\rm Tr}(e^{iG\theta}\rho e^{-i G\theta} E_a).
		\]
  \item[(ii)] There exists an eigenbasis $\{|n\rangle\}_n$ of $G$, a probability distribution $\{p_n\}_n$ over the eigenvectors $|n\rangle$, and positive semidefinite operators $\{S_a\}_{a\in \mathcal{A}}$ with $\sum_{a\in \mathcal{A}}S_a =\sum_n p_n |n\rangle\langle n|$ such that
	\[
        P(a|\theta)=\langle +|e^{-iG\theta}S_a e^{iG\theta}|+\rangle,
	\]
	where $|+\rangle:=\sum_n|n\rangle$ (note that this vector is not a normalized state).
	\item[(iii)] For every eigenbasis $\{|n\rangle\}_n$ of $G$, there exists a probability distribution $\{p_n\}_n$ over the eigenvectors $|n\rangle$, and positive semidefinite operators $\{S_a\}_{a\in \mathcal{A}}$ with $\sum_{a\in \mathcal{A}}S_a =\sum_n p_n |n\rangle\langle n|$ such that
	\[
        P(a|\theta)=\langle +|e^{-iG\theta}S_a e^{iG\theta}|+\rangle.
	\]
\end{itemize}
Moreover, the state $\rho$ in (i) can always be chosen as a \textit{pure} state, with real non-negative amplitudes in any fixed choice of eigenbasis of $G$.
\end{lemma}
\begin{proof}
To prove (i)$\Rightarrow$(iii), write $\rho=\sum_{jk}\rho_{jk}|j\rangle\langle k|$ in an arbitrary eigenbasis of $G$ where $G|j\rangle=g_j|j\rangle$ (when $G$ is degenerate, there exist values $ i \neq j$ such that $g_i = g_j$). Note that
\begin{eqnarray*}
P(a|\theta)&=&\sum_{jk} \rho_{jk} e^{i(g_j-g_k)\theta} \langle k|E_a|j\rangle\\
&=& \langle +|e^{-iG\theta} S_a e^{iG\theta}|+\rangle,
\end{eqnarray*}
where $S_a$ is defined by its matrix elements
\[
   (S_a)_{kj}:=\rho_{jk}(E_a)_{kj}.
\]
In other words, $S_a=\rho^\top\circ E_a$ for the Schur product $\circ$. Since the Schur product of positive semidefinite matrices is positive semidefinite, so is $S_a$. Moreover, since $\sum_{a \in \mA} E_a = \mathbf{1}$, $S:=\sum_{a\in\mathcal{A}} S_a$ satisfies $\langle k|S|j\rangle = \rho_{jk}\delta_{kj}$, i.e.\ it is a diagonal matrix with a probability distribution on its diagonal (namely, the diagonal elements of $\rho$).

The implication (iii)$\Rightarrow$(ii) is trivial. To prove (ii)$\Rightarrow$(i), define $|\psi\rangle:=\sum_n \sqrt{p_n}|n\rangle$, $\rho:=|\psi\rangle\langle\psi|$ (which is a pure state), and
\[
   E_a:=M S_a M^\dagger + \frac 1 {|\mathcal{A}|} \Pi_0,
\]
where $M:=\sum_{n:p_n\neq 0} p_n^{-1/2} |n\rangle\langle n|$ and $\Pi_0:=\sum_{n:p_n=0}|n\rangle\langle n|$. Then we have $E_a\geq 0$ and $\sum_a E_a=\mathbf{1}$. Note that $p_n=0$ implies $\langle n|S_a|n\rangle=0$, and thus $\langle m|S_a|n\rangle=0$ for all $m$. Hence
\[
   {\rm Tr}(e^{iG\theta}\rho e^{-iG\theta} E_a)=\langle +|e^{-iG\theta} S_a e^{iG\theta}|+\rangle=P(a|\theta).
\]
This proves the converse and the claim that $\rho$ can always be chosen pure and with non-negative amplitudes in the eigenbasis $\{|n\rangle\}$.
\end{proof}
Now, for every $N\in\mathbb{N}$, consider the representation
\[
U_\theta^{(N)}:=\bigoplus_{j=-J}^J \mathbf{1}_N e^{ij\theta},
\]
and denote by $\mathcal{Q}_{J,\mathcal{A}}^{(N)}$ the quantum spin-$J$ correlations that can be obtained with a suitable state and measurement on the corresponding Hilbert space. Clearly, every representation of the form~(\ref{eqRepDef}) is embedded into some $U_\theta^{(N)}$ for $N$ large enough, and so
\[
\mathcal{Q}_J^\mathcal{A}\subset \bigcup_{N\in\mathbb{N}} Q_{J,\mathcal{A}}^{(N)}.
\]
The next lemma will show that, in fact, all the correlation sets are the same, i.e.\ $\mathcal{Q}_{J,\mathcal{A}}^{(N)}=\mathcal{Q}_{J,\mathcal{A}}^{(1)}$ for all $N$, and hence $\mathcal{Q}_J^\mathcal{A}\subset\mathcal{Q}_{J,\mathcal{A}}^{(1)}$. Since the converse inclusion is trivial, we obtain $\mathcal{Q}_J^\mathcal{A}=\mathcal{Q}_{J,\mathcal{A}}^{(1)}$, and Lemma~\ref{LemMiguel} shows that the representing state can always be chosen pure, and with non-negative real amplitudes in the given eigenbasis $\{|j\rangle\}$. This establishes the validity of Theorem~\ref{thm:QJ_corr_form}.
\begin{lemma}
\label{Lem1N}
We have $\mathcal{Q}_{J,\mathcal{A}}^{(1)}=\mathcal{Q}_{J,\mathcal{A}}^{(N)}$ for all $N\in\mathbb{N}$.
\end{lemma}
\begin{proof}
Let $N\in\mathbb{N}$ be an integer. On the Hilbert space which carries the representation $U_\alpha^{(N)}$, define an eigenbasis $\{|j,n\rangle\}_{-J\leq j \leq J,1\leq n \leq N}$ such that
\[
     U_\alpha^{(N)}|j,n\rangle=e^{ij\alpha} |j,n\rangle,
\]
i.e.\ where $j$ labels the $\SO(2)$ irrep and $n$ labels the multiplicity. For operators $X=\sum_{(j,m),(k,n)} X_{(j,m),(k,n)} |j,m\rangle\langle k,n|$ on that Hilbert space, define an associated operator $\tilde X$ on $\mathbb{C}^{2J+1}$ via $\tilde X_{j,k}:=\sum_{m,n} X_{(j,m),(k,n)}$ (regarding the Hilbert space as a tensor product space $AB$, where $A=\mathbb{C}^{2J+1}$ and $B=\mathbb{C}^N$, this is $\tilde X=\langle +|_B X |+\rangle_B$, with $|+\rangle_B:=\sum_{n=1}^N|n\rangle_B$). Let us first show that $X\geq 0$ implies $\tilde X\geq 0$. To this end, if $|\psi\rangle=\sum_j \psi_j |j\rangle$ is an arbitrary vector, set $\varphi_{j,m}:=\psi_j$ for all $m$ and $|\varphi\rangle:=\sum_{j,m}\varphi_{j,m}|j,m\rangle$. Then
\begin{eqnarray*}
\langle \psi|\tilde X|\psi\rangle&=& \sum_{jk} \overline{\psi_j} \tilde X_{j,k} \psi_k=\sum_{jkmn} \overline{\varphi_{j,m}} X_{(j,m),(k,n)} \varphi_{k,n}\\
&=& \langle \varphi|X|\varphi\rangle\geq 0.
\end{eqnarray*}
It is easy to see that if $X=\sum_{jm} p_{j,m} |j,m\rangle\langle j,m|$ with $\{p_{j,m}\}$ some probability distribution, then $\tilde X = \sum_j q_j |j\rangle\langle j|$, with $\{q_j\}$ another probability distribution (namely, $q_j=\sum_m p_{j,m}$).

Now suppose that $P\in\mathcal{Q}_{J,\mathcal{A}}^{(N)}$, i.e.\ there is a quantum state $\rho$ and a POVM $\{E_a\}_{a\in\mathcal{A}}$ on the total Hilbert space such that
\[
P(a|\theta)={\rm Tr}\left(U_\theta^{(N)} \rho\left(U_\theta^{(N)}\right)^\dagger E_a\right).
\]
Let $G$ be a generator such that $U_\theta^{(N)}=e^{iG\theta}$, and let $|+^{(N)}\rangle:=\sum_{j,m} |j,m\rangle$. According to Lemma~\ref{LemMiguel}, (i)$\Rightarrow$(iii), this implies that there are positive semidefinite matrices $S_a$ and a probability distribution $\{p_{j,m}\}$ with $\sum_{a\in\mathcal{A}} S_a=\sum_{j,m} p_{j,m}|j,m\rangle\langle j,m|$ such that
\begin{eqnarray*}
P(a|\theta)&=& \langle +^{(N)}(U_\theta^{(N)})^\dagger S_a U_\theta^{(N)}|+^{(N)}\rangle\\
&=& \sum_{jkmn} e^{i(k-j)\theta}(S_a)_{(j,m),(k,n)}\\
&=& \sum_{jk} e^{i(k-j)\theta} (\tilde S_a)_{j,k}=\langle +|U_\theta^\dagger \tilde S_a U_\theta |+\rangle.
\end{eqnarray*}
Thus, due to Lemma~\ref{LemMiguel} (ii)$\Rightarrow$(i), we have $P\in\mathcal{Q}_{J,\mathcal{A}}^{(1)}$. 

We conclude that $\mathcal{Q}_{J,\mathcal{A}}^{(N)}\subseteq \mathcal{Q}_{J,\mathcal{A}}^{(1)}$. Conversely, $\mathcal{Q}_{J,\mathcal{A}}^{(1)}\subseteq \mathcal{Q}_{J,\mathcal{A}}^{(N)}$ because the former can be trivially embedded into the latter by padding the states with zeroes and the POVM with constants that sum up to one.
\end{proof}

\subsection{Generalization of the rotation boxes SDP in \cref{rotationBoxSDP} to arbitrary number of outcomes}
\label{app:SDPgralOutcomes}

Here, we generalize the SDP methodology in \cref{rotationBoxSDP} to account for an arbitrary finite number of outcomes. Following the notation introduced in~\ref{def:spin_J_gen_cor}, let us denote the outcome set with outcomes as $\mathcal{A}=\{b_1,\ldots,b_n\}$ with $|\mathcal{A}|=n$ and its corresponding set of spin-$J$ correlations as $\mathcal{R}_J^{\mathcal{A}}$. Then, a generalization of~\cref{rotationBoxSDP} immediately follows as:
\begin{equation}\label{rotationBoxSDPgral}
\begin{aligned}
\max_{\substack{ Q^{b_1},\ldots,Q^{b_{n-1}}, S }} \quad & f(\mathbf{c},\mathbf{s})\\
\textrm{s.t.} \quad & \bullet\; a_k^{b_i}=\sum_{0\leq j,j+k\leq 2J}Q_{j,j+k}^{b_i} \, \text{ for all } k\text{ and }i,\\
& \bullet\; \tilde{a}_k=-\sum_{0\leq j,j+k\leq 2J}S_{j,j+k} \, \text{ for all } k\neq 0,\\
& \bullet\; 1-\tilde{a}_0=\mathrm{Tr}(S),\\
  & \bullet\; Q^{b_1},\ldots,Q^{b_{n-1}},S\geq 0 ,
\end{aligned}
\end{equation}
where the entries of $Q^{b_1},\ldots,Q^{b_{n-1}}, S$ are labelled from $0$ to $2J$, and we have defined $\tilde{a}_k = \sum\limits_{i=1}^{n-1} a_k^{b_i}$. Note that the condition $\sum\limits_{i=1}^n P(b_i|\theta)=1$ removes one degree of freedom. Consequently, we take $i\in\{1,\ldots,n-1\}$, with $P(b_i|\theta)=\sum_{k=-J}^J a_k^{b_i}e^{\mathrm{i}k\theta}$. The generalization follows immediately from \cref{rotationBoxSDP}, which is the specific case for $n=2$. In particular, the conditions involving $Q^{b_i}$ imply $0\leq P(b_i|\theta)$ for all $i\in\{1,\ldots,{n-1}\},\theta\in\mathbb{R}$, and the constraints involving $S$ imply $\sum\limits_{i=1}^{n-1} P(b_i|\theta)\leq 1$ and, thus, $P(b_i|\theta) \leq 1$ for all $i\in\{1,\ldots,{n-1}\},\theta\in\mathbb{R}$. Finally, from $\sum\limits_{i=1}^n P(b_i|\theta)=1$ one can always find the missing $0\leq P(b_n|\theta) \leq 1$.

\subsection{Proof of Lemma~\ref{LemCompact}}\label{app:proof_LemCompact}
The arguments in the main text already demonstrate that every $\mathcal{Q}_J^\mathcal{A}$ is a compact convex set, and that every $P(a|\theta)$ is a trigonometric polynomial of degree at most $2J$, i.e.\ of the form
\[
   c_0+\sum_{k=1}^{2J} \left(\strut c_k \cos(k\theta)+s_k\sin(k\theta)\right).
\]
These are $4J+1$ parameters. If we have $|\mathcal{A}|$ functions of this kind that sum to one, then this tuple is determined by $(4J+1)(|\mathcal{A}|-1)$ parameters. All we need to show is that we can generate a set of correlations of this dimension via quantum rotation boxes. Denote the standard basis in $\mathbb{C}^{2J+1}$ by $\{|j\rangle\}_{j=-J}^J$, such that $U_\theta |j\rangle=e^{ij\theta}|j\rangle$. For $\ell=1,\ldots,2J$, define the pair of matrices $F^{(\ell)},G^{(\ell)}$ componentwise:
\[
   F^{(\ell)}_{kj}:=\delta_{j-k,\ell}+\delta_{k-j,\ell},\enspace
   G^{(\ell)}_{kj}:=i(\delta_{j-k,\ell}-\delta_{k-j,\ell}).    
\]
For example, if $J=3/2$ and $\ell=1$, then
\[
F^{(\ell)}=\left(\begin{array}{cccc} 0 & 1 & 0 & 0 \\ 1 & 0 & 1 & 0 \\ 0 & 1 & 0 & 1 \\ 0 & 0 & 1 & 0 \end{array}\right),
G^{(\ell)}=i\left(\begin{array}{cccc} 0 & 1 & 0 & 0 \\ -1 & 0 & 1 & 0 \\ 0 & -1 & 0 & 1 \\ 0 & 0 & -1 & 0 \end{array}\right).
\]
These band matrices are Hermitian. Consider the state $|+\rangle:=\frac 1 {\sqrt{2J+1}}\sum_{j=-J}^J |j\rangle$, then
\begin{eqnarray*}
\langle +|U_\theta^\dagger F^{(\ell)}U_\theta|+\rangle&=& f_{\ell,J} \cos(\ell\theta),\\
\langle +|U_\theta^\dagger G^{(\ell)}U_\theta|+\rangle&=& g_{\ell,J} \sin(\ell\theta),
\end{eqnarray*}
where $f_{\ell,J},g_{\ell,J}\neq 0$ are constants that only depend on $\ell$ and $J$. Now pick an arbitrary outcome $a_0\in\mathcal{A}$, and define a collection of Hermitian operators in the following way. If $a\neq a_0$, set
\[
E_a:=c_0^{(a)}\mathbf{1}+\sum_{\ell=1}^{2J}\left( c_\ell^{(a)} F^{(\ell)}+s_\ell^{(a)} G^{(\ell)}\right),
\]
and $E_{a_0}:=\mathbf{1}-\sum_{a\neq a_0} E_a$. If we choose the coefficients such that $0<c_\ell^{(a)},s_\ell^{(a)}\ll c_0^{(a)}$, then every $E_a$ for $a\neq a_0$ will be positive semidefinite, because the matrix $c_0^{(a)}$$ \mathbf{1}$ is contained in the interior of the set of positive semidefinite matrices. Furthermore, if we choose the $c_0^{(a)}$ small enough (but still non-zero), then $E_{a_0}$ will also be positive semidefinite such that we obtain a valid POVM. By construction,
\begin{eqnarray*}
P(a|\theta)&=& \langle +|U_\theta^\dagger E_a U_\theta |+\rangle\\
&=& c_0^{(a)}+\sum_{\ell=1}^{2J}\left( c_\ell^{(a)} f_{\ell,J} \cos(\ell\theta)+s_\ell^{(a)} g_{\ell,J}\sin(\ell\theta)\right),
\end{eqnarray*}
and varying the coefficients $c_\ell^{(a)},s_\ell^{(a)}$ while respecting the necessary inequalities to have a POVM produces a set of tuples of trigonometric polynomials of full dimension.
\qed

\subsection{Proof of Lemma~\ref{LemContained}}\label{app:LemContained}
Let $P\in\mathcal{Q}_J^{\mathcal{A}}$. Then $P(a|\theta)={\rm Tr}(U_\theta \rho U_\theta^\dagger E_a)$, with $\rho$ some quantum state and $\{E_a\}_{a\in\mathcal{A}}$ some POVM on the Hilbert space $\mathcal{H}=\mathbb{C}^{2J+1}={\rm span}\{|j\rangle\,\,|\,\, -J\leq j \leq J\}$ (every $j$ is an integer, or every $j$ is a half-integer), while $U_\theta |j\rangle=e^{ij\theta}|j\rangle$. Consider the Hilbert space $\mathcal{H}':=\mathbb{C}^{2 J'+1}={\rm span}\{|j'\rangle\,\,|\,\, -J'\leq j'\leq J'\}$, where $J':=J+\frac 1 2$. Define the isometry $W:\mathcal{H}\to\mathcal{H}'$ via $W|j\rangle:=|j+\frac 1 2\rangle$, then $W$ embeds $\mathcal{H}$ isometrically into $\mathcal{H}'$, and $W^\dagger W=\mathbf{1}_{\mathcal{H}}$ and $W W^\dagger=\mathbf{1}_{\mathcal{H}'}-|-J'\rangle\langle -J'|$. Set $\rho':=W\rho W^\dagger$, which is a quantum state on $\mathcal{H}'$, and $E'_a:=\frac 1 {|\mathcal{A}|} |-J'\rangle\langle -J'|+ W E_a W^\dagger$, then $\{E'_a\}_{a\in\mathcal{A}}$ is a POVM on $\mathcal{H}'$. Set $U'_\theta|j'\rangle:=e^{i j'\theta}|j'\rangle$, then
\[
e^{i\theta/2}W U_\theta W^\dagger =U_\theta'-e^{-i J'\theta} |-J'\rangle\langle -J'|,
\]
and hence
\[
{\rm Tr}(U'_\theta \rho' (U'_\theta)^\dagger E'_a)={\rm Tr}(U_\theta \rho U_\theta^\dagger E_a)=P(a|\theta),
\]
and so $P\in\mathcal{Q}_{J+1/2}^{\mathcal{A}}$.
\qed

\subsection{Proof of Lemma~\ref{lem:rep_decomp_Q}}\label{app:proof_lem_rep_decomp_Q}

We assume $U: \SO(2) \to \U(\mathcal{H})$, $U: \theta \mapsto U_\theta$, is a finite-dimensional unitary projective representation of ${\rm SO}(2)$ on the finite-dimensional complex Hilbert space $\mathcal{H}$, and show that this entails the following three propositions:

    \begin{itemize}
        \item $(i)$:
        It is proven in Lemma~\ref{lem:projSO2charac} that any finite-dimensional unitary projective representation  $\SO(2) \to \U(V)$ is of the form given in~\Cref{eqRepDef} (see also Section 1 of the Supplemental Materials of~\cite{Jones}), where $J$ is uniquely defined by the condition $n_J n_{-J}\neq 0$. 
        \item $(i) \Leftrightarrow (ii)$: Using the isomorphism of Lemma~\ref{lem:herm_sym_equiv}: $\LH(V) \to \Sym(V \otimes \bar V)$ , $\rho \mapsto \vec \rho$ and the dual isomorphism $\LH(V)^* \to \Sym(V \otimes \bar V)^*$, $E \mapsto \vec E^\top$ we have:
        \begin{align}
            \Tr(U_\theta \rho U_\theta^\dagger E) = \vec E^\top ( U \otimes \bar U) \vec \rho.
        \end{align}
        By \Cref{lem:real_rep_sym}, $\Sym(V \otimes \bar V)$ is closed (as a real vector space) under the action of $U \otimes \bar U$. Denoting $P$ the projector $P: V \otimes \bar V \to \Sym(V \otimes \bar V)$ we have:
        \begin{align}
             \vec E^\top ( U \otimes \bar U) \vec \rho =  \vec E^\top P ( U \otimes \bar U) P \vec \rho,
        \end{align}
        where $P (U \otimes \bar U) P \in \mL_\reals(\Sym(V \otimes \bar V))$ the space of real linear operators on $ \Sym(V \otimes \bar V)$. Since $\vec E^\top \in \Sym(V \otimes \bar V)^*$ and  $\vec \rho \in \Sym(V \otimes \bar V)$, the map $U\otimes \bar U 
        \to \vec E^\top ( U \otimes \bar U) \vec \rho$ can be linearly extended to a functional on $\mL_\reals(\Sym(V \otimes \bar V))$ and hence it is a linear combination of the entries in $U \otimes \bar U$. 
        As can be seen easily, and is done explicitly below in (iii), these entries are trigonometric polynomials of order at most $2J$, which entails $ \Tr(U_\theta \rho U_\theta^\dagger E)$ is a trigonometric polynomial of order at most $2J$. And since these maps span $\mL_\reals(\Sym(V \otimes \bar V))^*$ linearly, the degrees of the trigonometric polynomials $ \Tr(U_\theta \rho U_\theta^\dagger E)$ cannot all be strictly smaller than $2J$.

        \item $(i) \Leftrightarrow (iii)$:  Denote the Hilbert space space on which the projective representation acts by $\mathcal{H}_J$.
        The representation induced on the complex vector space of matrices $\mL(\mH_J)$ is given by $\theta \mapsto U_\theta \bullet U_\theta^\dagger$. Using the isomorphism  $\mL(\mH_J) \simeq \mathcal H_J \otimes \mathcal H_J^*  \simeq \mathcal H_J \otimes \bar{\mathcal {H}}_J$ with corresponding representations  $U_\theta \bullet U_\theta^\dagger \simeq U(\theta) \otimes U^*(\theta) \simeq U(\theta) \otimes \bar U(\theta)$, we obtain the following decomposition of $\mL(\mH_J)$ into irreducible representations:
\begin{align}
    U(\theta) \otimes \bar U(\theta) &= \bigoplus_{j=-J}^J \mathbf{1}_{n_j} e^{ij\theta} \otimes \bigoplus_{k=-J}^J \mathbf{1}_{n_k} e^{-ik\theta} \\
    &= \bigoplus_{l = -2J}^{2J} \mathbf{1}_{m_l} e^{il \theta}.
    \end{align}
The multiplicity $m_l$ for a given irreducible representation $e^{il \theta}$ is given by
\begin{align}
    m_l = \sum_{j,k| j - k = l} n_j \times n_k .
    \label{eqMultiplicity}
\end{align}
In particular, note that $m_{2J}=n_J n_{-J}>0$. This will imply that the representation on $\LH(\mH_J)$ is generated not by an arbitrary projective unitary representation of $\SO(2)$ but one specifically of the form $(i)$, with the specific value of $J$.

The multiplicity $m_{-l}$ is equal to $m_l$:
\begin{align}
    m_{-l} &= \sum_{j,k| j - k = -l} n_j \times n_k =\sum_{j,k| k - j = l} n_j \times n_k \\
    &= \sum_{j,k| k - j = l} n_k \times n_j = m_l .
\end{align}
From the equality
    \begin{align}
        \begin{pmatrix}
            e^{i k \theta} & 0 \\
            0 & e^{-ik \theta} 
        \end{pmatrix}
         = L  \begin{pmatrix}
            \cos(k \theta) & - \sin( k \theta) \\
            \sin(k \theta) & \cos( k \theta)
        \end{pmatrix}  L^{-1} ,
    \end{align}
    where
    \begin{align}
        L = \begin{pmatrix}
            - i & i \\
            1   & 1
        \end{pmatrix}, \ 
        L^{-1} =\frac{1}{2} \begin{pmatrix}
            i & 1\\
            -i   & 1
        \end{pmatrix} ,
    \end{align}
and from $m_l = m_{-l}$, it follows that
\begin{align}
     \bigoplus_{l = -2J}^{2J} \mathbf{1}_{m_l} e^{il \theta} \simeq \mathbf{1}_{m_0}\oplus
        \bigoplus_{k= 1}^{2J} \mathbf{1}_{m_k} \otimes \begin{pmatrix}
            \cos(k \theta) & - \sin( k \theta) \\
            \sin(k \theta) & \cos( k \theta)
        \end{pmatrix},
\end{align}
which is a decomposition of $\mathcal H_J \otimes \bar{ \mathcal H}_J$ into real irreducible subspaces. By Lemma~\ref{lem:real_rep_sym}, the real subspace $\Sym(\mH_J \otimes \bar \mH_J)$ carries the real representation of the above form. Thus, so does $\LH(\mH_J)$ due to Lemma~\ref{lem:herm_sym_equiv}.
\end{itemize}

We have thus shown that $(i)\Leftrightarrow(ii)$ and $(i)\Leftrightarrow (iii)$, hence all three statements are equivalent.
Finally, we consider the specific case where  $U_\theta^J:=e^{i\theta Z^J}$, with $Z^J={\rm diag}(J,J-1,\ldots,-J)$ \ The representation $\Gamma_\theta$ acts on the linear space spanned by density operators $\LH(\comp^{2J + 1})$ as $\Gamma_\theta \vec \rho= U_\theta \rho U^\dagger_\theta$. Using again the isomorphism $U_\theta \cdot U_\theta^\dagger \simeq U_\theta \otimes \bar U_\theta$, \cref{eqMultiplicity} in the special case $n_j=1$ entails $m_l=2J+1-l$.
\qed

\subsection{Proof of Lemma~\ref{lemRotationBox}}\label{app:lemRotationBox}
Suppose $p\in\mathcal{R}_J$, then the Fej\'er-Riesz theorem implies that there is $q(\theta)=\sum_{j=-J}^J b_j e^{ij\theta}$ such that $p(\theta)=|q(\theta)|^2$. Thus
\begin{eqnarray*}
\overline{q(\theta)}q(\theta)=\sum_{j,k=-J}^J e^{i(k-j)\theta} \overline{b_j} b_k = p(\theta),
\end{eqnarray*}
and hence $a_k=\sum_{0\leq j,j+k\leq 2J} \overline{b_j} b_{j+k}$. Define $Q_{jk}:=\overline{b_j} b_k$, then $Q\geq 0$ and the first condition in Lemma~\ref{lemRotationBox} follows. Similarly, $1-p(\theta)\geq 0$ implies the second and third condition.

Conversely, suppose that the first condition of Lemma~\ref{lemRotationBox} is satisfied. Then
\begin{eqnarray*}
p(\theta)&=&\sum_{k=-2J}^{2J} a_k e^{ik\theta}=\sum_{k=-2J}^{2J}\sum_{0\leq j,j+k\leq 2J} Q_{j,j+k} e^{ik\theta}\\
&=&\sum_{k=-2J}^{2J} \sum_{\ell=-2J}^{2J} Q_{j\ell} e^{i(\ell-j)\theta}=\langle v|Q|v\rangle\geq 0,
\end{eqnarray*}
where $v_k=e^{ik\theta}$, and where we have used the substitution $\ell:=j+k$. Similarly, the second and the third condition imply $1-p(\theta)\geq 0$ for all $\theta$.
\qed

\section{Proofs for Section~\ref{sec:R1}}
For clarity, we restate some of the lemmas or theorems before their proofs.

\subsection{Proof of Lemma~\ref{lem:non-cst-ext-R1}}\label{app:non-cst-ext-R1}

\textbf{Lemma~\ref{lem:non-cst-ext-R1}.}
    \textit{Every non-constant function $p\in\partial_{\rm ext}\mathcal{R}_1$ is contained in at least one face $F_{\theta_0,\theta_1}$.}

\begin{proof}
It is sufficient to show that all $p\in\partial_{\rm ext}\mathcal{R}_1$ satisfy $\min_\theta p(\theta)=0$ and $\max_\theta p(\theta)=1$. To show that the maximum is unity, let $m:=\max_\theta p(\theta)$. Since $p$ is not identically zero by assumption, we have $m>0$. Suppose that $m<1$. Then $q(\theta):=p(\theta)/m$ is itself an element of $\mathcal{R}_1$, and $p(\theta)=m\cdot q(\theta)+(1-m)\cdot 0$. Thus, $p$ is not extremal in $\mathcal{R}_1$, which contradicts our assumption that it is. The proof that the minimum is zero is analogous.
\end{proof}

\subsection{Proof of Lemma~\ref{lem:charac_R1_faces}}\label{app:charac_R1_faces}

Lemma~\ref{lem:charac_R1_faces} gives an explicit characterization of the sets $F_{0, \theta_1}$ showing for which values of $\theta_1$ the set is empty, and for values where $F_{0, \theta_1}$ is non-empty, and hence a face of $F_0$, it characterizes the functions in $\delta_{\rm ext} F_{0, \theta_1}$.

We first characterize the general form of the functions $p \in F_0$, which are of interest since $F_{0, \theta_1} \subset F_0$ for all $\theta_1$.  

\begin{lemma}\label{LemFormPoly}
Let $p(\theta)$ be a trigonometric polynomial of degree $2$ or less with $p(\theta)\geq 0$ for all $\theta$ and $p(0)=0$. Then there are constants $c\geq 0$, $\varphi\in[0,2\pi)$ and $0\leq s \leq 1$ such that
\[
   p(\theta)=c(1-\cos\theta)(1-s \cos(\theta-\varphi)).
\]
\end{lemma}
\begin{proof}
 Due to the Fej\'er-Riesz theorem, there is a complex polynomial
 \[
    h(z)=a_0+a_1 z+a_2 z^2
 \]
 with $p(\theta)=|h(e^{i\theta})|^2$; we can choose $a_2$ to be a real number by absorbing complex phases into the definition of $h$. We have $0=p(0)=h(1)$, and thus $a_0=-a_1-a_2$, hence $h(z)=(z-1)(a_2 z +a_1 +a_2)$. Write $-(a_1+a_2)/a_2=r e^{i\varphi}$ with $r\geq 0$ and $\varphi\in\mathbb{R}$, then
 \begin{eqnarray*}
    p(\theta)&=& |h(e^{i\theta})|^2\\
&=& |e^{i\theta}-1|^2\cdot \left| a_2 e^{i\theta}+a_1+a_2\right|^2\\
&=& 2 a_2^2 (1-\cos\theta)\left| e^{i\theta}-r e^{i\varphi}\right|^2\\
&=& 2 a_2^2(1-\cos\theta)(1+r^2-2r\cos(\theta-\varphi))\\
&=& c (1-\cos\theta)\left(1-\frac{2r}{1+r^2}\cos(\theta-\varphi)\right),
 \end{eqnarray*}
where $c=2 a_2^2(1+r^2)$, and $s:=2r/(1+r^2)\in [0,1]$.
\end{proof}

From the previous lemma we can immediately determine the maximal number of roots for functions $p \in \mR_1$:
\begin{lemma}\label{lem:zero_and_one_func}
    Every function $p \in \mR_1$ reaches value $p(\theta) = 0$ at most twice and value $p(\theta) = 1$ at most twice
 \end{lemma}

 \begin{proof}
     Consider a function $p' \in \mR_1$ such that $p'(\theta_0) = 0$. The function $p(\theta) = p'(\theta + \theta_0)$ is such that $p(0) = 0$ and has the same number of roots as $p'(\theta)$. Thus we can restrict ourselves to the case of function $p \in \mR_1$ such that $p(0) = 0$.

     By Lemma~\ref{LemFormPoly} these functions have the form:
     \begin{align}
          p(\theta)=c(1-\cos\theta)(1-s \cos(\theta-\varphi))
     \end{align}
     which attains value $0$ at $\theta = 0$ and at $\theta = \varphi$ if the parameter $s = 1$. Thus $p(\theta)$ has at most two roots.

     Conversely consider a function $p' \in \mR_1$ which reaches value $p'(\theta_1^i) = 1$ for $n$ points $\{\theta_{1}^1,..., \theta^n\}$. The function $p(\theta) = 1 - p'(\theta)$ has $n$ roots $p(\theta_1^i) = 0$ for  $\theta_1^i \in \{\theta_{1}^1,..., \theta^n\}$. However, since $p \in \mR_1$, $n$ is at most two. Thus, $p'$ has at most two points $\theta_1^i$ such that $p'(\theta_1^i) = 1$.
 \end{proof}

Note that compact convex faces have a well-defined dimensionality. We now show that for all faces (i.e.  non-empty $F_{0, \theta_1}$)  the dimensionality is either 0 (i.e the face contains a single point) or 1 (the face is the convex hull of two distinct points).

\begin{lemma}
\label{LemDim1}
Let $\theta_1\neq\pi$. Then either $F_{0,\theta_1}=\emptyset$ or $\dim( F_{0,\theta_1})\leq 1$.
\end{lemma}
\begin{proof}
Let $p\in F_{0,\theta_1}$, then Lemma~\ref{LemFormPoly} shows that
\[
p(\theta)=c(1-\cos\theta)(1-s\cos(\theta-\varphi)),
\]
where $c>0$ is uniquely determined by the equation $p(\theta_1)=1$. Furthermore, $\theta_1$ is a local maximum, hence
\begin{eqnarray*}
0&=&p'(\theta_1)\\
&=& 2c\left(\cos\frac{\theta_1} 2 -s \cos\left(\frac 3 2 \theta_1-\varphi\right)\right)\sin\frac{\theta_1} 2.
\end{eqnarray*}
Since $0<\theta_1<2\pi$, we know that $\sin\frac{\theta_1} 2 \neq 0$, hence
\[
\cos\frac{\theta_1} 2 - s \cos\left(\frac 3 2 \theta_1 - \varphi\right)=0.
\]
Suppose that $\cos\left(\frac 3 2 \theta_1-\varphi\right)=0$, then $\cos\frac{\theta_1} 2 =0$, which implies $\theta_1=\pi$, which contradicts the assumptions of the lemma. Hence $\cos\left(\frac 3 2 \theta_1-\varphi\right)\neq 0$, and
\[
s=\frac{\cos\frac{\theta_1} 2 } {\cos\left(\frac 3 2 \theta_1-\varphi\right)}.
\]
But this implies that every $p\in F_{0,\theta_1}$ is uniquely determined by the parameter $\varphi$. (Note that not all $\varphi\in[0,2\pi)$ yield valid $p\in F_{0,\theta_1}$, i.e.\ only a subset of $[0,2\pi)$ is allowed as possible values for $\varphi$, but this observation does not affect the present argumentation.) Hence $\dim (F_{0,\theta_1})\leq 1$.
\end{proof}
\begin{lemma}
\label{LemPiSpecial}
We have $\dim(F_{0,\pi})\leq 1$.
\end{lemma}
\begin{proof}
Let $p\in F_{0,\pi}$, then Lemma~\ref{LemFormPoly} implies
\begin{equation}
p(\theta)=c(1-\cos\theta)(1-s\cos(\theta-\varphi)),
\label{eqGenForm2}
\end{equation}
where $0\leq\varphi<2\pi$ and $0\leq s \leq 1$. Furthermore,
\[
1=p(\pi)=2 c (1+s \cos\varphi),
\]
hence $s\cos\varphi>-1$ and
\[
c=\frac 1 {2(1+s\cos\varphi)}.
\]
Substituting this into Eq.~(\ref{eqGenForm2}), and using that $\pi$ is a local maximum, the equation $0=p'(\pi)$ implies
\[
s\sin\varphi=0.
\]
Hence, either $s=0$ such that $p(\theta)=\frac 1 2(1-\cos\theta)$, or $\varphi=\pi$ such that
\begin{equation}
   p(\theta)=\frac{1-\cos\theta}{2(1-s)}(1+s\cos\theta),
   \label{FacePi}
\end{equation}
or $\varphi=0$ such that
\[
   p(\theta)=\frac{1-\cos\theta}{2(1+s)}(1-s\cos\theta).
\]
Equation~(\ref{FacePi}) contains the other two cases via $s=0$ and $s\geq -1$, and we conclude that the single parameter $-1\leq s< 1$ determines the element of $F_{0,\pi}$ uniquely (note that we do not claim that all these values of $s$ give valid functions in the face, just that they are all contained in this family of functions).
\end{proof}
A compact convex set of dimension $1$ has exactly $2$ extremal points. Thus
\begin{corollary}
Every face $F_{0,\theta_1}$ contains either one or two extremal points, depending on whether its dimension is $0$ or $1$ (in the former case, it contains only a single element).
\end{corollary}

The faces $F_{0,\theta_1}$ contain those functions $p \in \mR_1$ such that a global minimum is $p(\theta_0)= 0$ and a global maximum $p(\theta_1) = 1$. However, some functions in  $F_{0,\theta_1}$  can have multiple global maxima and minima, as we shall now see.

\begin{lemma}\label{lem:uniquenessR1}
Let $\theta_0\neq \theta'_0\in[0,2\pi)$ be two distinct angles. Then there is a unique $p\in\mathcal{R}_1$ with $p(\theta_0)=p(\theta'_0)=0$ and $\max_\theta p(\theta)=1$, and it is of the form
\[
   p(\theta)=c(1-\cos(\theta-\theta_0))(1-\cos(\theta-\theta'_0)),
\]
with some suitable uniquely determined $c>0$.

Similarly, if $\theta_1\neq\theta'_1\in[0,2\pi)$ are distinct angles, then there is a unique $p\in\mathcal{R}_1$ with $p(\theta_1)=p(\theta'_1)=1$ and $\min_\theta p(\theta)=0$, and it is of the form
\[
   p(\theta)=1-c(1-\cos(\theta-\theta_0))(1-\cos(\theta-\theta'_0)),
\]
with some suitable uniquely determined $c>0$.
\end{lemma}
\begin{proof}
The latter statement follows from the former by considering $q(\theta):=1-p(\theta)$. It is thus sufficient to prove the former statement. For symmetry reasons, it is enough to consider the case $\theta_0=0$. Due to Lemma~\ref{LemFormPoly},
\[
p(\theta)=c(1-\cos\theta)(1-s\cos(\theta-\varphi)),
\]
where $c\geq 0$, $\varphi\in[0,2\pi)$ and $0\leq s \leq 1$. Since $\theta'_0>0$, we have $1-\cos\theta'_0\neq 0$, and so $p(\theta'_0)=0$ implies
\[
1-s\cos(\theta'_0-\varphi)=0.
\]
This is only possible if $s=1$ and $\cos(\theta'_0-\varphi)=1$, hence $\varphi=\theta'_0$, and so
\begin{equation}
p(\theta)=c(1-\cos\theta)(1-\cos(\theta-\theta'_0)),
\label{eqSpecialForm}
\end{equation}
and $c>0$ is uniquely determined by the condition $\max_\theta p(\theta)=1$.
\end{proof}
\begin{corollary}
\label{CorExtremal}
Every $p\in\mathcal{R}_1$ that either
\begin{itemize}
    \item attains the value $0$ once and the value $1$ twice, or
    \item attains the value $1$ once and the value $0$ twice
\end{itemize}
is extremal in $\mathcal{R}_1$.
\end{corollary}
Actually, we can easily transform one of these into the other:
\begin{lemma}
Let $p\in\mathcal{R}_1$ as a $(2\pi)$-periodic function on $\mathbb{R}$, and suppose that
\[
   p(\theta_0)=0, \enspace p(\theta_1)=1,\enspace p(\theta'_0)=0.
\]
Then the $(2\pi)$-periodic function
\begin{equation}
   \tilde p(\theta):=1-p(\theta_0+\theta_1-\theta)
   \label{eqTrafoP}
\end{equation}
is also an element of $\mathcal{R}_1$, and it satisfies
\[
   \tilde p(\theta_0)=0,\enspace \tilde p(\theta_1)=1,\enspace \tilde p(\theta'_1)=1,
\]
where $\theta'_1:=\theta_0+\theta_1-\theta'_0$.
\end{lemma}
The proof is very simple and omitted. In general, we can consider the transformation
\[
T_{\theta_0,\theta_1}:p\mapsto \tilde p,
\]
where $\tilde p$ is defined by Eq.~(\ref{eqTrafoP}), which maps $\mathcal{R}_1$ onto itself and is linear. Moreover, the lemma above also shows that
\[
T_{\theta_0,\theta_1}(F_{\theta_0,\theta_1})=F_{\theta_0,\theta_1},
\]
i.e.\ it preserves the faces that we are interested in. The idea is that it maps one of the extremal point (with two zeros) to the other extremal point (with two ones).

Let us study whether functions can have more than two global maxima or minima.

\begin{lemma}\label{lem:min_max_ang_bound}
    Given a function $p \in \mR_1$ with $p(\theta_0)  =0$ and $p(\theta_1) = 1$ we have $|\theta_0 - \theta_1| \geq \frac{\pi}{2}$.
\end{lemma}
\begin{proof}
This is the special case $J=1$ of \Cref{LemDistZeroOne}.
\end{proof}
From this it follows
\begin{corollary}\label{cor:empty_faces}
        If $0  \leq \theta_1 <\frac{\pi}{2}$ or if $\frac{3 \pi}{2} < \theta_1 < 2 \pi $ then $F_{0 ,\theta_1} =  \emptyset$.
\end{corollary}
\begin{proof}
    The set $F_{0 ,\theta_1}$ contains those functions in $\mR_1$ such that $p(\theta_0) = 0$ and $p(\theta_1) = 1$, where $\theta_0 = 0$.  For $0 \leq \theta_1 < \frac{\pi}{2}$, we have $|\theta_0 - \theta_1| <\frac{\pi}{2}$. Thus, by Lemma~\ref{lem:min_max_ang_bound}, $F_{0 ,\theta_1}$ is empty.
    
    Similarly, since $p(2 \pi)  = 0$, it also follows that for $\frac{3 \pi}{2} \leq \theta_1 < 2 \pi $ that the face  $F_{0 ,\theta_1}$ is empty.
\end{proof}

By Lemma~\ref{lem:zero_and_one_func} a function $p \in \delta_{\rm ext} \mR_1$ has at most two global minima and at most two global maxima.

\begin{corollary}\label{cor:two_zero_two_one}
    A non-constant function $p \in \delta_{\rm ext} \mR_1$ with two global minima $\theta_0$ and $\theta_0'$ and two global maxima $\theta_1$ and $\theta_1'$ is such that $\theta_0' = \theta_0 + \pi$, $\theta_1 = \theta_0 + \frac{\pi}{2}$ and $\theta_1' = \theta_1 + \pi$.
\end{corollary}
\begin{proof}
     A function $p \in \delta_{\rm ext} \mR_1$ has global minimum $p(\theta_0) = 0$ by Lemma~\ref{lem:non-cst-ext-R1}. Thus, if it has two global minima, there is another $\theta_0' \neq \theta_0$ such that $p(\theta_0') = 0$.

     Similarly the global maxima of the function are reached for $p(\theta_1) = p(\theta_1') = 1$.

     By Lemma~\ref{lem:min_max_ang_bound} we have the following relations:
     \begin{align}
         |\theta_0 - \theta_1| \geq \frac{\pi}{2} , \quad |\theta_0' - \theta_1| \geq \frac{\pi}{2} ,\\
           |\theta_0 - \theta_1'| \geq \frac{\pi}{2} , \quad |\theta_0' - \theta_1'| \geq \frac{\pi}{2} .
     \end{align}
     Without loss of generality we assume $\theta_0 < \theta_0'$ and $\theta_1 < \theta_1'$. 
     This implies 
     \begin{align}
         |\theta_0 - \theta_0'| \geq \pi , \quad
          |\theta_1 - \theta_1'| \geq \pi.
     \end{align}
     Thus, $\theta_0$ and $\theta'_0$ must lie on antipodal points of the unit circle, and so do $\theta_1$ and $\theta'_1$. Moreover, since $\theta_0$ and $\theta_1$ must have distance at least $\pi/2$, they must have distance exactly $\pi/2$, and the four extrema form the corners of a square inside the circle. This proves the claimed equations.
\end{proof}
We now show that such a function exists and is unique.

\begin{lemma}
The only $p\in\mathcal{R}_1$ that have two distinct zeros and two distinct ones are
\[
   p(\theta)=(1-\cos(\theta-\theta_0))(1+\cos(\theta-\theta_0)),
\]
with $0\leq\theta_0< \pi$.
\end{lemma}

\begin{proof}
 Since $p(\theta_0) = p(\theta_0') = 0$  and $\max_\theta p(\theta) = 1$ Lemma~\ref{LemFormPoly} implies that $p(\theta)$ has the form:
    \begin{align}
        p(\theta)=c\cdot(1-\cos(\theta - \theta_0))(1- \cos(\theta-\theta_0')).
    \end{align}
    By Corollary~\ref{cor:two_zero_two_one}, $\theta_0 ' = \theta_0 + \pi$, hence 
        \begin{align}
        p(\theta)=c\cdot(1-\cos(\theta - \theta_0))(1+ \cos(\theta-\theta_0)),
    \end{align}
and $\max_\theta p(\theta)=1$ implies that $c=1$.
\end{proof}

\begin{lemma}\label{lem:max_val}
   The  unique global maximum of the function $f_{\theta_0'}:[0,2\pi)\to\mathbb{R}$,
\[
f_{\theta_0'}(\theta):=(1-\cos\theta)(1-\cos(\theta-\theta_0')),
\]
occurs at $\theta_1 = \frac{\theta_0'}{2} + \pi$ when  $\theta_0'\in (0,\pi)$ and at $\theta_1 = \frac{\theta_0'}{2}$ when $\theta_0'\in (\pi ,2\pi)$.
\end{lemma}

\begin{proof}

Let us find local extrema:
\begin{align}
    f_{\theta_0'}'(\theta) &= \sin (\theta ) (1-\cos (\theta_0' -\theta))-(1-\cos (\theta )) \sin (\theta_0' -\theta) \nonumber \\
    &=\sin (\theta_0'- 2\theta) - \sin (\theta_0'- \theta)+\sin (\theta ).
\end{align}
The equation $f_{\theta_0'}'(\theta) = 0$ has the following solutions in $[0, 2\pi)$:
\begin{align}
    \theta \in \left\{0,\theta_0',\frac{\theta_0'}{2} ,  \pi+ \frac{\theta_0'}{2}\right\}  \quad \mbox{if } \theta_0' \in [0,\pi ) ,\\
    \theta \in \left\{0 , \theta_0' ,- \pi+  \frac{\theta_0'}{2} , \frac{\theta_0'}{2} \right\}  \quad\mbox{if }  \theta_0' \in [\pi,2 \pi)  .
\end{align}
One can check directly that these are zeroes of $f_{\theta_0'}'$. Moreover, since $f_{\theta_0'}'$ is a trigonometric polynomial of degree $2$, it has at most $4$ zeroes (up to $(2\pi)$-periodicity), hence these are the only zeroes. Clearly,
$f_{\theta_0'}(\theta)$ attains a global minimum for $\theta = 0$ and $\theta = \theta_0'$. Let us  determine the global maximum:
\begin{align}
    f_{\theta_0'}\left(\frac{\theta_0'}{2}\right)&= \left(1 - \cos\frac{\theta_0'}{2}\right)^2 ,\\
    f_{\theta_0'}\left(  \pi + \frac{\theta_0'}{2}\right)&=  f_{\theta_0'}\left( - \pi + \frac{\theta_0'}{2}\right) = \left(1 + \cos\frac{\theta_0'}{2}\right)^2 .
\end{align}
We see that $f_{\theta_0'}(\frac{\theta_0'}{2}) < f_{\theta_0'}( \frac{\theta_0'}{2} \pm \pi)$ if and only if $\cos(\frac{\theta_0'}{2})> 0$. This implies that the unique global maximum occurs at $\theta_1 = \pi + \frac{\theta_0'}{2}$ when $\theta_0 \in (0, \pi)$, at $\theta_1 = \frac{\theta_0'}{2} $ for $\theta_0' \in (\pi, 2\pi)$.
\end{proof}

\begin{lemma}\label{lem:single_point_faces}
If $\theta_1=\frac \pi 2$ or $\theta_1=\frac{3\pi}2$, then $F_{0,\theta_1}$ contains a single element, namely
\[
F_{0,\frac\pi 2}=F_{0,\frac{3\pi}2}=\left\{p(\theta)=\sin^2\theta\right\}.
\]
\end{lemma}
\begin{proof}
Let $0\leq\theta_1\leq\frac \pi 2$, and suppose that $p\in F_{0,\theta_1}$. Consider $T(\theta):=2 p(\theta)-1$. We have $-1\leq T(\theta)\leq 1$ for all $\theta$, thus, we can use the result of~\cite[Theorem 2]{Jones}
\begin{equation}
T'(\theta)+n^2 T(\theta)^2\leq n^2,
\label{eqDiff}
\end{equation}
where $n$ is the degree of the trigonometric polynomial (here $n=2$). Thus,
\begin{eqnarray*}
\theta_1&=& \int_0^{\theta_1} d\theta \geq \int_0^{\theta_1} \frac{T'(\theta)d\theta}{2\sqrt{1-T(\theta)^2}}=\frac 1 2 \int_{T(0)}^{T(\theta_1)} \frac{dy}{\sqrt{1-y^2}}\\
&=& \frac 1 2 \left(\arcsin T(\theta_1)-\arcsin T(0)\right)=\frac\pi 2.
\end{eqnarray*}
This is a contradiction if $\theta_1<\frac\pi 2$, and so $F_{0,\theta_1}=\emptyset$ in this case. On the other hand, to have equality in the case $\theta_1=\frac\pi 2$, we must have equality in Eq.~(\ref{eqDiff}) for all $0\leq \theta\leq \frac \pi 2$, which implies that $T(\theta)=-\cos(2\theta)$. A similar calculation for $\frac{3\pi}2\leq\theta_1<2\pi$ proves the claim.
\end{proof}

\begin{lemma}\label{lem:inter_faces}
Let $\theta_1\in\left(\frac \pi 2,\frac{3\pi}2 \right)\setminus\{\pi\}$. Then $F_{0,\theta_1}$ contains exactly two distinct extremal points,
\[
\partial_{\rm ext}F_{0,\theta_1}=\{p(\theta),\tilde p(\theta)\},
\]
namely
\[
   p(\theta)=c(1-\cos\theta)(1-\cos(\theta-\theta_0)),
\]
and $\tilde p$ is defined as in Eq.~(\ref{eqTrafoP}). Here $\theta_0 = 2 \theta_1 $ for $\theta_1 \in (\frac{\pi}{2}, \pi)$ and $\theta_0 = 2 ( \theta_1 - \pi) $ for $\theta_1 \in (\pi, \frac{3\pi}{2})$, and $c>0$ is uniquely determined by the condition $\max_\theta p(\theta)=1$. 
\end{lemma}
\begin{proof}
Fix some $\theta_1\in(\frac{\pi}{2}, \pi)$. Then, by Lemma~\ref{lem:max_val}, the function $f_{\theta_0}$ for $\theta_0 = 2 \theta_1$ is such that $f_{\theta_0}(\theta_1)$ is its global maximum. For  $\theta_1\in( \pi,\frac{3 \pi}{2})$, the function $f_{\theta_0}$ with $\theta_0 = 2 (\theta_1 - \pi)$ is such that $f_{\theta_0}(\theta_1)$ is its global maximum. 

 Set $c:=1/f_{\theta_0}(\theta_1)$ and $p(\theta):=c\, f_{\theta_0}(\theta)$, then $p(0)=0$, $p(\theta_1)=1=\max_\theta p(\theta)$, and $p(\theta)\geq 0$ for all $\theta$, hence $p\in F_{0,\theta_1}$. By Lemma~\ref{lem:max_val}, $p(\theta)$ reaches value $0$ twice at $\theta = 0, \theta_0$, and value $1$ at $\theta_1$. Hence, due to Corollary~\ref{CorExtremal}, $p$ is extremal in $\mathcal{R}_1$ and thus also extremal in $F_{0,\theta_1}$. Since $p$ does not attain the value $1$ twice, we have $\tilde p\neq p$. Moreover, for the same reason as for $p$, we have $\tilde p\in\partial_{\rm ext}F_{0,\theta_1}$.

We have discovered two distinct extremal points of $F_{0,\theta_1}$. Since $\dim F_{0,\theta_1}\leq 1$ according to Lemma~\ref{LemDim1}, there cannot be any more extremal points.
\end{proof}

The following uses the terminology of Lemma~\ref{LemPiSpecial}.

\begin{lemma}\label{lem:pi_face}
The face $F_{0,\pi}$ contains exactly two extremal points, namely
\[
   F_{0,\pi}=\{p(\theta),\tilde p(\theta)\},
\]
where $p(\theta)=\sin^4\frac\theta 2$, and $\tilde p$ is defined as in~(\ref{eqTrafoP}) (concretely, $\tilde p(\theta)=\frac 1 4(1-\cos\theta)(3+\cos\theta)$).
\end{lemma}
\begin{proof}
Every $p\in F_{0,\pi}$ corresponds to some element of the family of functions $p_s$ defined in Eq.~(\ref{FacePi}), with $-1\leq s\leq 1$. Indeed, the case $s=-1$ yields a valid function $p\in F_{0,\pi}$, and since it is in the topological boundary of the parameter range, it must correspond to an extremal point of the one-dimensional face. But the reversible transformation $T_{0,\pi}$ maps extremal points to extremal points, and hence $\tilde p:=T_{0,\pi} p$ must also be an extremal point of $F_{0,\pi}$ (in fact, it is the function corresponding to $s=\frac 1 3$). Since $\dim F_{0,\pi}\leq 1$ according to Lemma~\ref{LemPiSpecial}, these must be the only extremal points. (Note that this also shows that the face corresponds to the parameter range $-1\leq s \leq \frac 1 3$).
\end{proof}

The four statements of Lemma~\ref{lem:charac_R1_faces} are now proven in Corollary~\ref{cor:empty_faces}, Lemma~\ref{lem:single_point_faces}, Lemma~\ref{lem:inter_faces}, and Lemma~\ref{lem:pi_face}, respectively.

\subsection{Proof of Theorem~\ref{thm:R1eqQ1}}\label{app:proofR1eqQ1}

\begin{theorem*}[$\mQ_1 = \mR_1$]
    The  correlation  set $\mR_1$ is equal to $\mQ_1$.
\end{theorem*}

By Lemma~\ref{lem:qspin_sub_gspin}, we have $\mQ_1 \subset \mR_1$. To show the converse, we will use \Cref{LemExtremalImpliesR1} and show that all correlations in $\delta_{\rm ext} \mR_1$ have a quantum spin-$1$ realization. 

\begin{lemma}
    If $p(\theta) \in \mQ_1$ then $p'(\theta) := p(\theta + \theta_0) \in \mQ_1$.
\end{lemma}

\begin{proof}
The assumption $p(\theta) \in \mQ_1$ implies that there is a quantum state $\rho$ and a POVM element $E$ such that
    \begin{align}
        p(\theta) = \Tr(E U_\theta \rho U_\theta^\dagger), 
    \end{align}
    hence 
    \begin{align}
        p'(\theta) &= p(\theta+ \theta') = \Tr(E U_{\theta + \theta'} \rho U_{\theta + \theta'}^\dagger) \\
        &=  \Tr(E U_{\theta} (U_{\theta'} \rho U_{\theta'}^\dagger) U_{\theta}^\dagger) = \Tr(E U_{\theta}  \rho'  U_{\theta}^\dagger),
    \end{align}
    with $\rho' = (U_{\theta'} \rho U_{\theta'}^\dagger)$ a valid quantum state, hence $p'(\theta) \in \mQ_1$.
\end{proof}

Thus, we only need to show that the extremal points  $p \in \delta_{\rm ext} \mR_1$ with $p(0) = 0$ are quantum realizable. 

\begin{lemma}
    If $p(\theta) \in \mQ_1$ with $p(\theta_0) = p(\theta_0') = 0$ and $p(\theta_1) = 1$, then $\tilde p(\theta):=1-p(\theta_0+\theta_1-\theta) \in \mQ_1$.
\end{lemma}

\begin{proof}
$p(\theta) \in \mQ_1$ entails there exists a qutrit state $\rho$ and a qutrit effect $E$ such that
\begin{align}
    p(\theta) = \Tr(E U_\theta \rho U_\theta^\dagger) , 
\end{align}
where $U_\theta = \diag(e^{i \theta} , 1 , e^{- i\theta})$.

Define the effect $E' = \id  -  U_{\theta_0 + \theta_1}^\dagger  E U_{\theta_0 + \theta_1}$, then:
\begin{align}
    p'(\theta) &= \Tr(E' U_{- \theta} \rho U_{-\theta}^\dagger)\\
    &= \Tr(\id \rho) - \Tr(E U_{\theta_0 + \theta_1 - \theta} \rho U_{\theta_0 + \theta_1 - \theta}^\dagger)\\
    &= 1 - p(\theta_0 + \theta_1 - \theta) = \tilde p(\theta) .
\end{align}
Since $\theta\mapsto U_{-\theta}$ is also a quantum spin-$1$ rotation box, this implies that $\tilde p \in\mathcal{Q}_1$.
\end{proof}

The above two lemmas and Lemma~\ref{lem:charac_R1_faces} imply that  $\mR_1 = \mQ_1$ follows from this lemma:
\begin{lemma}
    The following functions are contained in $\mathcal{Q}_1$:
    \begin{enumerate}
        \item $p(\theta) = \sin^2 \theta$,
        \item $p(\theta) = \sin^4 \frac{\theta}{2}$,
        \item $p(\theta) =c(1-\cos\theta)(1-\cos(\theta-\theta_0))$ for $\theta_0 \in (0, 2 \pi) \setminus\{\pi\}$, where $c>0$ is uniquely determined by the condition $\max_\theta p(\theta)= 1$.
    \end{enumerate}
\end{lemma}
\begin{proof}
Consider the following $\SO(2)$ orbit for a quantum spin-$1$ system:
\begin{align}
    \ket{\psi(\theta)} = \frac{1}{\sqrt{2}} (e^{i \theta} \ket 1 - e^{-i \theta} \ket{-1}).
\end{align}
For effect $E_+ = \ketbra{\phi}{\phi}$ with $\ket{\phi} = \frac{1}{\sqrt{2}} (\ket 1 + \ket{-1})$, we obtain
\begin{align}
    P(+|\theta) = |\!\braket{\phi}{\psi(\theta)} \!|^2 = \frac{1}{4} (e^{i \theta} - e^{-i \theta})^2 =   \sin^2 \theta.
\end{align}
This proves item 1. To show item 2., consider the following orbit:
\begin{align}
      \ket{\psi(\theta)} = \frac{1}2 (e^{i \theta} \ket 1 + \sqrt{2} \ket 0 + e^{-i \theta} \ket{-1}),
\end{align}
and the effect $E_+ = \ketbra{\phi}{\phi}$ with $\ket{\phi} = \frac{1}{2} ( - \ket 1 + \sqrt{2} \ket 0 - \ket{-1})$. They generate the conditional probability
\begin{align}
    P(+|\theta) &= \frac{1}{16}\left(2 - e^{i \theta} - e^{-i\theta}\right)^2 =\sin^4 \frac\theta 2.
\end{align}
Finally, let us prove item 3. First, define
\[
\theta_1:=\left\{
   \begin{array}{cl}
      \frac{\theta_0} 2+\pi & \mbox{if }0<\theta_0<\pi, \\
      \frac {\theta_0} 2 & \mbox{if }\pi<\theta_0<2\pi.
   \end{array}
\right.
\]
Note that $\frac\pi 2 <\theta_1<\frac{3\pi}2$. Now define
\[
\alpha:=\sqrt{1-\frac 1 {1-\cos\theta_1}},\quad \beta:=\frac 1 {\sqrt{2(1-\cos\theta_1)}},
\]
then $|\alpha|^2+2|\beta|^2=1$. Consider the orbit
\[
|\psi(\theta)\rangle=\alpha|0\rangle+\beta e^{i\theta}|1\rangle+\beta e^{-i\theta}|-1\rangle,
\]
and the effect $E_+:=|\psi(\theta_1)\rangle\langle\psi(\theta_1)|$. Then we have
\[
\langle \psi(\theta_1)|\psi(\theta)\rangle=\frac{\cos(\theta-\theta_1)-\cos\theta_1}{1-\cos\theta_1},
\]
and the square of this expression becomes
\begin{eqnarray*}
P(+|\theta)&=& \frac 1 {4 \sin^4\frac{\theta_1} 2} (1-\cos\theta)(1-\cos(2\theta_1-\theta))\\
&=&\frac 1 {4 \sin^4\frac{\theta_1} 2} (1-\cos\theta)(1-\cos(\theta-\theta_0)).
\end{eqnarray*}
By construction, $P(+|\theta_1)=1$, and this is the maximal value over all $\theta$. Thus, we have shown that the family of functions of item 3.\ is contained in $\mathcal{Q}_1$.
\end{proof}

\subsection{Proof of Lemma~\ref{lem:R1prop}}\label{app:R1prop}

\begin{proof}
\begin{enumerate}[(i)]
   \item 
\begin{enumerate}[(a)]
    \item     We first consider a quantum $\SO(2)$ rotation box and show that is has three perfectly distinguishable states belonging to a common $\SO(2)$ orbit. 

The following three vectors are an orthonormal basis of $\mathbb{C}^3$:
    \begin{align}
       \ket{\bf 1} = \frac{1}{\sqrt{3}} (\ket{0} + \ket{1} + \ket 2), \\
       \ket{\omega} = \frac{1}{\sqrt{3}} (\ket{0} + e^{\frac{2 \pi i}{3}} \ket{1} + e^{\frac{4 \pi i}{3}}\ket 2), \\
       \ket{\omega^2} = \frac{1}{\sqrt{3}} (\ket{0} + e^{\frac{4 \pi i}{3}} \ket{1} + e^{\frac{2 \pi i}{3}}\ket 2) .
    \end{align}

    It is immediate that these states belong to the following $\U(1)$ orbit:

    \begin{align}
        \ket{\psi(\theta)} = 
        \begin{pmatrix}
            1 & 0 & 0 \\
            0 & e^{i \theta} & 0 \\
            0 & 0 & e^{i 2 \theta} 
        \end{pmatrix} \ket{\bf 1}.
    \end{align}
    Using the measurement $\{\ketbra{\omega^a}{\omega^a}\}_{a = 0,1,2}$ allows us to perfectly distinguish the three states $\ket{\psi(0)}=|\mathbf{1}\rangle, \ket{\psi(\frac{2 \pi }{3})}=|\omega\rangle,\ket{\psi(\frac{4 \pi }{3})}=|\omega^2\rangle$.

    By definition, the three probability distributions

\begin{align}
    P(a|\theta) = | \!\braket{\omega^a}{\psi(\theta)} \!|^2  \qquad (a = 0, 1,2),
\end{align}
are in $\mathcal{Q}_1^{\{0,1,2\}}\subset \mR_1^{\{0,1,2\}}$. Thus, according to \Cref{LemRJA}, there is a measurement $\{e_a\}_{a=0,1,2}$ on $\tR_1$ such that
\begin{align}
    P(a|\theta) = e_a \cdot \omega(\theta)  \quad  (a = 0,1,2).
\end{align}
By construction, the measurement $\{e_a\}_{a = 0,1,2}$ perfectly distinguishes the states $\{\omega(0), \omega(\frac{2 \pi}{3}), \omega(\frac{4 \pi}{3})\}$ of $\tR_1$,i.e.
\[
e_a\cdot \omega\left(b\cdot  \frac{2\pi}3\right)=P\left(a\left|b\cdot \frac{2 \pi}3\right.\right)=\delta_{ab} \quad (a,b=0,1,2).
\]
\item If there are $n$ jointly perfectly distinguishable states, then there are also $n$ jointly perfectly distinguishable \textit{pure} states $\omega_1,\ldots,\omega_n$. In particular, there is an effect $e_n$ with $e_n\cdot\omega_1=\ldots=e_n\cdot \omega_{n-1}=0$, but $e_n\cdot \omega_n=1$. Thus, $\omega_1,\ldots,\omega_{n-1}$ are $n-1$ disjoint pure states in a proper face of $\Omega_1$. However, by Theorem 1 of~\cite{smilansky_convex_1985}, there is no face with three or more pure states (aside from the whole state space), since all proper faces are at most one-dimensional.
\end{enumerate}
    \item 
    Consider the following states:
    \begin{align}
        \omega(0) = \begin{pmatrix}
            1 \\
            1 \\
            0 \\
            1 \\
            0
        \end{pmatrix} , \ 
          \omega\left(\frac{\pi}{2}\right) = \begin{pmatrix}
            1 \\
            0 \\
            1 \\
            -1 \\
            0
        \end{pmatrix}, \\ 
          \omega(\pi) = \begin{pmatrix}
            1 \\
            -1 \\
            0 \\
            1 \\
            0
        \end{pmatrix} ,  \
          \omega\left(\frac{3 \pi}{2}\right) = \begin{pmatrix}
            1 \\
            0 \\
            -1 \\
            -1 \\
            0
        \end{pmatrix}.
    \end{align}

    We define the following effects:
    \begin{align}
        e_{\pm \frac{\pi}{2}} = \begin{pmatrix}
            \frac{1}{2} & 0 & 0 & \frac{1}{2} & 0 
        \end{pmatrix}, \\
         e_{0,\pi} = \begin{pmatrix}
            \frac{1}{2} & \frac{1}{2} & 0 & 0 & 0 
        \end{pmatrix}, \\
         e_{\frac{\pi}{2},\frac{3\pi}{2}} = \begin{pmatrix}
            \frac{1}{2} & 0 & \frac{1}{2} & 0 & 0 
        \end{pmatrix}, 
    \end{align}
    One can straightforwadly check that these are indeed valid effects, i.e. they give values in $[0,1]$ when evaluated on the orbit  of pure states $\omega(\theta)$ (and therefore on the who convex set of states):
    \begin{align}
        e_{\pm \frac{\pi}{2}}\cdot \omega(\theta) = \frac{1}{2} + \frac{\cos(2 \theta)}{2} \in [0,1] , \\
        e_{0, \pi} \cdot \omega(\theta) = \frac{1}{2} + \frac{\cos( \theta)}{2} \in [0,1] , \\ 
        e_{\frac{\pi}{2}, \frac{3 \pi}{2}}\cdot \omega(\theta) = \frac{1}{2} + \frac{\sin( \theta)}{2} \in [0,1] \ .
    \end{align}
    The unit effect is:
\begin{align}
    u = \begin{pmatrix}
        1 & 0 & 0 & 0 & 0
    \end{pmatrix} .
\end{align}
    In the following addition is defined mod $2 \pi$. The measurement $\{e_{\pm \frac{\pi}{2}}, u -e_{\pm \frac{\pi}{2}}\} $ can be used to perfectly distinguish the state $\omega(\theta)$ for $\theta \in \{0, \frac{\pi}{2}, \pi ,\frac{3 \pi}{2} \}$ from either of the states $\omega(\theta \pm \frac{\pi}{2})$:
        \begin{align}
            e_{\pm \frac{\pi}{2}} \cdot \omega(\theta) &= 1 \ ,\quad \theta \in \{0, \pi\} , \\
             e_{\pm \frac{\pi}{2}} \cdot \omega(\theta) &= 0 \ ,\quad \theta \in \left\{\frac{\pi}{2}, \frac{3\pi}{2}\right\}  \ .
        \end{align}
        The measurement  $\{e_{0,\pi}, u - e_{0,\pi}\}$ can be used to perfectly distinguish $\omega(0)$ from $\omega(\pi)$:
\begin{align}
    e_{0,\pi} \cdot \omega(0) = 1 , \\
    e_{0,\pi} \cdot \omega(\pi) = 0 .
\end{align}
The measurement  $\{e_{\frac{\pi}{2},\frac{3 \pi}{2}}, u - e_{\frac{\pi}{2},\frac{3 \pi}{2}}\}$ can be used to perfectly distinguish $\omega(\frac{\pi}{2})$ from $\omega(\frac{3 \pi}{2})$:
\begin{align}
     e_{\frac{\pi}{2},\frac{3 \pi}{2}} \cdot \omega\left(\frac{ \pi}{2}\right)= 1 , \\
     e_{\frac{\pi}{2},\frac{3 \pi}{2}} \cdot \omega\left(\frac{3 \pi}{2}\right)= 0 .
\end{align}
Thus, any pair of states in $\{\omega(0), \omega(\frac{\pi}{2}), \omega(\pi), \omega(\frac{3 \pi}{2})\}$ can be perfectly distinguished.
    \item From the existence of four pure pairwise perfectly distinguishable states $\{\omega(0), \omega(\frac{\pi}{2}), \omega(\pi), \omega(\frac{3 \pi}{2})\}$, violation of bit symmetry follows immediately for reversible transformations $T(\theta)$ of the form in Equation~\eqref{eq:R1_SO2_rep}. Take for example the pairs of perfectly distinguishable states $\{\omega(0), \omega(\frac{\pi}{2})\}$ and $\{\omega(0), \omega(\pi)\}$, then there is no reversible transformation $T(\phi)$ mapping one pair to the other, i.e.\  such that $T(\phi) \omega(0) = \omega(0)$ and  $T(\phi) \omega(\frac{\pi}{2}) = \omega(\pi)$.

    However, there exist other transformations $T$ which are symmetries of $\Omega_1$ such as $T = \diag(1,1,-1,1,-1)$. We now show that bit symmetry is violated for all symmetries of $\Omega_1$, not just the $\SO(2)$ subgroup $\{T(\theta)\,\,|\,\,\theta \in [0, 2\pi)\}$.

    Let us denote by $\mathcal{G}$ the group of all symmetries of $\Omega_1$. There exists a group invariant inner product $\langle \cdot,\cdot\rangle$ such that $\langle Gx,Gy\rangle=\langle x,y\rangle$ for all $G\in\mathcal{G}$ and $x,y\in\mathbb{R}^5$. As for every inner product, there is a positive definite symmetric matrix $M>0$, $M=M^\top$, such that $\langle x,y\rangle=x\cdot My$. Group invariance implies that $M$ commutes with all elements of $\mathcal{G}$; in particular, $[M,T(\theta)]=0$ for all $\theta$. A straightforward calculation shows that this implies that $M={\rm diag}(a,b,b,c,c)$ for some $a,b,c>0$. If all pairs of perfectly distinguishable pure states $\omega_1,\omega_2$ were related by a reversible transformation, then their invariant inner products $\langle\omega_1,\omega_2\rangle$ would all be identical. But the following are inner products between pairs of perfectly distinguishable pure states:
\begin{eqnarray*}
\left\langle\omega(0),\omega\left(\frac{3\pi}2\right)\right\rangle&=&a-c,\\ \langle\omega(0),\omega(\pi)\rangle&=&a-b+c, \\
\left\langle \omega(0),\omega\left(\frac{2 \pi} 3 \right)\right\rangle &=& a-\frac 1 2 b-\frac 1 2 c.
\end{eqnarray*}
For these to be identical, we would need to have $b=c=0$, which contradicts the positive definiteness of $M$. Thus, bit symmetry cannot hold.
    \end{enumerate}
\vskip -1em$\strut$
\end{proof}

\section{SDP-based algorithm to explore the correlations set boundaries}
\label{app:SDPboundary}
Here we outline an algorithm to numerically explore and compare the boundary of the correlations sets $\mathcal{Q}_J, \mathcal{R}_J$ which in \Cref{SubsecGap} has led to the derivation of an inequality proving $\mathcal{Q}_J \subsetneq \mathcal{R}_J$ for $J\geq 3/2$. The idea is to first choose a plane in some direction of the trigonometric coefficients affine space, and then discretize a circle around its origin to use the SDP-based methodologies in \Cref{SubsecRelaxation} to probe the boundary of the sets $\mathcal{Q}_J, \mathcal{R}_J$ for that particular plane. In other words, we numerically find a 2D projection of the sets  $\mathcal{Q}_J, \mathcal{R}_J$ in the trigonometric coefficient space.

In particular, the algorithm goes as follows:
\begin{enumerate}
    \item Select two directions $\mathbf{v}_1:=(\mathbf{c}_1,\mathbf{s}_1)$, $\mathbf{v}_2:=(\mathbf{c}_2,\mathbf{s}_2)$ in the ($4J+1$)-dimensional affine space to define the plane.
    \item Parametrize a direction in the plane $\mathbf{p}=\cos(\theta)\mathbf{v}_1+\sin(\theta)\mathbf{v}_2$, for some angle $\theta$.
    \item Use the SDP in \cref{rotationBoxSDP} to find the boundary of $\mathcal{R}_J$ in the direction $\mathbf{p}$ and/or the see-saw methodology presented in \Cref{SubsecRelaxation} to approximate the boundary of $\mathcal{Q}_J$ in the direction $\mathbf{p}$.
    \item Repeat step 3 for all values of $\theta\in\{0,\ldots,2\pi\}$ to complete a full circle discretized up to desired numerical accuracy.
\end{enumerate}
In \Cref{fig:2DshadowJthreeHalves} of the main text, we present an example of the final result for $J=3/2$ in the plane given by the directions $\mathbf{v}_1=(c_0,c_1,c_2,c_3,s_1,s_2,s_3)=(0,0,1,0,0,0,0)$ and $\mathbf{v}_2=(0,0,0,0,0,0,1)$ (\textit{i.e.,} the $c_2$-$s_3$ plane).

\section{Several results and proofs for Section~\ref{SubsecGap}}
\subsection{Proof of Lemma~\ref{LemEigenvalue}}
\label{app:proofBhatia}

In the following, we will denote the eigenvalues of any self-adjoint $n\times n$ matrix $A$ in decreasing order by $\lambda_1(A),\lambda_2(A),\ldots,\lambda_n(A)$ such that $\lambda_1(A)\geq \lambda_2(A)\geq\ldots\geq \lambda_n(A)$.
\begin{lemma}
Consider the $4\times 4$ block matrix
\[
   M=\left(\begin{array}{cc} 0 & B \\ B^\dagger & 0 \end{array}\right),
\]
where $B$ is a $2\times 2$ matrix. Then its eigenvalues are
\begin{eqnarray*}
\left(\strut\lambda_1(M),\lambda_2(M),\lambda_3(M),\lambda_4(M)\right)=\\
\left(\sqrt{\lambda_1(B^\dagger B)},\sqrt{\lambda_2(B^\dagger B)},-\sqrt{\lambda_2(B^\dagger B)},-\sqrt{\lambda_1(B^\dagger B)}\right).
\end{eqnarray*}
\end{lemma}
\begin{proof}
We have
\[
   M^2 = \left(\begin{array}{cc}B B^\dagger & 0 \\ 0 & B^\dagger B\end{array}\right).
\]
Thus, the squares of the eigenvalues of $M$ are the eigenvalues of $B B^\dagger$ and $B^\dagger B$, which are known to agree. Up to a sign, this determines the eigenvalues of $M$, and the signs in turn are determined by ${\rm Tr}(M)=0=\sum_i \lambda_i(M)$.
\end{proof}
Applying this lemma to the matrix $M[E]$, we obtain
\[
\lambda_1(M[E])=\sqrt{\lambda_1(B[E]^\dagger B[E])},
\]
where $B[E]=\left(\begin{array}{cc} E_{20} & -i E_{30} \\ 0 & E_{31} \end{array}\right)$. It is straightforward to compute the eigenvalues of the matrix $B[E]^\dagger B[E]$, and the result proves \Cref{LemEigenvalue}.

\subsection{Proof that $\beta = \frac{1}{\sqrt{3}}$}
\label{app:maxPolytope}

The feasible set for the optimization problem in \Cref{eqMaxOverPolytope} is given by a polytope $R$ with vertices $\{(0,0,0),(0,0,\frac{1}{4}),(0,\frac{1}{4},0),(\frac{1}{4},0,0),(\frac{1}{4},0,\frac{1}{4})\}$ (see \Cref{fig:feasibleRegion} for an illustration). Our goal is to compute the maximum of the function
\[
f(x,y,z):=x+y+z+\sqrt{(x+y+z)^2-4xz}
\]
over all $(x,y,z)\in R$. We find that $\nabla f=0$ has no solutions in the topological interior of $R$, hence the maximum must be attained on one of the lower-dimensional faces of this polytope.

There are five two-dimensional faces $F_1,\ldots,F_5$, but $f$ restricted to face $F_i$ has no stationary points in the relative interior of $F_i$, for all $i$. For example, if we define the face $F_1$ by the condition $x=0$, it is parametrized by $0\leq y\leq\frac 1 4$ and $0\leq z\leq \frac 1 4$. The function $f$ becomes $f_{F_1}(y,z)=2(y+z)$, and $(\partial_y f_{F_1},\partial_z f_{F_1})=(2,2)\neq (0,0)$ in the relative interior (where $0<y,z<\frac 1 4$) of $F_1$, and so $f$ cannot have any local maxima there.

Thus, the global maximum must be attained on one of the eight edges $E_1,\ldots, E_8$ (one-dimensional faces) or one of the five vertices $V_1,\ldots,V_5$ (zero-dimensional faces). For seven of the edges, $E_1,\ldots, E_7$, $f$ has no stationary points in their relative interior, but on one of the edges it does: define $E_8$ as the points in $R$ with $x+y=\frac 1 4$ and $y+z=\frac 1 4$, which we can parametrize via $0\leq x \leq \frac 1 4$ and $y=\frac 1 4-x$, $z=x$, such that
\[
f_{E_8}(x)=x+\frac 1 4+\sqrt{\left(x+\frac 1 4\right)^2-4x^2}.
\]
Then $f'_{E_8}(x)=0$ has a solution in the interior $0<x<\frac 1 4$, namely $x=\frac 1 6$, and $f_{E_8}(\frac 1 6)=\frac 2 3$. Indeed, this is the global maximum, since $f$ attains only the values $0$ and $\frac 1 2$ on the vertices $V_1,\ldots,V_5$.

We thus find $\max_{(x,y,z)\in R}f(x,y,z)=f(\frac{1}{6},\frac{1}{12},\frac{1}{6})=\frac{2}{3}\geq 2\beta^2$. This gives the bound $\beta\leq\frac{1}{\sqrt{3}}$.

\begin{figure}
\centering 
\includegraphics[width=.6\columnwidth]{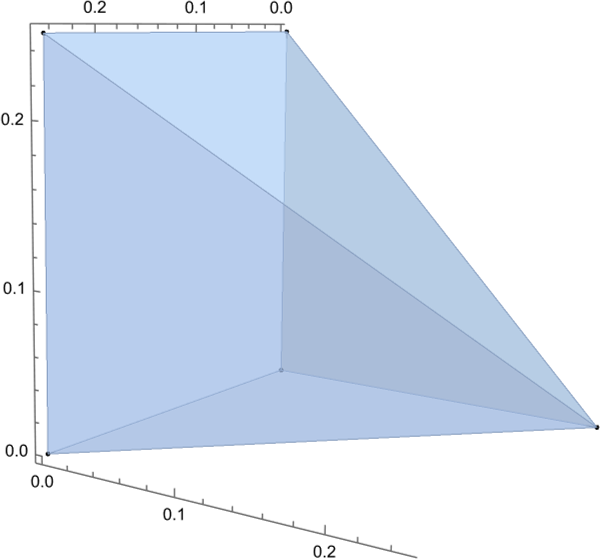}
\caption{Region $R$, defined by the constraints $x,y,z\geq 0, x+y\leq 1/4,y+z\leq 1/4$.}
\label{fig:feasibleRegion}
\end{figure}

The bound can be attained by a POVM that satisfies $|E_{02}|^2=|E_{20}|^2=\frac{1}{6}$, $|E_{03}|^2=|E_{30}|^2=\frac{1}{12}$ and $|E_{13}|^2=|E_{31}|^2=\frac{1}{6}$. Using semidefinite programming, we found the following possible solution for $E$:
\[
E=\left(\begin{array}{cccc}	
      \frac{1}{2} & 0 & \frac{1}{\sqrt{6}} & -\frac{i}{2\sqrt{3}} \\
      0 & \frac{1}{2} & -\frac{i}{2\sqrt{3}} & \frac{1}{\sqrt{6}} \\
      \frac{1}{\sqrt{6}} & \frac{i}{2\sqrt{3}} & \frac{1}{2} & 0 \\
      \frac{i}{2\sqrt{3}} & \frac{1}{\sqrt{6}} & 0 & \frac{1}{2} \end{array}\right),
\]
for which the state $\rho$ would be given by
\[
\rho= \left(\begin{array}{cccc}	
      \frac{1}{3} & \frac{1}{3\sqrt{2}} & \frac{1}{3\sqrt{2}} & \frac{1}{3} \\
      \frac{1}{3\sqrt{2}} & \frac{1}{6} & \frac{1}{6} & \frac{1}{3\sqrt{2}} \\
      \frac{1}{3\sqrt{2}} & \frac{1}{6} & \frac{1}{6} & \frac{1}{3\sqrt{2}} \\
      \frac{1}{3} & \frac{1}{3\sqrt{2}} & \frac{1}{3\sqrt{2}} & \frac{1}{3} \end{array}\right).
\] 

\subsection{Proof that the quantum correlations satisfy  $c_{2J-1}+s_{2J}\leq \beta = \frac{1}{\sqrt{3}}$}
\label{app:quantumBoundGralJ}
\begin{lemma}
Let $P\in\mathcal{Q}_J$ for $J\geq \frac 3 2$, then its trigonometric coefficients as defined in Lemma~\ref{LemmReal2Complex} satisfy
\[
c_{2J-1}+s_{2J}\leq\frac 1 {\sqrt{3}}.
\]
\end{lemma}
\begin{proof}
The proof follows closely the lines of the $J=3/2$ case, proven in \Cref{SubsecGap}. Here we briefly describe the relevant adaptations. First, we have
\begin{eqnarray}
(c_{2J-1}+s_{2J})[p]&=&2\,{\rm Re}(a_{2J-1}[p])-2\,{\rm Im}(a_{2J}[p])\nonumber\\
&=& 2\,{\rm Re}(Q_{0,2J-1}+Q_{1,2J})-2\,{\rm Im}(Q_{0,2J})\nonumber\\
&=&2\,{\rm Re}(E_{0,2J-1}\rho_{0,2J-1}+E_{1,2J}\rho_{1,2J})\nonumber\\
&&-2\,{\rm Im}(E_{0,2J}\rho_{0,2J})\nonumber\\
&=&{\rm Tr}(M[E]\rho),\label{eqBoundAppendix}
\end{eqnarray}
where now the matrix $M[E]$ is given, in block-matrix notation,
\[
M[E]=\left( \begin{array}{ccc}
   \mathbf{0}_{2\times 2} & \mathbf{0}_{2\times (2J-3)} & B[E] \\ \mathbf{0}_{(2J-3)\times 2} & \mathbf{0}_{(2J-3)\times(2J-3)} & \mathbf{0}_{(2J-3)\times 2}\\
   B(E)^\dagger & \mathbf{0}_{2\times (2J-3)} & \mathbf{0}_{2\times 2}
\end{array}\right),
\]
and $B[E]=\left(\begin{array}{cc} E_{2J-1,0} & -i E_{2J,0} \\ 0 & E_{2J,1} \end{array}\right)$. Maximizing Eq.~(\ref{eqBoundAppendix}) over all quantum states $\rho$ will again give us the maximal eigenvalue of $M[E]$. Since
\[
M[E]^2 = \left( \begin{array}{ccc}
   B[E]B[E]^\dagger & \mathbf{0}_{2\times (2J-3)} & \mathbf{0}_{2\times 2} \\ \mathbf{0}_{(2J-3)\times 2} & \mathbf{0}_{(2J-3)\times (2J-3)} & \mathbf{0}_{(2J-3)\times 2} \\ \mathbf{0}_{2\times 2} & \mathbf{0}_{2\times (2J-3)} & B[E]^\dagger B[E]
\end{array}
\right),
\]
we obtain again $\lambda_1(M[E])=\lambda_1(B[E]^\dagger B[E])$, and this eigenvalue can be bounded exactly as in the $(J=3/2)$-case by using that
\[
   |E_{2J-1,0}|^2 +|E_{2J,0}|^2\leq \frac 1 4,\quad |E_{2J,0}|^2+|E_{2J,1}|^2\leq \frac 1 4.
\]
We hence obtain exactly the same upper bound of $1/\sqrt{3}$.
\end{proof}

\subsection{Examples of correlations in $\mathcal{R}_J\setminus\mathcal{Q}_J$ for $J\geq 2$}
\label{SubsecExamplesSpin2andHigher}
We begin with the case $J\geq 7/2$.
\begin{lemma}
For every $J\geq 7/2$, we have $\mathcal{Q}_J\subsetneq \mathcal{R}_J$.
\end{lemma}
\begin{proof}
For $\beta\geq 0$ and $J\geq 1$, consider the following trigonometric polynomial
\begin{eqnarray}
   p_{J,\beta}(\theta) & := & \frac 1 2 +\frac 1 4\beta \sin(2J\theta) \\
   && \quad -\frac 3 4 \beta \sum_{k=1}^{2J-1} \left(\frac 1 2\right)^{2J-k}\sin\left[k\left(\frac\pi 2+\theta\right)-J\pi\right].\nonumber
\end{eqnarray}
This is a trigonometric polynomial (in $\theta$) of degree $2J$ with $s_{2J}=\frac 1 4 \beta$ and $c_{2J-1}=\frac 3 8 \beta$, coming from an educated guess based on numerical results. If we can show that it satisfies $0\leq p_{J,\beta}(\theta)\leq 1$ for all $\theta$, for some $\beta$ that is sufficiently close to $1$, we have a non-quantum rotation box, since $\frac 1 4+\frac 3 8=0.625$. The polynomial has the following closed-form expression
\[
   p_{J,\beta}(\theta)=\frac 1 2 +\frac 1 2 \beta \frac{(-3)\cdot 4^{-J}(\cos\theta+2\sin(J\pi))+f_J(\theta)
   }{5+4\sin\theta},
\]
where
\begin{eqnarray*}
   f_J(\theta)&=&4\cos(\theta-2J\theta)-\cos(\theta+2J\theta)+4\sin(2J\theta)\\
   &=&3\cos\theta\cos^2(J\theta)+(8+10\sin\theta)\cos(J\theta)\sin(J\theta)\\
   &&-3\cos\theta\sin^2(J\theta)\\
   &=&\left(\begin{array}{c} \cos(J\theta) \\ \sin(J\theta)\end{array}\right)\cdot
   \left(\begin{array}{cc} 3\cos\theta & 4+5\sin\theta \\ 4+5 \sin\theta & -3\cos\theta\end{array}\right)\cdot \left(\begin{array}{c} \cos(J\theta) \\ \sin(J\theta)\end{array}\right).
\end{eqnarray*}
The result must be between the smallest and largest eigenvalues of this matrix, and those eigenvalues turn out to be $-5-4\sin\theta$ and $5+4\sin\theta$. Thus,
\[
-5-4\sin\theta\leq f_J(\theta)\leq 5+4\sin\theta.
\]
Since $f_J(\theta)$ is by far the dominant term in the numerator (the other part goes to zero exponentially in $J$), we have almost shown that $p_{J,\beta=1}$ is a valid rotation box. Now let us be more careful and scale a bit with $\beta<1$.
Clearly, $p_{J,\beta}\in\mathcal{R}_J$ if and only if
\[
\left| \beta \frac{(-3)\cdot 4^{-J}(\cos\theta+2\sin(J\pi))+f_J(\theta)
   }{5+4\sin\theta}\right|\leq 1 \mbox{ for all }\theta.
\]
But, due to what we have just shown, the left-hand side is upper-bounded by $\beta(3\cdot 4^{-J}(1+2)+1)$, and hence
$p_{J,\beta=1/(1+9\cdot 4^{-J})}\in\mathcal{R}_J$. This establishes a gap if
\[
   c_{2J-1}+s_{2J}=0.625\beta>\frac 1 {\sqrt{3}},
\]
which is the case for all $J\geq 7/2$.
\end{proof}
In what follows we treat the remaining cases $J=3/2,2,5/2,3$ on a case-by-case basis.

The way we proceed is by finding explicit counterexamples for each remaining $J$. These counterexamples have been found numerically via the following SDP based on \cref{rotationBoxSDP}:
\begin{equation}
\begin{aligned}
\max_{Q,S,\bm{a}} \quad & c_{2J-1}+s_{2J}\\
\textrm{s.t.} \quad & \bullet \; a_k=\sum_{0\leq j,j+k\leq2J}Q_{j,j+k} \, \text{ for all } k,\\
& \bullet \; a_k=-\sum_{0\leq j,j+k\leq2J}S_{j,j+k} \, \text{ for all } k\neq 0,\\
& \bullet \; 1-a_0=\mathrm{Tr}(S),\\
  & \bullet \; Q,S\geq 0 .
\end{aligned}
\end{equation}
When the SDP is feasible, it finds some $(2J+1)\times (2J+1)$ matrices $Q,S$ and some complex variables $a_k$ with $k\in\{0,\ldots,2J\}$ thus obtaining a valid rotation box correlation (c.f.\ \Cref{lemRotationBox}). Then, if these values lead to $c_{2J-1}+s_{2J}>\frac{1}{\sqrt{3}}$ then the correlation goes beyond the quantum bound and we have the counterexample.

As an example, let us take the case $J=3/2$ with the coefficients $c_0=2/5, c_1=0,c_2=48/125,c_3=0$, $s_0=0,s_1=6/25,s_2=0,s_3=32/125$. Then, one can check that this forms a valid spin-$3/2$ correlation since one can define matrices $Q, S\geq 0$ fulfilling \Cref{lemRotationBox} such as
\[
Q_{3/2}:=\frac{1}{125}\left(\begin{array}{cccc}	
      16 & -12i & 12 & -16i \\
      12i & 9 & 9i & 12 \\
       12& -9i & 9 & -12i \\
      16i & 12 & 12i & 16 \end{array}\right) \geq 0,
\]
\[
S_{3/2}:=\frac{1}{125}\left(\begin{array}{cccc}	
      24 & 2i & -12 & 16i \\
      -2i & 27/2 & 11i & -12 \\
       -12& -11i & 27/2 & 2i \\
      -16i & -12 & -2i & 24 \end{array}\right) \geq 0.
\]
Finally, observe that for this case we have $(c_{2J-1}+s_{2J})[p^*]=\frac{78}{125}=0.624>\frac{1}{\sqrt{3}}$ and thus the point lies outside of $\mathcal{Q}_{3/2}$. The same follows for the remaining cases $J=2,5/2,3$, for which for the sake of completion we proceed to provide some numerically found examples and their corresponding $Q_{J},S_{J}$ certificates.

\paragraph*{J=2.}
Consider now $c_0=1/2,c_1=-17/250,c_2=0,c_3=19/50,c_4=0$, $s_0=0,s_1=0,s_2=87/500,s_3=0,s_4=6/25$ such that $(c_{2J-1}+s_{2J})[p^*]=0.62>\frac{1}{\sqrt{3}}$ and to fulfill \Cref{lemRotationBox} define the following matrices:

\bigskip

\begin{widetext} 
\[
Q_{2}:=\left(\begin{array}{ccccc}	
      377/2400 & -35/2438 - 95/2314i & 11/782 - 31/743i &  19/200 & -3/25i\\ 
 -35/2438 + 95/2314i & 62/811 & -4/1513 + 95/2314i & -273/9704 - 3/844i & 
 19/200 \\ 11/782 + 31/743i & -4/1513 - 95/2314i & 
 243/7378 & -4/1513 + 95/2314i & 
 11/782 - 31/743i\\  19/200 & -273/9704 + 3/844i & -4/1513 - 
  95/2314i & 62/811 & -35/2438 - 95/2314i\\  3/25i & 19/200 & 
 11/782 + 31/743i & -35/2438 + 95/2314i & 377/2400 \end{array}\right) ,
\]
\[
S_{2}:=\left(\begin{array}{ccccc}	
      377/2400 & 35/2438 - 95/2314i & 11/782 + 31/743i & -19/200 & 3/25i \\
      35/2438 + 95/2314i & 62/811 & 4/1513 + 95/2314i & -273/9704 + 3/844i & -19/200 \\
       11/782 - 31/743i & 4/1513 - 95/2314i & 243/7378 & 4/1513 + 95/2314i & 
11/782 + 31/743i \\
      -19/200 & -273/9704 - 3/844i & 4/1513 - 95/2314i & 62/811 &
35/2438 - 95/2314i \\
-3/25i & -19/200 & 11/782 - 31/743i & 35/2438 + 95/2314i & 377/2400
\end{array}\right).
\]   
\end{widetext}

\paragraph*{J=5/2.}
In this case one can take $c_0=0.5261,c_1=0,c_2=-0.1044,c_3=0,c_4=0.3695,c_5=0$, $s_0=0,s_1=-0.0639,s_2=0,s_3=0.1926,s_4=0,s_5=0.2564$ such that $(c_{2J-1}+s_{2J})[p^*]=0.626>\frac{1}{\sqrt{3}}$ and to fulfill \Cref{lemRotationBox} define the following matrices
\begin{widetext}
\[
Q_{5/2}:=\left(\scalemath{0.75}{\begin{array}{cccccc}	
   0.1665 & - 0.0320i & -0.0022 & - 0.0374i &  0.0924 & - 0.1282i \\
   0.0320i  & 0.0739 & + 0.0387i & -0.0239 & -0.0214i  & 0.0924 \\
  -0.0022 & - 0.0387i  & 0.0227 & + 0.0185i & -0.0239 & -0.0374i \\
   0.0374i & -0.0239 & - 0.0185i &  0.0227 & 0.0387i & -0.0022 \\
   0.0924 & 0.0214i & -0.0239 & -0.0387i  & 0.0739 & -0.0320i \\
  0.1282i &  0.0924 & 0.0374i & -0.0022 & 0.0320i &  0.1665 \end{array}}\right),
  S_{5/2}:=\left(\scalemath{0.75}{\begin{array}{cccccc}	
       0.1465 &   -0.0664i &  0.0541 &  0.0618i & -0.0924  & 0.1282i \\
    0.0664i  & 0.0644 & 0.0390i & -0.0280 & -0.0274i & -0.0924 \\
   0.0541 & - 0.0390i &  0.0261 & 0.0228i & -0.0280 &  0.0618i \\
   -0.0618i & -0.0280 &  - 0.0228i  & 0.0261& + 0.0390i  & 0.0541 \\
  -0.0924 & 0.0274i & -0.0280 & -0.0390i   &0.0644 & -0.0664i\\
  -0.1282i  &-0.0924 & -0.0618i &  0.0541 & 0.0664i &  0.1465 
\end{array}}\right).
\]

\end{widetext}

\paragraph*{J=3.}
Finally, in this case one can take $c_0=1/2,c_1=0.0173,c_2=0,c_3=-0.0915,c_4=0,c_5=0.3763,c_6=0$, $s_0=0,s_1=0,s_2=-0.0433,s_3=0,s_4=0.1864,s_5=0,s_6=0.2485$ such that $(c_{2J-1}+s_{2J})[p^*]=0.6248>\frac{1}{\sqrt{3}}$ and to fulfill \Cref{lemRotationBox} define the following matrices
\begin{widetext}
\[
Q_{3}:=\left(\scalemath{0.75}{\begin{array}{ccccccc}	
         0.1564 &  0.0032 - 0.0459i &  0.0220 + 0.0097i & -0.0221 + 0.0097i & -0.0032 - 0.0461i &  0.0941 &  - 0.1242i \\
   0.0032 + 0.0459i  & 0.0711 &   0.0008 + 0.0355i & -0.0180 + 0.0009i & -0.0008 - 0.0097i &  0.0064 - 0.0010i &  0.0940 \\
   0.0220 - 0.0097i &  0.0008 - 0.0355i &  0.0183 &   0.0003 + 0.0105i & -0.0080 + 0.0004i & -0.0008 - 0.0097i & -0.0032 - 0.0461i \\
  -0.0221 - 0.0097i & -0.0180 - 0.0009i &  0.0003 - 0.0105i &  0.0083 &   0.0003 + 0.0105i & -0.0180 + 0.0009i & -0.0220 + 0.0097i \\
  -0.0032 + 0.0461i & -0.0008 + 0.0097i & -0.0080 - 0.0004i &  0.0003 - 0.0105i &  0.0183 &   0.0008 + 0.0355i  & 0.0220 + 0.0097i \\
   0.0941 &   0.0064 + 0.0010i & -0.0008 + 0.0097i & -0.0180 - 0.0009i  & 0.0008 - 0.0355i  & 0.0712 &  0.0032 - 0.0460i \\
   0.1242i &  0.0940 & -0.0032 + 0.0461i & -0.0220 - 0.0097i &  0.0220 - 0.0097i &  0.0032 + 0.0460i &  0.1563 \end{array}}\right) ,
\]
\[
S_{3}:=\left(\scalemath{0.75}{\begin{array}{ccccccc}	
      0.1563 &  -0.0032 - 0.0460i &  0.0220 - 0.0097i  & 0.0220 + 0.0097i & -0.0032 + 0.0461i & -0.0940 &  0.1242i \\
  -0.0032 + 0.0460i &  0.0712 &  -0.0008 + 0.0355i & -0.0180 - 0.0009i &  0.0008 - 0.0097i &  0.0064 + 0.0010i & -0.0941 \\
   0.0220 + 0.0097i & -0.0008 - 0.0355i &  0.0183 & -0.0003 + 0.0105i & -0.0080 - 0.0004i &  0.0008 - 0.0097i & -0.0032 + 0.0461i \\
   0.0220 - 0.0097i & -0.0180 + 0.0009i & -0.0003 - 0.0105i  & 0.0083 &  -0.0003 + 0.0105i & -0.0180 - 0.0009i  & 0.0220 + 0.0097i \\
  -0.0032 - 0.0461i &  0.0008 + 0.0097i & -0.0080 + 0.0004i & -0.0003 - 0.0105i &   0.0183 & -0.0008 + 0.0355i  & 0.0220 - 0.0097i \\
  -0.0940 &   0.0064 - 0.0010i &  0.0008 + 0.0097i & -0.0180 + 0.0009i & -0.0008 - 0.0355i &  0.0712 &  -0.0032 - 0.0459i \\
  0.1242i & -0.0941 & -0.0032 - 0.0461i  & 0.0220 - 0.0097i  & 0.0220 + 0.0097i & -0.0032 + 0.0459i &  0.1564
\end{array}}\right).
\]    
\end{widetext}

\section{Proofs for Section~\ref{sec:Jinftyquantum}: $J\rightarrow\infty$}
    \label{inftyproofs}
Here we will present the details of the proof of Theorem \ref{inftyrotasquantumboxes}. The first step of the proof in the main text is given by Lemma \ref{firststep}, the second step by Lemma \ref{secondstep} and the final and third step is presented right after the proof of Lemma \ref{secondstep}.
We will consider the Hilbert space $L^2({\rm SO}(2))$, with inner product
\[
\langle f,g\rangle=\frac 1 {2\pi}\int_0^{2\pi} \overline{f(\theta)}g(\theta)\,{\rm d}\theta.
\]
It carries the regular representation of ${\rm SO}(2)$, defined by $(U(\theta)f)(\theta'):=f(\theta'+\theta)$.
As usual, we will pick a representative $f$ of $[f]\in L^2(SO(2))$ whenever we do concrete calculations. All angle additions (like $\theta+\theta'$ or $\theta_0-1/n$) are understood modulo $(2\pi)$.
\begin{lemma}\label{firststep}
 Let $P\in \mathcal{R}_\infty$, then we can write it as a limit of a convergent sequence $P(+|\theta_0+\theta')=\lim_{n\to\infty}\braket{U^\dagger(\theta')f_{\theta_0,n}}{\hat P U^\dagger(\theta')f_{\theta_0,n}}$, where $f_{\theta_0,n}\in L^2({\rm SO}(2))$ for all $n\in\mathbb{N}$ and $\theta_0$, while $0\leq\hat P \leq \mathbf{1}$.
\end{lemma}
\begin{proof}
For the choice of $P\in\mathcal{R}_\infty$, we begin by defining an associated operator on $L^2({\rm SO}(2))$
\[
   (\hat P \psi)(\theta):=P(\theta)\psi(\theta).
\]
It is easy to see that $\hat P$ is a bounded, self-adjoint operator. Furthermore,
\begin{eqnarray*}
\langle\psi|\hat P|\psi\rangle &=&\frac 1 {2\pi}\int_0^{2\pi} \overline{\psi(\theta)}P(\theta)\psi(\theta)\,{\rm d}\theta\in[0,\langle\psi,\psi\rangle],
\end{eqnarray*}
and so $0\leq \hat P\leq \mathbf{1}$, i.e.\ $\hat P$ defines a valid POVM element.

We define
\begin{eqnarray}
    f_{\theta_0,n}:=\sqrt{\frac \pi n}\,\chi_{[\theta_0-\frac{1}{n},\theta_0+\frac{1}{n}]},
\end{eqnarray} 
where \begin{eqnarray}
    \chi_{[\theta_0-\frac{1}{n},\theta_0+\frac{1}{n}]}(\theta)=\left\{\begin{matrix}
 1&\mbox{if}& \theta_0-\frac{1}{n}\leq\theta\leq\theta_0+\frac{1}{n}\\
 0&\mbox{else}&
\end{matrix}\right.,
\end{eqnarray}
and it is clear that $ f_{\theta_0,n}\in L^2({\rm SO}(2))$. Furthermore, it is easy to show that $\|f_{\theta_0,n}\|=1$ for all $\theta,n$.

    Now, for $0<\theta_0<2\pi$ and $n$ large enough, we calculate
    \begin{eqnarray*}
        \|(\hat P-P(\theta_0)\id)f_{\theta_0,n}\|^2&=&\frac{1}{2\pi}\int_{\theta_0-\frac{1}{n}}^{\theta_0-\frac{1}{n}} (P(\theta)-P(\theta_0))^2 f^2_{\theta_0,n}(\theta){\rm d}\theta\nonumber\\
        &\leq& (P(\Delta_{\mbox{\tiny{max}}}(n))-P(\theta_0))^2\|f_{\theta_0,n}\|^2\nonumber\\
        &=&(P(\Delta_{\mbox{\tiny{max}}}(n))-P(\theta_0))^2,
    \end{eqnarray*}
    where $\theta_0-1/n\leq \Delta_{\mbox{\tiny{max}}}(n)\leq \theta_0+1/n$ is chosen such that $(P(\Delta_{\mbox{\tiny{max}}}(n))-P(\theta_0))^2$ is maximal. Since $P$ is continuous, it follows that $ \|(\hat P-P(\theta_0)\id))f_{\theta_0,n}\|\to 0$ for $n\to \infty$. 
    Now we use the Cauchy-Schwarz inequality to show 
    \begin{eqnarray}
          \|(\hat P-P(\theta_0)\id)f_{\theta_0,n}\|&=&\|f_{\theta_0,n}\|
          ´\cdot\|(\hat P-P(\theta_0)\id)f_{\theta_0,n}\|\nonumber\\
          &\geq& \left|\braket{f_{\theta_0,n}}{(P(\hat{\theta}-P(\theta_0)\id))f_{\theta_0,n}}\right|\nonumber.
    \end{eqnarray}
Hence,
\begin{equation}
    \lim_{n\to\infty} \braket{f_{\theta_0,n}}{(\hat P-P(\theta_0)\id)f_{\theta_0,n}}=0.
\end{equation}
We can rewrite this as
\begin{eqnarray*}
    \lim_{n\to\infty} \braket{f_{\theta_0,n}}{\hat P f_{\theta_0,n}}=P(\theta_0).
\end{eqnarray*}
The above is also true for $\theta=0$ if all angles are understood modulo $2\pi$. In a final step, we consider the transformation of $f_{\theta_0,n}$ under the regular representation $U$ of ${\rm SO}(2)$:
\begin{eqnarray}
    U(\theta')f_{\theta_0,n}(\theta)=f_{\theta_0,n}(\theta+\theta')=f_{\theta_0-\theta',n}(\theta).
\end{eqnarray}
Hence we can write
\begin{eqnarray}
P(+|\theta_0+\theta')&=&\lim_{n\to\infty}\braket{U^\dagger(\theta')f_{\theta_0,n}}{\hat P U^\dagger (\theta')f_{\theta_0,n}}\nonumber\\
    &=& P(\theta_0+\theta').
\end{eqnarray}
The claim follows.
\end{proof}
Lemma~\ref{firststep} implies that given any $P\in \mathcal{R}_\infty$, we can approximate it arbitrarily well by
\begin{eqnarray}
P_n(+|\theta_0+\theta')=\braket{U^\dagger(\theta')f_{\theta_0,n}}{\hat P U^\dagger(\theta')f_{\theta_0,n}}.
\end{eqnarray}
The following standard definitions can be found, for example, in Ref.~\cite{Nielsen2010}.
\begin{definition}
 For two probability distributions $\{p_x\}$ and $\{q_x\}$ we define the classical trace distance by
 \begin{equation*}
     \tilde{D}(p_x,q_x)=\frac{1}{2}\sum_x|p_x-q_x|.
 \end{equation*}    
\end{definition}
We observe that for $x\in\{\pm\}$, we have
\[
    \tilde{D}(p_x,q_x)=\frac{1}{2}(|p_+-q_+|+|p_--q_-|)=|p_+-q_+|.
\]
\begin{definition}
    Let $A\in\mathcal{T}(\mathcal{H})$, we define the norm\begin{eqnarray*}
        \|A\|_1={\rm Tr}(|A|),
    \end{eqnarray*}
    where $|A|=\sqrt{A^*A}$.
\end{definition}
\begin{definition}
   Let $A,B \in \mathcal{T}(\mathcal{H})$,we define the trace distance\begin{eqnarray}
        D(A,B)=\frac{1}{2}\|A-B\|_1.
    \end{eqnarray}
    \end{definition}
We will write
\begin{eqnarray}
    \sigma_{\theta_0,n}=\ketbra{f_{\theta_0,n}}{f_{\theta_0,n}}.
\end{eqnarray}
From the Peter-Weyl Theorem~\cite{Robert}, we know that $L^2({\rm SO}(2))=\bigoplus_{j\in\mathbb{Z}}\mathcal{H}_j$, where \begin{eqnarray*}
    \mathcal{H}_j=\mbox{span}\{\phi_j\,\,|\,\,\phi_j(\alpha)=e^{ij\alpha}\}.
\end{eqnarray*}
In the orthonormal basis $\{\phi_j\}_{j\in\mathbb{Z}}$, we can write \begin{eqnarray}
    U(\theta)=\sum_{j=-\infty}^\infty e^{ij\theta} \ketbra{\phi_j}{\phi_j}.
\end{eqnarray}
Furthermore, we define the projector $\Pi_J$ onto the finite-dimensional subspace $\mathcal{H}_{\leq J}=\bigoplus_{j=-J}^J \mathcal{H}_j$ by\begin{eqnarray*}
    \Pi_J=\sum_{j=-J}^J \ketbra{\phi_j}{\phi_j}.
\end{eqnarray*}
We write
\begin{eqnarray}
    \sigma^J_{\theta_0,n}&=&\frac{\Pi_J \sigma_{\theta_0,n}\Pi_J}{\mbox{Tr}(\Pi_J\sigma_{\theta_0,n})},\\
    {\hat P^J}&=&\Pi_J \hat P\Pi_J,\\
    U^J(\theta)&=& \Pi_J U(\theta)\Pi_J,
\end{eqnarray}
where $ \sigma^J_{\theta_0,n}\in\mathcal{S}(\mathcal{H}_{\leq J}), P^J(\hat{\theta})\in \mathcal{E}(\mathcal{H}_{\leq J}), U^J(\theta)\in\mathcal{U}(\mathcal{H}_{\leq J})$ and $U^J:{\rm SO}(2)\rightarrow \mathcal{U}(\mathcal{H}_{\leq J})$ defined by $\theta\mapsto U^J(\theta)$ is a representation of ${\rm SO}(2)$, because $U^J(\theta)={\rm diag}(e^{-iJ\theta},e^{-i(J-1)\theta},\ldots,e^{i(J-1)\theta},e^{iJ\theta})$. We denote
\begin{eqnarray}
    P^J_n(+|\theta_0+\theta')=\mbox{Tr}(\hat P^J(U^J)^\dagger(\theta')\sigma^J_{\theta_0,n} U^J(\theta')).
\end{eqnarray}
By observing that $U(\theta)\Pi_J=\Pi_JU(\theta)$, we find \begin{eqnarray}
 P^J_n(+|\theta_0+\theta')&=&\mbox{Tr}(\hat P^J(U^J)^\dagger(\theta')\sigma^J_{\theta_0,n} U^J(\theta'))\nonumber\\
 &=&\mbox{Tr}(\Pi_J \hat P\Pi_J^2U^\dagger(\theta')\Pi_J\sigma^J_{\theta_0,n}\Pi_J U(\theta')\Pi_J)\nonumber\\
 &=&\mbox{Tr}( \hat P \Pi_J U^\dagger(\theta')\sigma^J_{\theta_0,n}U(\theta')\Pi^2_J)\nonumber\\
&=&\mbox{Tr}( \hat P U^\dagger(\theta')\Pi_J\sigma^J_{\theta_0,n}\Pi_J U(\theta'))\nonumber\\
&=&\mbox{Tr}( \hat P U^\dagger(\theta')\sigma^J_{\theta_0,n}U(\theta')),
\end{eqnarray}
where we have used in the third line that the trace is cyclic.
\begin{lemma}
    Suppose that ${\rm Tr}(\Pi_J\sigma_{\theta_0,n})\geq 1-\epsilon$ then $\sqrt{\epsilon}\geq|P_n(+|\theta_0+\theta')-P^J_n(+|\theta_0+\theta')|$.
    \label{secondstep}
\end{lemma}
\begin{proof}
   The Gentle Measurement Lemma ~\cite{Wilde2017} states that if  $\mbox{Tr}(\Pi_J\sigma_{\theta_0,n})\geq 1-\epsilon$ then \begin{eqnarray*}
        \|\sigma_{\theta_0,n}-\sigma_{\theta_0,n}^J\|_1\leq 2\sqrt{\epsilon}
    \end{eqnarray*}
    holds. Furthermore, we will use Theorem 9.1 from ~\cite{Nielsen2010}, which states \begin{equation}
        D(\varrho,\sigma)=\max_{\{E_m\}}\tilde{D}(p_m,q_m),
    \end{equation} 
    where the maximization is over all POVMs $\{E_m\}$, $p_m=\mbox{Tr}(\varrho E_m)$) and $q_m=\mbox{Tr}(\sigma E_m)$.
We show \begin{eqnarray}
    \sqrt{\epsilon}&\geq& D(\sigma_{\theta_0,n},\sigma_{\theta_0,n}^J)\nonumber\\
    &=& \max_{E_m} \tilde{D}(\mbox{Tr}(\sigma_{\theta_0,n} E_m),\mbox{Tr}(\sigma_{\theta_0,n}^J E_m))\nonumber\\
    &\geq& \tilde{D}(P_{n}(\theta_0+\theta'),P_n^J(\theta_0+\theta'))\nonumber\\
    &=&|P_n(+|\theta_0+\theta')-P^J_n(+|\theta_0+\theta')|,
\end{eqnarray}
where $P^{J}_{n}(\theta_0+\theta')$ denotes the probability distribution $\{(P^{J}_{n}(+|\theta_0+\theta'),(P^{J}_{n}(-|\theta_0+\theta')=1-(P^{J}_{n}(+|\theta_0+\theta')\}$, and in the third line we have used that $\{P^+_{\theta'}=U(\theta')\hat P U^\dagger(\theta'),P^-_{\theta'}=\id-P^+_{\theta'}\}$ is a POVM.
\end{proof}
Let us check that $\Pi_J\to \id$ for $J\to \infty$ strongly. From the Peter-Weyl Theorem, we know that $\{\phi_j\}_{j=-\infty}^\infty$ defines an orthonormal basis of $L^2({\rm SO}(2))$ and hence
\begin{eqnarray*}
\left\|(\Pi_J-\mathbf{1})\psi\right\|^2&=&\left\| \sum_{j=-J}^J \langle \phi_j|\psi\rangle \phi_j - \psi\right\|^2\\
&=& \sum_{|j|>J} |\langle \phi_j|\psi\rangle|^2 = 1-\sum_{j=-J}^J |\langle\phi_j|\psi\rangle|^2\\
&\stackrel{J\to\infty}\longrightarrow& 0,
\end{eqnarray*}
which is true for every $\psi\in L^2({\rm SO}(2))$, and thus the claim follows. The last observation implies that we can make $\epsilon$ arbitrarily small by making $J$ larger and larger.

Everything said so far in this section can be easily generalized to more than two (say, $N$) measurement outcomes. Let us define an $N$-outcome rotational box as $N$ continuous non-negative real functions $P_k$ on the unit circle such that $\sum^N_{k=1} P_k(\theta)=1$ for every $\theta\in[0,2\pi)$. Similarly as above, we have associated operators $\hat P_k$, defining a POVM, and we can project those into the subspaces $\mathcal{H}_{\leq J}$ via $\hat P_k^J:=\Pi_J \hat P_k\Pi_J$. The approximating measurement on this spin-$J$ system will have POVM elements $\hat P_1^J,\ldots,\hat P_{N-1}^J,\mathbf{1}-\hat P_1^J-\ldots-\hat P_{N-1}^J$. Adaption of all further proof steps from above is straightforward and proves the analogous result for $N$-outcome rotation boxes.

\section{Proofs for Section~\ref{sec:rotationbell}}

\subsection{Proof of Theorem~\ref{TheQuantumCorrelations}}
\label{AppProofCorr}
First note that
\begin{eqnarray*}
P(-a,b|\alpha+\pi,\beta)&=& P_{b,\beta}^A(-a|\alpha+\pi)P^B(b|\beta)\\
&=&P_{b,\beta}^A(a|\alpha) P^B(b|\beta)=P(a,b|\alpha,\beta),
\end{eqnarray*}
and similarly, $P(a,-b|\alpha,\beta+\pi)=P(a,b|\alpha,\beta)$. Therefore, $P(a,b|\alpha,\beta)$ can be determined from the values of $a\cos\alpha$, $a\sin\alpha$, $b\cos\beta$ and $b\sin\beta$. In particular, we can find a function $f$, defined on two copies of the circle $\{(1,x,y)\,\,|\,\,x^2+y^2=1\}$, such that
\[
P(a,b|\alpha,\beta)=f\left(\strut (1,a\cos\alpha,a\sin\alpha),(1,b\cos\beta,b\sin\beta)\right).
\]
Let $a=1$ and $\alpha_1=0$, $\alpha_2=\pi/4$ and $\alpha_4:=\pi/2$, and $e_i:=(1,a\cos\alpha_i,a\sin\alpha_i)$ for $i=1,2,3$, then these three vectors are linearly independent and span $\mathbb{R}^3$. Let $g:\mathbb{R}^3\times\mathbb{R}^3\to\mathbb{R}$ be the bilinear form that satisfies $g(e_i,e_j)=f(e_i,e_j)$ for $i,j=1,2,3$.

Now suppose we fix some value of $b$ and of $\beta$, then
\begin{eqnarray*}
P(a,b|\alpha,\beta)&=&P_{b,\beta}^A(a|\alpha)P^B(b|\beta)\\
&=&\left(\frac 1 2+c_1 a \cos\alpha + s_1 a\sin\alpha\right)P^B(b|\beta),
\end{eqnarray*}
where $c_1$ and $s_1$ may depend on $b$ and $\beta$. For every fixed $b$ and $\beta$, this is a linear functional of the vector $(1,a\cos\alpha,a\sin\alpha)$. Similar argumentation applies to the roles of $A$ and $B$ exchanged. Thus, $f$ and $g$ must agree on $f$'s domain of definition, and so
\[
P(a,b|\alpha,\beta)=g\left(\strut (1,a\cos\alpha,a\sin\alpha),(1,b\cos\beta,b\sin\beta)\right).
\]
Now, every $2\times 2$ Hermitian matrix $M\in\LH(\comp^2)$ can be parameterized in the form
\[
M=\frac 1 2\left(\begin{array}{cc}r_0+r_3 & r_1-ir_2 \\r_1+i r_2 & r_0-r_3\end{array}\right)
\]
(for $r_0=1$, this is the well-known Bloch representation of quantum states). Define the linear map $\vec{r}(M):=(r_0,r_1,r_2)$, dropping the $r_3$-component. Finally, define the bilinear form $\omega:\LH(\comp^2)\times\LH(\comp^2)\to\mathbb{R}$ via
\[
\omega(M,N):=g\left(\vec{r}(M),\vec{r}(N)\right).
\]
This bilinear form is unital:
\begin{eqnarray*}
\omega(\mathbf{1},\mathbf{1})&=& g\left(\strut (2,0,0),(2,0,0)\right)\\
&=&g\left(\strut (1,1,0),(2,0,0)\right)+g\left(\strut (1,-1,0),(2,0,0)\right)\\
&=& g\left(\strut (1,1,0),(1,1,0)\right)+g\left(\strut (1,1,0),(1,-1,0)\right)\\
&&+ g\left(\strut (1,-1,0),(1,1,0)\right)+g\left(\strut (1,-1,0),(1,-1,0)\right)\\
&=& P(+1,+1|0,0)+P(+1,-1|0,0)\\
&&+P(-1,+1|0,0)+P(-1,-1|0,0)=1.
\end{eqnarray*}
Let us now show that $\omega(M,N)\geq 0$ if $M$ and $N$ are positive semidefinite. If $M\geq 0$, then $r_0={\rm Tr}(M)\geq 0$, and non-negativity of the eigenvalues enforces $r_1^2+r_2^2+r_3^2\leq r_0^2$, hence $r_1^2+r_2^2\leq r_0^2$. Hence $\vec r(M)$ lies in the disc of radius $r_0$, and can thus be written as a convex combination of points on the circle. Since $g$ is bilinear, this will give the corresponding convex combination of values, and it is thus sufficient to restrict our attention to the case that $r_1^2+r_2^2=r_0^2$. In this case, there will be some angle $\alpha$ such that $(r_1,r_2)=(r_0\cos\alpha,r_0\sin\alpha)$. Similar reasoning for the matrix $N\geq 0$ (denoting the first component of $\vec{r}(N)$ by $s_0$) yields
\begin{eqnarray*}
\omega(M,N)&=&g\left(\strut (r_0,r_0\cos\alpha,r_0\sin\alpha),(s_0,s_0\cos\beta,s_0\sin\beta)\right)\\
&=&r_0 s_0 g\left(\strut (1,\cos\alpha,\sin\alpha),(1,\cos\beta,\sin\beta)\right)\\
&=& r_0 s_0 P(+1,+1|\alpha,\beta)\geq 0.
\end{eqnarray*}
Set $M_\pm(\alpha):=e^{-i\alpha Z}|\pm\rangle\langle \pm|e^{i\alpha Z}$, where $|\pm\rangle:=\frac 1 {\sqrt{2}}(|0\rangle\pm|1\rangle)$, and similarly for $M_\pm(\beta)$, then
\begin{eqnarray*}
\omega(M_a(\alpha),N_b(\beta))&=&g\left(\strut (1,a\cos\alpha,a\sin\alpha),(1,b\cos\beta,b\sin\beta)\right)\\
&=& P(a,b|\alpha,\beta).
\end{eqnarray*}
It follows from the results of Barnum et al.~\cite{Barnum} (see also Ac\'in et al.\cite{Acin} for a simplified proof, and Kleinmann et al.~\cite{Kleinmann}) that there is a quantum state $\rho_{AB}$ on the two qubits and a positive unital linear map $\tau:\LH(\comp^2)\to\LH(\comp^2)$ such that
\[
\omega(M,N)={\rm Tr}\left(\rho_{AB} M\otimes \tau(N)\right).
\]
This completes the proof.\qed

\section{Proofs for Section~\ref{sec:connectionothertopics}: Connections to other topics}

\subsection{Background on transitive GPTs}

We briefly introduce some necessary background on transitive GPT systems and refer the reader to~\cite{galley2021dynamics} for a more complete introduction.

A finite-dimensional transitive GPT system is one with pure states $X$ and dynamical group $G$ which is compact (this includes the possibility of finite groups). The space of pure states $X$ is isomorphic to $G/H$ with $H$ the stabilizer subgroup.

To each transitive GPT system is associated a representation of $G$ which we denote $\rho: G \to \GL(V)$. Let us denote its decomposition into irreps by
\begin{align} 
    V &\simeq \bigoplus_{i \in \mathcal I} V_i, \\
    \rho(g) &\simeq \bigoplus_{i \in \mathcal I} \rho_i(g),
\end{align}
where $\mathcal I$ may contain repeated entries.  

By transitivity, the state space can be obtained by applying the representation $\rho(g)$ to a reference pure state $\omega_x \in V$:
\begin{align}
    \omega_{gx} = \rho(g) \omega_x, 
\end{align}
which is necessarily invariant under $\rho(h)$ for $h \in H_x$, the stabilizer of $x$. The state $\omega_x$ has support in every irrep $V_i$ for $i \in \mathcal I$ (this is in fact not an assumption but follows from what it means for the representations $\rho(g)$ to be associated to the system).

It follows from Theorem 2 of~\cite{galley2021dynamics} that when $(G,H)$ form a Gelfand pair, any two transitive GPT systems with associated representations $\mathcal I$ and $\mathcal J$ which are equal as sets (i.e.\ contain the same irreducible representations ignoring repetitions) are equivalent as GPT systems (assuming that they are effect unrestricted). Two GPT systems $(\Omega_1, E_1, \Gamma_1)$ and $(\Omega_2, E_2, \Gamma_2)$ with associated vector spaces $V_1$ and $V_2$ and with dynamical group $G$ are equivalent if there exists an invertible linear transformation $L:V_1\to V_2$ relating them:
\begin{align}
    L(\Omega_1)=\Omega_2, \\
    E_1\circ L^{-1}=E_2, \\
    L\Gamma_1 L^{-1} = \Gamma_2.
\end{align}

\subsection{Proof of Theorem~\ref{thm:symmetric_entanglement}}\label{app:symmetric_entanglement}

We will make use of the following lemma:
\begin{lemma}
    The symmetric product states $\ketbra{\psi}{\psi}^{\otimes d} \in \mD(\Sym^d(\reals^2))$ have full support in $\Sym^{2d}(\reals^2)$ and therefore have support in one copy of every irrep in  $\{0,...,d\}$. 
\end{lemma}
\begin{proof}
The rebit pure states $\ket \psi$ transform under the real projective irreducible representation $\frac{1}{2}$ of $\SO(2)$; a generic rebit state can be written as: 
 \begin{align}
   \ket{\psi}
   &= \cos\frac{\theta}{2} \ket{+} +  \sin\frac{\theta}{2} \ket{-} .
   \end{align}
If we complexify the vector space this is equal to:
    \begin{align}
      \ket{\psi} = \frac{1}{\sqrt{2}}( e^{i \frac{\theta}{2}} \ket{0} +  e^{-i \frac{\theta}{2}}  \ket{1} ) . 
\end{align}
Hence the symmetric product states of $d$ rebits $\ketbra{\psi}{\psi}^{\otimes d} \in \mD(\Sym^d(\reals^2))$ are isomorphic (there exists an equivariant invertible linear map) to the product states $\ketbra{\psi(\theta)}{\psi(\theta)}^{\otimes d}\in \mD(\Sym^d(\comp^2))$ where $\ket{\psi(\theta)} = \frac{1}{\sqrt{2}} (e^{i \frac{\theta}{2}} \ket 0 +e^{-i \frac{\theta}{2}} \ket 1) $.
    Using the isomorphism $\mL(\Sym^d(\comp^2)) \simeq \Sym^d(\comp^2) \otimes \overline{\Sym^d(\comp^2)}$ and  $\Sym^d(\comp^2)) \simeq \overline{\Sym^d(\comp^2)}$ (by \Cref{cor:self_dual}), we have  the following isomorphism: $\mL(\Sym^d(\comp^2)) \simeq \Sym^d(\comp^2) \otimes \Sym^d(\comp^2)$,
\begin{align}
    \ketbra{\psi(\theta)}{\psi(\theta)}^{\otimes d} \mapsto \ket{\psi(\theta)}^{\otimes 2d}.
\end{align}
Expanding $\ket{\psi(\theta)}^{\otimes 2d}$ gives:
\begin{align}
    \ket{\psi(\theta)}^{\otimes 2 d} &= \frac{1}{2^d} (e^{i d\theta} \ket{0}^{\otimes 2d} \nonumber \\
    +& e^{i (d -1)\theta} (\ket{0}^{\otimes 2d -1} \ket{1} + \ket{0}^{\otimes 2d-2} \ket 1 \ket 0  \nonumber \\
    +& ... + \ket 1 \ket{0}^{\otimes 2d -1}) \nonumber \\
    + &... + e^{- i d \theta} \ket{1}^{\otimes 2d} ) \nonumber\\
    &= \left(\frac{1}{\sqrt{2}}\right)^{2d} \sum_{j = 0}^{2 d} e^{i (d - j)\theta} P_{\Sym} \ket{0}^{2d - j} \ket{1}^j \\
     &= \frac{1}{2^d} \sum_{j = 0}^{2d} \sum_{x \in \{0,1\}^{2d}| H(x)= j} e^{i (d - j)\theta}  \ket{x} , 
\end{align}
where $P_{\Sym} \ket{v_1}_1 \otimes ... \otimes \ket{v_d}_d = \sum_{\sigma \in \Sigma_d} \ket{v_1}_{\sigma^{-1}(1)} \otimes ... \otimes \ket{v_d}_{\sigma^{-1}(d)}  $, $\Sigma_d$ the symmetric group on $d$ elements and $H(x)$ the Hamming weight of the bit string $x$.

Each $\ket{k = d-j} = \sum_{x \in \{0,1\}^{2d}| H(x)= j}   \ket{x}$ belongs to a subspace carrying a projective representation $k$ of $\SO(2)$. Thus, $\ket{\psi(\theta)}^{\otimes 2d}$ has support on a copy of every complex irrep $\{d, - d + 1 ,..., d\}$. The projection of $\ket{\psi(\theta)}^{\otimes 2d}$ on the subspace carrying the representation $\{-k,k\}$ is:
\begin{align}
    &\frac{1}{2^d} (e^{i k \theta} \ket k + e^{-i k \theta} \ket {-k}) \\ 
    =& \frac{1}{2^d}  \left(\cos(k \theta) \frac{\ket{k}  + \ket{-k} }{2} +  \sin(k \theta) \frac{\ket{k}  + \ket{-k}}{2} \right),
\end{align}
implying that  $\ket{\psi(\theta)}^{\otimes 2d}$ it has support in every real irrep $\{0,...,d\}$. This implies $ \ketbra{\psi(\theta)}{\psi(\theta)}^{\otimes d}$ has suport in every real irrep $\{0,...,d\}$, and so $\ketbra{\psi}{\psi}^{\otimes d}$ does also.
\end{proof}

Hence,
\begin{align}
    \Omega_{\Sym}^{d} := \conv\{(\ketbra{\psi}{\psi})^{\otimes d}\,\,|\,\, \ket \psi \in P_{\reals^2} \},
\end{align}
where $ P_{\reals^2}$ is the set of rebit pure states, is the state space of a transitive GPT system with pure states $\SO(2)$ and dynamical group $\SO(2)$. It has associated to it the real representation $\{0,..,d\}$. The stabilizer subgroup is just the trivial group $\mathbb I=\{\mathbf{1}\}$.

The set of unrestricted effects on $\Omega_{\Sym}^{d}$ is given by
\begin{align}
    E_\Sym^d:= \{E|E \in \LS(\Sym^d(\reals^2), \nonumber \\
    0 \leq \Tr(E \ketbra{\psi}{\psi}^{\otimes d}) \leq 1 \}.
\end{align}
Since $(\SO(2), \mathbb{I})$ forms a Gelfand pair, it follows from~\cite[Theorem 2 (iii)]{galley2021dynamics} that all unrestricted GPT systems generated by applying a real representation $\{0,..,d\}$ of $\SO(2)$ to a reference vector with support in each irrep are equivalent. 

Hence the GPT systems $(\Omega_{\Sym}^d, E_\Sym^d)$ and $(\Omega_{\frac{d}{2}},E_{\frac{d}{2}})$ are equivalent as GPT systems and generate the same $\SO(2)$ correlations.

\end{document}